\documentclass[superscriptaddress,prapplied,aps,onecolumn,10pt,longbibliography]{revtex4-2}

\makeatletter     
\usepackage{times}
\usepackage[font=small, labelfont=bf]{caption}
\usepackage{amsmath,amsthm}   
\usepackage{amssymb}
\usepackage{lastpage}
\usepackage[pdftex]{graphicx}   
\usepackage[dvipsnames]{xcolor}
\definecolor{ao(english)}{rgb}{0.0, 0.5, 0.0}
\definecolor{raspberrypink}{rgb}{0.89, 0.31, 0.61}
\usepackage{epstopdf}
\usepackage{tikz}
\usepackage{enumerate}
\usepackage{upgreek}
\usepackage{braket}
\usepackage{physics}
\usepackage{verbatim}
\usepackage{url}

\usepackage{setspace}
\SetSinglespace{1.2}
\usepackage{hyperref}
\hypersetup{
    colorlinks=true,       
    linkcolor=cyan,          
    citecolor=magenta,        
    filecolor=magenta,      
    urlcolor=cyan,           
    runcolor=cyan
}

\usepackage[normalem]{ulem}

\newtheorem{theorem}{Theorem}
\numberwithin{theorem}{section}
\numberwithin{equation}{section}

\newtheorem{conjecture}[theorem]{Conjecture}
\newtheorem{lemma}[theorem]{Lemma}
\newtheorem{corollary}[theorem]{Corollary}
\newtheorem{proposition}[theorem]{Proposition}

\newtheorem{remark}[theorem]{Remark}

\newtheorem{protocol}[theorem]{Protocol}

\newcommand{\mc}[1]{\mathcal{#1}}

\newcommand{\cre}{a^\dagger}

\newcommand{\vac}{\ket{\vec{0}}}
\newcommand{\bvac}{\bra{\vec{0}}}
\newcommand{\kbvac}{\ketbra{\vec{0}}}

\newcommand{\h}[1]{\mathcal{H}_{\textnormal{#1}}}
\newcommand{\hn}{\mathcal{H}^{\otimes n}}
\newcommand{\hne}{\h{ext}^{\otimes n}}

\newcommand{\ole}{\overline{e}}

\newcommand{\perm}{\textnormal{perm}}

\makeatletter
\def\l@subsubsection#1#2{}
\makeatother

\begin{document}
\title{General protocols for the efficient distillation of indistinguishable photons}

\author{Jason Saied}\thanks{jason.saied@nasa.gov}
\affiliation{QuAIL, NASA Ames Research Center, Moffett Field, CA 94035, USA}

\author{Jeffrey Marshall}
\affiliation{QuAIL, NASA Ames Research Center, Moffett Field, CA 94035, USA}
\affiliation{USRA Research Institute for Advanced Computer Science, Mountain View, CA 94043, USA}

\author{Namit Anand} 
\affiliation{QuAIL, NASA Ames Research Center, Moffett Field, CA 94035, USA}
\affiliation{KBR, Inc., 601 Jefferson St., Houston, TX 77002, USA}

\author{Eleanor G. Rieffel}
\affiliation{QuAIL, NASA Ames Research Center, Moffett Field, CA 94035, USA}

\begin{abstract} 
We introduce state-of-the-art protocols to distill indistinguishable photons, reducing distinguishability error rates by a factor of \(n\), while using a modest amount of resources scaling only linearly in \(n\). 
Our resource requirements are significantly lower and the protocols have fewer hardware requirements than in previous works, making large-scale distillation experimentally feasible for the first time. 
This efficient reduction of distinguishability error rates has direct applications to fault-tolerant linear optical quantum computation, potentially leading to improved thresholds for photon loss errors and allowing smaller code distances, thus reducing overall resource costs. 
Our protocols are based on Fourier transforms on finite abelian groups, special cases of which include the discrete Fourier transform and Hadamard matrices. 
This general perspective allows us to unify previous results on distillation protocols and introduce a large family of efficient schemes. 
We utilize the rich mathematical structure of Fourier transforms, including symmetries and related suppression laws, to quantify the performance of these distillation protocols both analytically and numerically.
Finally, our work resolves an open question concerning suppression laws for the $n$-photon discrete Fourier transform: the suppression laws are exactly characterized by the well-known Zero Transmission Law if and only if $n$ is a prime power.
\end{abstract}

\maketitle

\tableofcontents

\section{Introduction}

Highly pure and indistinguishable photons are a prerequisite for use in quantum information processing. 
The effect of distinguishability has been famously demonstrated by the Hong-Ou-Mandel (HOM) experiment, which shows fundamentally different statistics in the cases when photons are identical or not.
The HOM effect (and its generalizations) is a crucial ingredient for realizing linear optical quantum computation (QC) \cite{knill_optimal_2007}, where the interference between identical photons can be used to create entanglement over computational degrees of freedom (aided by post-selective measurements). For example, fusion measurements can be used to create large cluster states out of primitive entangled states, such as Bell states or small Greenberger-Horne–Zeilinger (GHZ) states \cite{browne_resource-efficient_2005}. These states can then be used to realize fault-tolerant quantum computation, in paradigms such as fusion-based QC \cite{bartolucci_fusion-based_2023-1}.
However, the presence of distinguishability will generally result in less entanglement generated over the computational degrees of freedom 
\cite{sparrow_quantum_2017, shaw_errors_2023}, meaning their computational `resource' is reduced. 
In fact, the degree of distinguishability has been shown to reduce the classical computational complexity (and thus potential quantum advantage) of boson sampling \cite{renema_efficient_2018}. 
Distinguishability between photons can also be unheralded (since the correct number of photons are eventually detected), making such errors potentially harder to deal with in linear optical QC than
photon loss, the dominant source of noise in linear optics, which is generally heralded. 
We sketch the potential impact of distinguishability errors on fusion-based QC in Section~\ref{sec:error correction}, where we argue that reducing distinguishability may greatly improve error thresholds for photon loss and reduce the overall resource requirements of the fault-tolerant quantum computation. 
In particular, due to the benefit of significantly reduced distinguishability error rates and the efficiency of the distillation protocols we present in this work, we argue in the appendix that our protocols may be worthwhile as a complement (or cheaper alternative) to state-of-the-art single photon sources to enable more scalable and resource-efficient fault-tolerant linear optical quantum computation. 

The conventional approach to achieving highly pure photons relies upon spectral filtering. While spectral filtering can produce very uniform photons, it has drawbacks. First, such methods are typically unheralded, meaning if a single photon impinges upon a filter, it is unknown if it exited the filter. Second, as the target fidelity is increased, the probability a photon passes the filter decreases \cite{mosley_heralded_2008}. To overcome these issues, in 2017 a fully linear optical scheme (with Fock basis measurements) to distill indistinguishable photons from partly distinguishable ones was presented \cite{sparrow_quantum_2017}, which was subsequently improved upon in Ref.~\cite{marshall_distillation_2022}. 
An experimental demonstration of the distillation scheme of Ref.~\cite{sparrow_quantum_2017} was recently given in Ref.~\cite{faurby2024purifying}.

Current experiments typically measure the degree of photon-photon distinguishability in linear optics via a HOM dip which counts coincidence events, allowing one to compute the \emph{visibility} $V = \mathrm{Tr}[\rho_a \rho_b]$, where e.g.~`$a$', `$b$' label two different photon sources. Values for the visibility vary quite a bit, depending on the type of single photon source used, photon encoding, spectral filtering, etc., with an approximate range found in the literature to be $V\in [0.74, 0.99]$ \cite{halder_high_2008, tsujimoto_high_2017,tambasco_quantum_2018, wang_research_2019, ollivier_hong-ou-mandel_2021, somhorst_quantum_2023, alexander_manufacturable_2024}. In this work we characterize the error via a parameter $\epsilon$ we call the \emph{distinguishability error rate}, discussed in Sect.~\ref{sec:models}. We have $\epsilon\approx 1-\sqrt{V}$, which for reference gives $\epsilon$ approximately in the range $[0.005, 0.15]$, based on the quoted visibility values. 

The general idea of these distillation protocols is shown in Fig.~\ref{fig:dist-cartoon}, whereby $n$ (noisy) single photons are evolved under a linear optical unitary, with post-selection performed by photon-number-resolving detectors (PNRD). It can be shown that by carefully performing the post-selection based on the interference generated by the unitary evolution, one can guarantee that the output photons will have a higher expected overlap with the target state.
In this work we generalize the state-of-the-art $3$ and $4$ photon protocols of Ref. \cite{marshall_distillation_2022} to allow the scheme to work with any number $n$ of photons (as in Fig.~\ref{fig:dist-cartoon}). 
We prove in Theorem~\ref{thm:error rate} that for input photons with distinguishability error rate $\epsilon$, 
our $n$-photon distillation protocols herald the output of single photons with distinguishability error rate $e_n(\epsilon)\approx \epsilon/n$. 
In particular, large enough values of $n$ lead to arbitrary reduction of the distinguishability error rate. 
The $n$-photon protocols use approximately $4n$ photons to obtain this reduction (see Theorem~\ref{thm:resources}). 
As discussed in more detail below and in Remark~\ref{remark:efficiency}, the previous state of the art \cite{marshall_distillation_2022} requires $O(n^2)$ photons to achieve the same amount of error reduction, implying (for example) that our protocols are over $20$ times more efficient when $n=81$. 
Further, in order to reduce the error rate by more than a factor of $4$ (as in Marshall's 4-photon protocol \cite{marshall_distillation_2022}) \emph{without} exponentially scaling resource costs, 
the previous techniques require iteration of several smaller distillation protocols, 
along with active feed-forward and quantum memory (e.g., optical delay lines) so that successfully distilled photons can be used in later rounds of distillation. The additional hardware requirements of these previous protocols may ultimately \emph{increase} the error rate: for example, the delay lines may cause the photons to experience different amounts of dispersion. 
For the main protocols presented in this work, however, there is no need for iteration, feed-forward, or quantum memory during distillation. 
Thus the new distillation protocols are significantly more resource-efficient and require less complex hardware, making them far more feasible for experimental implementation. 

We now briefly discuss our general framework for distillation and its utility. 
The $n$-photon protocols depend on an $n\times n$ unitary matrix $U$, determining an $n$-mode linear optical unitary. 
We consider protocols with $U$ coming from the discrete Fourier transform $F_n$ (for $n\ge 3$) and Hadamard matrices $H_n=H^{\otimes r}$, where $H$ is the $2\times 2$ Hadamard matrix and $n=2^r\geq 4$. 
We note that the previous state-of-the-art $3$ and $4$ photon protocols of \cite{marshall_distillation_2022}, corresponding to $U=F_3$ and $U=H_4$ respectively, were discovered numerically. 
Ref. \cite{marshall_distillation_2022} gives rigorous proofs regarding the performance of these protocols, but these proofs are by direct computation: there was no theoretical understanding of why those unitaries would be particularly useful for distillation, nor any indication of how one might generalize them. 
In this work, we put both protocols into a larger theoretical context, generalizing them to the families $F_n$ and $H_n$. 
We further unify these families by observing that both correspond to Fourier transforms on finite abelian groups of order $n$ ($\mathbb{Z}_n$ for Fourier and $\mathbb{Z}_2^{\times r}$ for Hadamard), and in fact any such Fourier transform (with $n\geq 3$) leads to a similarly efficient distillation protocol. 
(The main text focuses on the $F_n$ and $H_n$ cases, with the generalization discussed in detail in Appendix~\ref{sec:fourier abelian}.) 
Questions about distillation then become questions about suppression laws for general Fourier transforms. 
This generality allows us to combine powerful results on suppression laws in the Fourier \cite{tichy2010zero} and Hadamard \cite{crespi2015suppression} settings into a single framework and apply it to characterize the performance of our distillation protocols. 
In particular, as discussed further below, suppression laws are often determined by symmetry properties of the relevant unitary \cite{dittel2018totally}; for the Fourier transform on the finite abelian group $G$, the relevant group of symmetries is simply $G$. 
This is the perspective motivating the proofs and conjectures in this work. 
We also note that, due to the theoretical connections we establish between distillation and suppression laws, our results have important consequences beyond distillation. In particular, we resolve an open problem due to the 2010 paper Ref. \cite{tichy2010zero} by proving precisely when the known suppression laws for $F_n$ are both necessary and sufficient conditions for an output pattern to be suppressed (see Theorem~\ref{thm:suppression body}). We return to this subject below. 

As with the earlier works \cite{sparrow_quantum_2017, marshall_distillation_2022}, the new $n$-photon distillation protocols are non-deterministic with some \emph{heralding rate} $h_n(\epsilon)<1$, the probability of heralding the output of a distilled single photon. 
Large heralding rates are desirable in order to minimize overhead. 
In Theorem~\ref{thm:herald first order}, we give a formula for the heralding rate $h_n(\epsilon)$ up to first order in $\epsilon$, depending only on the heralding rate in the error-free case, $h_n(0)$. 
In Theorem~\ref{thm:herald ideal} and the following discussion, we give a formula for $h_n(0)$ in terms of hypergeometric functions and observe that it quickly approaches $1/4$, with $h_n(0) - \frac{1}{4}\sim \frac{1}{16n}$.
In particular, for small error rates $\epsilon$ (and $n\geq 4$), the $n$-photon distillation protocol succeeds with probability near $1/4$, and therefore we expect to use $4n$ input photons to obtain an output photon with reduced error rate approximately $\epsilon/n$. 
This is a significant improvement over the protocol of Sparrow and Birchall \cite{sparrow_quantum_2017}, which requires exponentially many photons to achieve the same reduced error rate $\epsilon/n$ \cite{marshall_distillation_2022}. 
The protocols of \cite{sparrow_quantum_2017} and \cite{marshall_distillation_2022} may be iterated, with successfully distilled photons stored and used as input to later rounds of distillation; 
this iterative approach allows for 
similar error rates with cubic or quadratic resources respectively, but requires active feed-forward and quantum memory (optical delay lines) during distillation. 
Our protocols, on the other hand, use only linearly scaling resources \emph{without} requiring active feed-forward or memory during distillation. 
In Section~\ref{sec:iterated}, we discuss the resource requirements of the protocols in \cite{sparrow_quantum_2017, marshall_distillation_2022} and the present work in more detail. 
Generally, we see that for small $\epsilon$, the new protocols 
achieve comparable reduction in error at significantly lower cost. 
We give a concrete example in Remark~\ref{remark:efficiency}: to reduce the error rate $\epsilon$ by a factor of $n=81$, the new protocols require $20$ times fewer photons than the previous state-of-the-art protocol of \cite{marshall_distillation_2022}. 
In particular, the new $F_{81}$ protocol would require an average of only $4$ runs using $81$ photons each, with no interaction between subsequent runs. 
We note, however, that these resource estimates require $\epsilon$ to be sufficiently small. Each protocol has a \emph{threshold} value of $\epsilon$ above which it fails to reduce the error rate; as $n$ increases, this threshold decreases. (See Figure~\ref{fig:epsilon_critical}.) 
Further, even below threshold, it can be advantageous to concatenate multiple distillation protocols, beginning with smaller values of $n$ and feeding the less distinguishable output photons into larger $n$ protocols in subsequent stages. We discuss such iterated distillation schemes in Section~\ref{sec:iterated}. 

For larger values of $\epsilon$, the first-order approximations of $h_n(\epsilon)$ discussed above no longer suffice. 
In Section~\ref{sec:symmetry body} we conjecture a lower bound on $h_n(\epsilon)$ as a function of $\epsilon$. 
This lower bound may be used to give explicit guarantees on the heralding rate for large $n$, as long as $\epsilon$ is sufficiently small relative to $n$. For example, if $\epsilon\leq 1/n$, we conjecture that the $n$-photon protocol's heralding rate is bounded below by $\frac{1}{4e}\approx 0.092$. 
These conjectures are supported by numerical evidence 
and are motivated by certain symmetry properties of the distillation protocols, discussed in Sect.~\ref{sec:symmetry body}. 
As discussed above, these symmetry properties are related to the suppression laws of \cite{tichy_sampling_2015, crespi2015suppression, dittel2018totally} that underlie our main results and motivate the use of the Fourier and Hadamard unitaries. 
These suppression laws have also been applied to the related problem of verifying boson sampling devices \cite{tichy2014stringent, viggianiello2018experimental}. 
For the Hadamard case, the suppression is entirely governed by certain parity-preservation conditions determined by the symmetries \cite{crespi2015suppression, dittel2018totally}. 
For the Fourier transform, the corresponding \emph{Zero Transmission Law} (ZTL) was only known to determine some of the suppression, with cases (the smallest being $F_6$) in which there are more suppressed patterns than dictated by the ZTL \cite{tichy2010zero}. 
In Theorem~\ref{thm:suppression body}, we prove that the suppressed patterns are exactly determined by the ZTL if and only if $n$ is a prime power. 

In Sect.~\ref{sec:numerics}, we give some numerics of interest. 
We provide detailed analysis of the case $n=8$ as a representative example. This is the smallest new protocol with $n$ a power of $2$, which means that we may consider both Fourier and Hadamard distillation protocols on the same number of photons and compare their performance. 
We also study the following problem: given an approximately known input error rate $\epsilon$, which distillation protocol gives the greatest reduction in error if we are free to choose the value of $n$? 
This problem is numerically answered for $n\leq 16$ in Figures~\ref{fig:opt_n} and \ref{fig:opt_n_both}. 
We observe several interesting trends in these results, related to the symmetry properties of Section~\ref{sec:symmetry body}. 
In particular, the optimal protocols seem to be Fourier transform protocols in which there are \emph{more} suppressed patterns than explained by the ZTL (i.e., by Theorem~\ref{thm:suppression body}, with $n$ not a prime power). 

In Sect.~\ref{sec:loss}, we consider the effect of photon loss on our protocols. 
Since photon loss will be detected by a distillation protocol with high probability (especially for large $n$), 
we are most concerned with the effect on the heralding rates (and thus the resource requirements). 
We give numerics and prove lower bounds on the lossy heralding rates as a function of $h_n(\epsilon)$, thus obtaining upper bounds on the required resources. 
In the $\epsilon = 0$ limit, the lower bound is exactly attained. 

We further demonstrate in App.~\ref{sec:haar} that the ability to distill photons is in fact typical, as shown by a randomized scheme based on Haar sampled unitaries. However, these versions are significantly less 
resource efficient (requiring a far greater number of photons) compared to those based on the Fourier or Hadamard matrices, due to a greater degree of constructive interference 
as a result of symmetries present in these matrices (as discussed in Sect.~\ref{sec:symmetry body}).

While readers familiar with linear optics should be able to skip most of the preliminaries in Section~\ref{sec:preliminaries}, we recommend reviewing the notation used in Section~\ref{sec:models}, especially the error models and \eqref{eq:error decomp}. 

We also note that a paper with similar results on distillation protocols, restricting to the case $F_n$ \cite{somhorst2024photon}, was released simultaneously with the first version of the present work. 

\begin{figure}[!t]
    \centering
    \includegraphics[width=0.7\columnwidth]{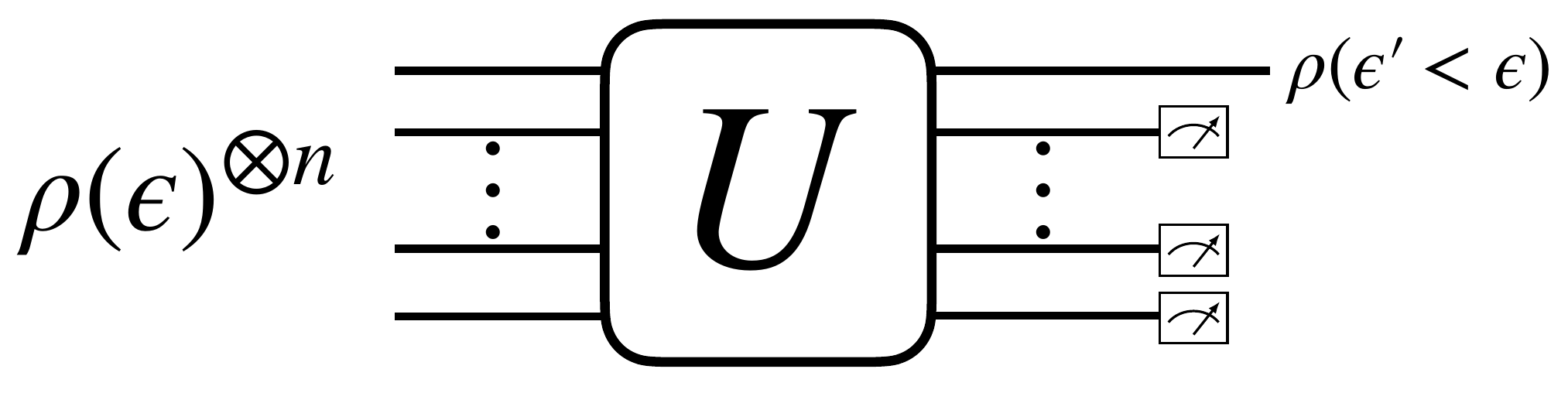}
    \caption{Cartoon of $n$-photon distillation scheme, where $n$ photons are injected into a linear optical unitary $U$. PNRDs are used to post-select on certain $n-1$ photon outcomes, resulting in a single photon out (in the absence of losses, dark counts, etc.). It is possible to engineer the unitary and post-selection criteria so that the output single-photon error is reduced.}
    \label{fig:dist-cartoon}
\end{figure}

\subsection{Summary of contributions}

We now briefly summarize the main results of this work. For each positive integer $n\geq 3$, we introduce several $n$-photon distillation protocols of the form in Figure~\ref{fig:dist-cartoon}, with the unitaries $U$ corresponding to Fourier transforms on finite abelian groups $G$ of order $n$. 
We assume that we have distinguishability error rate $\epsilon$ (defined in Section~\ref{sec:models} below). 
We then prove the following results: 

\begin{enumerate}
    \item (Theorem~\ref{thm:herald first order}) Upon successful heralding, the $n$-photon protocols output a single photon with reduced distinguishability error rate 
    \begin{equation*}
        e_n(\epsilon) = \epsilon/n + O(\epsilon^2). 
    \end{equation*}
    \item (Theorem~\ref{thm:error rate}) The probability of successful heralding, called the \emph{heralding rate} $h_n(\epsilon)$, may be expressed in terms of the error-free heralding rate $h_n(0)$ as follows: 
    \begin{equation*}
        h_n(\epsilon) = h_n(0) - (n-1)h_n(0)\epsilon + O(\epsilon^2).
    \end{equation*}
    For $n\geq 4$, we have $h_n(0)\approx 1/4$, and therefore for small $\epsilon$, we obtain heralding rate
    \begin{equation*}
        h_n(\epsilon) \approx \dfrac{1}{4} - \dfrac{n-1}{4}\epsilon.
    \end{equation*}
    \item (Theorem~\ref{thm:resources}) The expected number of photons required to distill a single output photon using an $n$-photon protocol is $n/h_n(\epsilon)$. For $n\geq 4$ and small $\epsilon$, the number of photons required is then approximately $4n$. This is in contrast to the previous state-of-the-art protocols \cite{marshall_distillation_2022}, for which the cost is $O(n^2)$ (see Remark~\ref{remark:efficiency}). 
    \item (Theorem~\ref{thm:loss}) If each beamsplitter in the linear optical circuit has loss rate $\lambda$, then the lossy heralding rate is 
    \begin{equation*}
        h_n(\epsilon; \lambda)\geq (1-\lambda)^{(n-1)\log n} h_n(\epsilon).
    \end{equation*}
    The lower bound becomes an equality when $\epsilon=0$. In this setting, the number of photons required for successful distillation is, on average, approximately
    \begin{equation}
        \dfrac{4n}{(1-\lambda)^{(n-1)\log n}},
    \end{equation}
    as shown in \eqref{eq:lossy resource estimate}. 
    For $\lambda = 0.01$, $n=16$, distillation to reduce the error rate by a factor of $16$ then requires an average of $7.3$ runs of $16$ photons at a time, with a total expected cost of approximately $117$ photons. 
    We note that this outperforms even the \emph{lossless} version of the previous state-of-the-art protocols \cite{marshall_distillation_2022}, which would require approximately $256$ photons to achieve a similar reduction in error. 
    \item (Eq. \eqref{eq:lossy fidelity}) Let $\Lambda := 1- (1-\lambda)^{\log n}$ be the average loss probability per photon. Given successful heralding, the probability of obtaining a single ideal photon in the output mode (i.e., no loss and no distinguishability errors) is 
    \begin{equation*}
        (1-\Lambda)(1-e_n(\epsilon)) + O(n\Lambda\epsilon),
    \end{equation*}
    where $e_n(\epsilon)$ is the error rate in the lossless case as above. 
    In other words, up to \emph{second-order} corrections involving both a photon loss and a distinguishability error, the output fidelity is simply the product of the output fidelity in the lossless case and the probability that the output photon is not lost. 
    \item (Theorem~\ref{thm:suppression body}) We prove that for the discrete Fourier transform $F_n$, the suppression laws are precisely characterized by the Zero Transmission Law of Ref. \cite{tichy2010zero} if and only if $n$ is a prime power. 
    The connection between suppression laws and the performance of the corresponding distillation protocols is discussed in Section~\ref{sec:symmetry body} and Appendix~\ref{sec:symmetry}. 
    \item (Section~\ref{sec:numerics}) Given a distinguishability error rate $\epsilon$, we provide numerical simulations to study which distillation protocols (i.e., which choice of $n$ and abelian group of order $n$) lead to the smallest output error rates. 
\end{enumerate}

\section{Preliminaries}\label{sec:preliminaries}

\subsection{Linear optics}\label{sec:linear optics}

In what follows, we will consider the space of $n$ photons (not necessarily indistinguishable) in $m$ (external) modes. 
We begin with the \emph{first quantization}, following the notation of \cite{englbrecht2024indistinguishability}. 
In this framework, each photon has state space $\mc{H} = \h{ext}\otimes\h{int}$, where $\h{ext}$ corresponds to the ``external" modes\textemdash those manipulated by the experiment\textemdash and $\h{int}$ corresponds to the ``internal" modes\textemdash those not manipulated by the experiment. 
We identify $\h{ext} = \mathbb{C}^m$ with basis $\ket{0}, \dots, \ket{m-1}$, where $\ket{k}$ describes a photon in mode $k$. 
The state space corresponding to $n$ such photons is 
\begin{equation}
    \hn = \h{ext}^{\otimes n}\otimes \h{int}^{\otimes n}.
\end{equation}
Given a $\ket{\psi}\in\hn$, it has a corresponding \emph{external density matrix} in $\hne$ obtained by taking the partial trace over $\h{int}^{\otimes n}$. 
For experiments involving only linear optics on the external modes, the behavior of the state is fully characterized by its external density matrix \cite{englbrecht2024indistinguishability}. 
For our purposes, however, we will often need to consider both the external and internal degrees of freedom. 

The standard basis for $\hne$ in the first quantization picture is $\ket{m_1, \dots, m_n} = \ket{m_1}\otimes\cdots\otimes\ket{m_n}$, where the $i$th photon is in mode $m_i\in \{0, \dots, m-1\}$. 
When it is necessary to clarify that a state is using first quantization notation, we will instead write $\ket{m_1, \dots, m_n}_{1Q}$. 

In this work, we will consider many different representations of permutation groups on $\hn$ and the related subspaces. 
In the \emph{photon permutation representation}, the symmetric group $S_n$ acts on $\hn$ by permuting the tensor factors, thus permuting the particles. 
The group $U_{n,m}$ of \emph{linear optical unitaries} is the group of unitary operators on $\hn = \h{ext}^{\otimes n}\otimes \h{int}^{\otimes n}$ of the form $U^{\otimes n}\otimes I^{\otimes n}$, where $U\in U(m)$ is an $m$-mode unitary operator on $\h{ext}$. 
In particular, linear optical unitaries operate only on the external degrees of freedom and are symmetric with respect to permutation of photons. 
We will often write $\hat{U}$ to denote the linear optical unitary corresponding to $U\in U(m)$; when the context is clear, we often simply write $U$. 
We also extend the notation to arbitrary operators $T$ on $\mc{H}_{\textnormal{ext}}$ (not necessarily unitary), 
writing $\hat{T}$ for the operator $T^{\otimes n}\otimes I^{\otimes n}$ on $\hn = \h{ext}^{\otimes n}\otimes \h{int}^{\otimes n}$. 

We will also use the \emph{second quantization} representation of photonic states. We begin with the case in which only external degrees of freedom are considered. 
In this setting, we use the \emph{Fock basis} states $\ket{s_0, \dots, s_{m-1}} = \ket{s_0, \dots, s_{m-1}}_{2Q}$, which describe an $m$-mode state with $s_0$ photons in mode $0$, $s_1$ photons in mode $1$, and so on. 
For any mode $i$, we consider the \emph{creation operator} acting on that mode, 
\begin{equation}
    \cre_i = \sum_{n\geq 0} \sqrt{n} \ketbra{n}{n-1}_i.
\end{equation}
We also have $a_i=(a_i^\dag)^\dag$, the corresponding annihilation operator. 
Representing the \emph{vacuum} state by $\vac$, we may express all Fock basis states (up to normalization) by applying appropriate creation operators: 
\begin{equation}
    \ket{s_0, \dots, s_{m-1}}_{2Q} = \dfrac{1}{\sqrt{s_0! s_1!\cdots s_{m-1}!}}(\cre_0)^{s_0}\cdots (\cre_{m-1})^{s_{m-1}}\vac. 
\end{equation}
Linear optical unitary evolution in the second quantization is characterized by the same $U\in U(m)$ matrix mentioned above, which independently evolves each creation operator: $a_j^\dag \rightarrow \sum_{i=0}^{m-1} U_{ij}a_i^\dag$.
Note that the Fock basis states are symmetric with respect to permutation of photons: for example, 
\begin{equation}
    \ket{1,1}_{2Q} = \cre_0 \cre_1\vac = \frac{1}{\sqrt{2}}\left(\ket{0,1}_{1Q} + \ket{1,0}_{1Q}\right).
\end{equation}
Thus the Fock basis is a basis for the \emph{symmetric subspace} of $\h{ext}^{\otimes n}$. 

\begin{remark}\label{rem:conversion}
We may convert between first and second quantization as follows. We let $F$ be the symmetrization map from $\h{ext}^{\otimes n}$ to its symmetric subspace, with (for example) 
\begin{equation}
    F(\ket{0,1}_{1Q}) = \frac{1}{\sqrt{2}}\left(\ket{0,1}_{1Q} + \ket{1,0}_{1Q}\right) = \ket{1,1}_{2Q}.
\end{equation}
In particular, we have $F(\ket{m_0, \dots, m_{n-1}}_{1Q}) = \ket{s_0, \dots, s_{m-1}}_{2Q}$, where $s_i$ is the number of photons in mode $i$ (the number of indices $j$ with $m_j = i$). 
The map $F$ is not invertible, but given $\ket{s_0,\dots, s_{m-1}}_{2Q}$, there is a unique \emph{weakly increasing} $\ket{m_0, \dots, m_{n-1}}_{1Q}$ with $F(\ket{m_0, \dots, m_{n-1}}_{1Q}) = \ket{s_0,\dots, s_{m-1}}_{2Q}$. 
\end{remark}

When there are internal degrees of freedom under consideration, we use the notation $\cre_i[\xi_j]$ to indicate the creation of a particle with internal state $\ket{\xi_j}$ in external mode $i$. 
For a fixed orthonormal basis $\{\ket{\xi_0}, \ket{\xi_1}, \dots, \}$ of $\h{int}$, we will consider (appropriately normalized) states of the form 
\begin{equation}\label{eq:internal fock}
    \cre_{m_0}[\xi_{i_0}]\cdots \cre_{m_{n-1}}[\xi_{i_{n-1}}]\vac.
\end{equation}
These states are symmetric under permutation of photons, as the notation $\cre_k[\xi_{i_j}]$ only indicates the mode and internal state, not the particular photon (tensor factor in $\hn$); in fact, they form a basis for the photon-permutation-symmetric subspace of $\mathcal{H}^{\otimes n}$. 
This is in alignment with \cite{englbrecht2024indistinguishability}. 

We will consider two main types of measurements. The first is \emph{photon number resolving detection} (PNRD), formalizing the notion of a measurement that counts the number of photons in each mode. 
(We may also perform PNRD on a subset of modes in the obvious way.) 
This is a projective measurement on $\hn_\textnormal{ext}$, projecting onto the subspaces $V_{s_0, \dots, s_{m-1}}$ where
\begin{equation}
    V_{s_0, \dots, s_{m-1}} = \textnormal{span}\{\ket{m_0, \dots, m_{n-1}}_{1Q}: F(\ket{m_0, \dots, m_{n-1}}_{1Q}) = \ket{s_0, \dots, s_{m-1}}_{2Q}\}
\end{equation}
(recalling the map $F$ of Remark~\ref{rem:conversion}). 
This is just the space spanned by all states with $s_0$ photons in mode $0$, $s_1$ photons in mode $1$, and so on. 
We call $(s_0, \dots, s_{m-1})$ the obtained \emph{measurement pattern}. 
For symmetric states in $\h{ext}^{\otimes n}$, PNRD is simply projection onto the Fock basis; the above version allows for nonsymmetric states and internal degrees of freedom. 
Note that the internal degrees of freedom may be traced out for PNRD. 

For certain calculations, we will be interested in measuring the internal modes as well. 
Specifically, we consider \emph{internal-external measurement}, a projective measurement with respect to the states \eqref{eq:internal fock} (appropriately normalized). 
We note, however, that this will only be used mathematically, to calculate output error rates of distillation protocols; 
the protocols themselves do not require any measurements beyond PNRD. 

We now define perfect indistinguishability of photons, following \cite{englbrecht2024indistinguishability}. Let $\rho$ be a density matrix, and let $\rho_{\textnormal{ext}}$ be the corresponding external state, obtained by tracing out the internal Hilbert space. 
Let $P^{(n)}$ be the symmetrizer on $\h{ext}^{\otimes n}$, projecting to states that are invariant under permutation of photons. 
Then $\rho$ is \emph{perfectly indistinguishable} if $P^{(n)}\rho_\textnormal{ext} = \rho_\textnormal{ext}$. 
The main examples in this paper correspond to states 
$\rho = \ketbra{\psi}$, where $\ket{\psi} = \cre_{m_0}[\xi_0]\cdots \cre_{m_{n-1}}[\xi_0]|\vec{0}\rangle$. 
More generally, any state $\ket{\psi}$ that can be expressed as a pure tensor $\ket{\psi_\textnormal{ext}}\otimes\ket{\psi_\textnormal{int}}$, with $\ket{\psi_\textnormal{ext}}$ and $\ket{\psi_\textnormal{int}}$ symmetric under permutation of photons, is perfectly indistinguishable. 
We discuss a nontrivial such example in Appendix~\ref{sec:one}.
On the other hand, any state of the form \eqref{eq:internal fock} involving two distinct (orthogonal) internal states $\ket{\xi_i}, \ket{\xi_j}$ is \emph{not} perfectly indistinguishable.

\subsection{Models for distinguishability}\label{sec:models}

Following \cite{sparrow_quantum_2017, marshall_distillation_2022, saied2024advancing}, we will consider variants of the \emph{random source} (RS) model for photon distinguishability, 
in which photon sources independently produce photons with average internal state $\rho = \sum_{j\geq 0} p_j \ketbra{\xi_j}$, where the $\ket{\xi_j}$ range over an orthonormal basis for $\h{int}$.
We fix this basis $\{\ket{\xi_0}, \ket{\xi_1}, \dots\}$ for $\h{int}$ going forward. 
Note that this model is equivalent, under the assumption that the internal degrees of freedom cannot be resolved or manipulated by the experiment, to a pure state description $\sum_j e^{i\phi_j}\sqrt{p_j}|\xi_j\rangle$ \cite{sparrow_quantum_2017, saied2024advancing}. 
Without loss of generality, we assume that $p_0$ is the dominant eigenvalue and write $\epsilon = 1-p_0 = \sum_{j\geq 1} p_i$, so that
\begin{equation}
    \rho = \rho(\epsilon) = (1-\epsilon)\ketbra{\xi_0} + \epsilon \sum_{j\geq 1} p_j' \ketbra{\xi_j}
    \label{eq:rho_eps}
\end{equation}
in the RS model, with $p_j' = p_j/\epsilon$. 
We write $\rho = \rho(\epsilon)$ to emphasize the dependence on the parameter $\epsilon$ (suppressing the $p_i$ from the notation). 
We view the state as a probabilistic mixture: with large probability $1-\epsilon$ we obtain the ``ideal" state $\ket{\xi_0}$, and with small probability $\epsilon$ we obtain one of the ``distinguishable" error states $\ket{\xi_j}$, orthogonal to $\ket{\xi_0}$. 
In terms of creation operators, we may represent $\rho(\epsilon)$ by the appropriate probabilistic mixture of the creation operators $\cre[\xi_j]$. 

In this work, we often restrict our attention to the \emph{uniform RS} (URS) model, in which we take $p_1' = p_2' = \cdots = p_{R}'=1/R$ and $p_j = 0$ for $j> R$, so that
\begin{equation}
    \rho(\epsilon) = (1-\epsilon)\ketbra{\xi_0} + \frac{\epsilon}{R} \sum_{j=1}^{R} \ketbra{\xi_j}. 
\end{equation}

For ease of exposition and simulation, 
we will often consider two extreme cases that are useful for restricting the possible internal states. 
First, we have the Same Bad Bits (SBB) model \cite{saied2024advancing}, where we take 
$R = 1$ in the URS model, so that each photon has average internal state
\begin{equation}
    \rho(\epsilon) = (1-\epsilon)\ketbra{\xi_0} + \epsilon\ketbra{\xi_1}.
\end{equation}
This model captures a situation with systematic errors in which photons are prone to exhibit the same type of distinguishability error. 
This may be viewed as the special case in which 
the internal state space $\h{int}$ is $2$-dimensional. 
On the other hand, we consider the Orthogonal Bad Bits (OBB) model, where the $i$th photon has internal state
\begin{equation}
    \rho_i(\epsilon) = (1-\epsilon)\ketbra{\xi_0} + \epsilon\ketbra{\xi_{i+1}}. 
\end{equation}
Recall that $\braket{\xi_i}{\xi_j} = \delta_{ij}$. 
As discussed in \cite{sparrow_quantum_2017}, this corresponds to the limit of the URS model as $R\rightarrow\infty$. 
This model captures a situation in which error states are uniformly random with many degrees of freedom, so that it is vanishingly unlikely for the same error to occur twice in the same experiment. 
Since the OBB and SBB models may be viewed as opposite extremes of the URS model, we choose them as a particular focus for numerical simulation and examples. 

Finally, we compare our parameter $\epsilon$ to the \emph{visibility}. In the URS model, one can compute the visibility between two states drawn from the distribution as $V = \mathrm{Tr}[\rho(\epsilon)^2] = (1-\epsilon)^2 + \epsilon^2/R$. For the OBB limit, this gives visibility $V = (1-\epsilon)^2$. 
In any case we have $V = 1-2\epsilon + O(\epsilon^2)$, so when $V$ is large there is little difference between the corresponding values of $\epsilon$ in different error models. 
Using the above, visibility $V\geq 0.74$ (a rough lower bound on experimental visibility values in some recent literature \cite{halder_high_2008, tsujimoto_high_2017, tambasco_quantum_2018, wang_research_2019, ollivier_hong-ou-mandel_2021, somhorst_quantum_2023, alexander_manufacturable_2024}) approximately translates to $\epsilon \leq 0.15$. 
We are therefore primarily interested in effective distillation protocols for these relatively small values of $\epsilon$.

\subsubsection{Initial states}\label{sec:initial}

We briefly return to the general RS model. 
We will take $m=n$, considering 
experiments in which $n$ photons are injected into $n$ different (external) modes of an interferometer, with the photon in mode $i$ having internal state $\rho_i(\epsilon)$. 
When $\epsilon=0$ and internal degrees of freedom are neglected, this corresponds to the idealized Fock state $\ket{1, \dots, 1}_{2Q}$. 
In general, we call the initial state $\rho^{(n)} = \rho^{(n)}(\epsilon)$. 
We may decompose $\rho^{(n)}$ as follows: 
\begin{align}\label{eq:total_2Q}
    \rho^{(n)} = \sum_{j_0, \dots, j_{n-1}} (p_{j_0}\cdots p_{j_{n-1}})\cre_0[\xi_{j_0}] \cdots \cre_{n-1}[\xi_{j_{n-1}}]\kbvac a_{0}[\xi_{j_{0}}]\cdots a_{n-1}[\xi_{j_{n-1}}].
\end{align}
This is a probabilistic mixture of many terms; 
by abuse of notation, we often write
\begin{equation}\label{eq:internal fock experiment}
    \cre_0[\xi_{j_0}] \cdots \cre_{n-1}[\xi_{j_{n-1}}]\vac = \ket{1, \dots, 1}_{2Q}\otimes\ket{\xi_{j_0},\dots,\xi_{j_{n-1}}}.
\end{equation}
Note that in general, the left-hand side is not actually a pure tensor. 
For example, $\cre_0[\xi_0]\cre_1[\xi_1]\vac = |0,1\rangle_{1Q}\otimes |\xi_0, \xi_1\rangle + |1,0\rangle_{1Q}\otimes |\xi_1, \xi_0\rangle$ (up to normalization). 

We say that a state \eqref{eq:internal fock experiment} has $k$ \emph{distinguishability errors}, where $k$ is the number of indices $i$ with $j_i\neq 0$. 
We will often consider small values of $\epsilon$, so that the expression is dominated by the terms that are first order or less in $\epsilon$; in other words, those with $0$ or $1$ distinguishability errors. 
We note that, for any RS model, the weight of the terms with exactly $1$ distinguishability error is $n(1-\epsilon)^{n-1}\sum_{j\geq 1}p_j = n\epsilon (1-\epsilon)^{n-1}$. 
For the calculations of interest, we will see that these first-order terms exhibit essentially the same behavior regardless of the particular model. 

In the URS model (or its OBB limit), we write the input state as 
\begin{align}\label{eq:error decomp}
    \rho^{(n)}(\epsilon) = \sum_{k=0}^n \binom{n}{k}\epsilon^k(1-\epsilon)^{n-k} \Phi_k, 
\end{align}
where $\Phi_k$ is the mixed state involving all terms with $k$ distinguishability errors, normalized to have $\Tr(\Phi_k)=1$. 
Going forward, we will exclusively discuss the URS framework, but the results involving only terms with $0$ or $1$ distinguishability errors, such as Theorems \ref{thm:error rate}, \ref{thm:herald first order}, \ref{thm:herald ideal}, \ref{thm:resources}, \ref{thm:symm body} in fact hold for any RS model.

\subsection{Distillation and error correction}\label{sec:error correction}

Before proceeding to the main results on distillation of distinguishability errors, we give a brief overview of a major motivation for this work, namely the potential relevance of distillation to photonic quantum error correction. 
As discussed later in this section, distinguishability errors can lead to logical errors in fault-tolerant quantum computation. As a result, reducing distinguishability error rates can lead to higher thresholds for other errors (such as photon loss), which may be harder to engineer away. 
This can therefore allow for smaller code distances and smaller system sizes. 
Of course, this could in principle be achieved by either distillation or source engineering. In practice however, engineering a source to achieve (for example) an order of magnitude reduction in error, historically, appears to be a challenging task. 
Distillation, on the other hand, is a general, tunable, and source-agnostic approach that can be used in a modular fashion as part of single photon generation. 
The new protocols are highly efficient (with an order of magnitude error reduction requiring around $40$ photons) and can be implemented with existing technologies \cite{faurby2024purifying}, making them a promising supplement to source engineering. 
We therefore believe both source engineering and distillation to be important aspects to consider for realizing highly pure, indistinguishable, and heralded single photon sources. 

For our discussion of quantum error correction, we consider \emph{fusion-based quantum computation} (FBQC), a paradigm for universal fault-tolerant quantum computation starting from single photon sources \cite{bartolucci_fusion-based_2023-1}. The details of FBQC are beyond the scope of this work, but the main idea is to use probabilistic entangling measurements and multiplexing to generate many copies of constant-sized entangled \emph{resource states}, then carry out (destructive) probabilistic Bell measurements, called \emph{Type II fusions}, between appropriate pairs of qubits (photons). This is similar in spirit to other measurement-based approaches to quantum error correction, but with the state generation happening in a modular fashion to allow for multiplexing. 
We view Type II fusion as measuring two observables, $X\otimes X$ and $Z\otimes Z$. 
One may show \cite{rohde2006error, sparrow_quantum_2017} that fusion between an ideal photon and another with internal state $\rho(\epsilon)$ results in flipping the value of the $X\otimes X$ observable with probability $\epsilon/2$. 
Then in FBQC, we may approximate distinguishability as leading to errors in fusion measurements. In particular, since the roles of the $X\otimes X$ and $Z\otimes Z$ observables are often swapped, we roughly approximate a distinguishability error rate $\epsilon$ by measurement errors with probability $\epsilon/4$. 

To illustrate the effect of distinguishability errors in this setting, we consider the (already very small) error rate $\epsilon\approx 5\times 10^{-3}$, roughly corresponding to the visibility achieved in \cite{alexander_manufacturable_2024}. This translates to an approximate measurement error rate of $p_{m}=\epsilon/4\approx 1.25\times 10^{-3}$. 
For simplicity, we consider the most basic FBQC architecture, using the $6$-ring resource state to construct an analogue of the surface code \cite{bartolucci_fusion-based_2023-1, bombin2023logical}. In that setting, our error rate $p_m$ is below the measurement error threshold and yields a photon loss threshold of approximately $6\times 10^{-3}$, as seen in Figure 9 of \cite{bombin2023increasing}. 
However, we will see that there is significant benefit to decreasing $p_m$ further. 
As calculated in Section~\ref{sec:iterated}, applying our $H_{16}$ distillation protocol reduces the distinguishability error rate to $\epsilon'\approx 3.6\times 10^{-4}$, with a corresponding approximate FBQC measurement error rate of $p_m'\approx 8.9\times 10^{-5}$ and photon loss threshold very near the maximum threshold of $1.6\times 10^{-2}$. 
Thus we nearly \emph{triple} the photon loss threshold by improving the distinguishability error rate even beyond the values obtained by state-of-the-art single photon sources. 
Due to the high photon loss rates in linear optics, increasing this threshold is a major priority. 

We now consider the effect of reduced error rates on the code distance and resource requirements. 
Let us consider a simple (unrealistic) model with no photon loss, only measurement error due to distinguishability. 
The logical error rate $p_L$ roughly scales as follows, where $p_m$ is the measurement error rate, $p_{th}$ is the measurement error threshold, and $d$ is the code distance \cite{fowler2012surface, bombin2023logical, litinski2022active_doi}: 
\begin{equation}
    p_L\sim C \left(\dfrac{p_m}{p_{th}}\right)^{d/2}. 
\end{equation}
Taking $C=1$ for convenience, we consider $p_m=1.25\times 10^{-3}$ as in the above example, and the corresponding lossless measurement error threshold $p_{th}=2.1\times 10^{-3}$ (from Figure 9, \cite{bombin2023increasing}). 
This leads to
\begin{equation}
    p_L\sim (0.77)^d. 
\end{equation}
If we target logical error rate $p_L\leq 10^{-6}$, this requires a code distance of $d\geq 53$. 
On the other hand, if we use $H_{16}$ distillation as above to reduce the measurement error rate to $p_m'\approx 8.9\times 10^{-5}$, we obtain 
\begin{equation}
    p_L'\sim (0.21)^d,
\end{equation}
requiring a code distance of $d\geq 9$ to obtain the same desired logical error rate. 
Since this FBQC architecture is carried out in ``logical blocks" of $d^3$ resource states \cite{bombin2023logical}, distillation would reduce the required number of resource states per logical block from $148877$ to $729$. 
Of course, in practice, one will need $d$ much larger than $9$ in order to handle photon loss, which is the dominant source of noise. 
But as argued above, reduced distinguishability also leads to increased tolerance of photon loss. 
This in turn will reduce the required code distance via a similar argument. 

In summary, we see that reducing distinguishability error rates far below the error threshold can improve thresholds for other types of errors, such as photon loss, and can greatly reduce the large-scale resource requirements of a fault-tolerant linear optical quantum computation. 
We may view distillation and source engineering as complementary means toward achieving this reduction. 
Since distillation requires only standard linear optical hardware, it is likely far cheaper to use distillation rather than purchasing or engineering a state-of-the-art photon source. 
Further, as seen in the example above, even the best known photon sources may not have low enough distinguishability error rates for efficient fault-tolerance. 
In such a setting, one may apply distillation to the output of state-of-the-art single photon sources to obtain even smaller distinguishability error rates and therefore more scalable and efficient linear optical quantum computation.

\section{Main results}\label{sec:main}

\subsection{Distillation protocols}\label{sec:distillation}

In this work, we will consider \emph{distillation protocols}  generalizing the work of \cite{marshall_distillation_2022}, described in Protocol~\ref{proto:distillation} and illustrated in Figure~\ref{fig:dist-cartoon}. 
Unless stated otherwise, we assume that the only errors are distinguishability errors, and we model distinguishability with the URS error model (or the OBB limit) with error rate $\epsilon$. 
(Photon loss will be addressed in Section~\ref{sec:loss}.) 

Consider an $n\times n$ matrix $U$. Define the set of \emph{ideal patterns} corresponding to $U$ to be the set of all $(s_0, \dots, s_{n-1})$ such that 
\begin{equation}\label{def:ideal pattern}
    s_0 = 1\textnormal{ and }\bra{s_0, \dots, s_{n-1}}\hat{U}\ket{1, \dots, 1} \neq 0,
\end{equation}
where both states are Fock basis states (with no internal degrees of freedom). 

\begin{protocol}\label{proto:distillation}
Given a number $n$ of modes and an $n\times n$ unitary matrix $U$ (corresponding to a linear optical unitary), we describe the corresponding $n$-photon \emph{distillation protocol}. 
\begin{enumerate}
    \item Input a state $\sigma$ of $n$ photons in $n$ distinct modes. 
    We will typically consider $\sigma = \rho^{(n)}(\epsilon)$. 
    \item Apply $\hat{U}$, mapping $\sigma\mapsto \mu = \hat{U}\sigma \hat{U}^\dagger$. 
    \item Perform PNRD on the final $n-1$ modes of $\mu$, 
    obtaining a \emph{measurement pattern} $(s_1, \dots, s_{n-1})$ with $s_i$ being the number of photons found in mode $i$. 
    Label mode $0$, the unmeasured mode, as the \emph{output mode}; we call the single-mode post-measurement state the \emph{output state}. Letting $s_0 = n-\sum_{j\geq 1} s_j$, the output state has $s_0$ photons. We call $(s_0, \dots, s_{n-1})$ the corresponding \emph{completed measurement pattern}. 
    \item \emph{Post-select} for ideal patterns, as defined in \eqref{def:ideal pattern}. In other words, if the pattern is not ideal, the photon is rejected. 
\end{enumerate}
\end{protocol}
The purpose of the distillation protocol is to herald the output of single photons with reduced distinguishability error rate, enabling more efficient linear optical quantum computation. 
By post-selecting for patterns with $s_0=1$, we guarantee that the output state has a single photon. (Assuming distinguishability is the only error.) 
By post-selecting for ideal patterns in particular, we \emph{herald} the output of a single photon with (hopefully) a smaller rate of distinguishability error. 
This is simply an application of conditional probability: one needs to show that, given the knowledge that the completed measurement pattern is ideal, the output photon is more likely to be ideal.

Given such a distillation protocol with input state $\sigma = \rho^{(n)}(\epsilon)$, there will be two main metrics of interest, functions of the input error rate $\epsilon$. 
First is the \emph{heralding rate} $h_n(\epsilon)$, the probability of obtaining an ideal pattern during the measurement step. 
Second is the \emph{output error rate} $e_n(\epsilon)$. If distinguishability is the only error, this is the probability that, given the heralding of an ideal pattern, 
the output photon has a ``distinguishable" internal state $\ket{\xi_i}$, $i>0$. 
(This is made rigorous via an internal-external measurement, as discussed in Section~\ref{sec:linear optics}.  
We discuss the case in which there are loss errors in addition to distinguishability in Section~\ref{sec:loss}.
) 
For an arbitrary density matrix $\sigma$, we write $h_n(\sigma)$, $e_n(\sigma)$ to be the corresponding heralding and output error rates for the distillation protocol with input state $\sigma$. 

We will be particularly interested in protocols using the quantum Fourier transform and Hadamard matrices. 
We briefly discuss these matrices and the corresponding ideal patterns now: for a more detailed analysis, we refer to Appendix~\ref{sec:symmetry}. 
We also compare to the case of Haar random matrices in Appendix~\ref{sec:haar}. 

The Hadamard matrices under consideration are defined for $n=2^r$ by $H_n = H^{\otimes r}$, where $H$ is the standard $2\times 2$ Hadamard matrix. 
(We note that this family is typically called the family of \emph{Sylvester matrices}.) 
A recursive linear optical circuit for $H_n$ is given in Figure~\ref{fig:hadamard-circuit}. 
For $H_n$, we always take $n$ to be a power of $2$ with $n\geq 4$, even if not explicitly stated. 
The case $n=4$ is equivalent to the $4$-photon distillation protocol in the appendix of \cite{marshall_distillation_2022}, where there is only one ideal pattern, $(1,1,1,1)$. 
To describe the ideal patterns for $H_n$, we take a completed measurement pattern $p = (s_0, \dots, s_{n-1})$ and ``convert to first quantization," obtaining integers $0\leq g_0\leq g_1\leq \dots\leq g_{n-1}\leq m-1$ such that $s_i$ is the number of indices $j$ with $g_j = i$. 
For example, $p=(1,2,0,1)$ corresponds to $g = (0,1,1,3)$. 
With this notation, the ideal patterns for $H_n$ are exactly those satisfying $s_0=1$ and $g_0\oplus \cdots \oplus g_{n-1} = 0$, where $\oplus$ is binary XOR without carrying \cite{crespi2015suppression}. 

Letting $\omega = \exp(2\pi i/n)$, the $n$-mode Fourier transform matrix is given by 
\begin{equation}
    F_n = \frac{1}{\sqrt{n}}(\omega^{ij})_{0\leq i,j\leq n-1}. 
\end{equation}
The distillation protocol for $F_3$ is equivalent to the $3$-photon protocol in the main text of \cite{marshall_distillation_2022}, where there is only one ideal pattern, $(1,1,1)$. 
With the same ``first quantization" notation as above, the ideal patterns for $F_n$ satisfy the \emph{Zero Transmission Law} (ZTL), $g_0 + \dots + g_{n-1} \equiv 0 \mod n$ \cite{tichy2010zero, dittel2018totally}. 
In Theorem~\ref{thm:suppression body}, we prove that whenever $n$ is a prime power, then the ideal patterns are exactly characterized by the ZTL (and the condition $s_0=1$). 
Further, for $n$ not a prime power, we show that the set of ideal patterns is a proper subset of those satisfying this condition (and $s_0=1$). We further discuss these properties in Sections~\ref{sec:symmetry body} and Appendix~\ref{sec:fourier symm} below. 

More generally, we also consider 
\begin{equation}\label{eq:fourier general body}
    F_{(n_1, \dots, n_\ell)} = F_{n_1}\otimes \cdots\otimes F_{n_\ell},
\end{equation}
where $n=n_1\cdots n_\ell$ and the $n_i$ are integers satisfying $n_i\geq 2$. 
Noting that $H=F_2$, we have $F_n = F_{(n)}$ and $H_n = F_{(2,2,\dots, 2)}$. 
We discuss this case in far greater detail in Appendix~\ref{sec:fourier abelian}. 
Our discussion in the main text will focus primarily on the $F_n$ ($n\geq 3$) and $H_n$ ($n\geq 4$) cases, with mentions of the ``Fourier" case generally referring to $F_n$ unless stated otherwise. 
(Note, even if not explicitly stated, we will always consider $n>2$ because $H_2 = F_2 = H$ has no ideal patterns, by the standard HOM effect.) 
We have the following result. 

\begin{theorem}\label{thm:error rate}
Consider an $n$-photon distillation protocol with unitary $U=F_{(n_1, \dots, n_\ell)}$, $n=n_1\cdots n_\ell > 2$, and input state $\rho^{(n)}(\epsilon)$. 
(Note this encompasses the cases $F_n$ and $H_n$.) 
The output error rate satisfies
\begin{align}
    e_n(\epsilon) = \frac{\epsilon}{n} + O(\epsilon^2).
\end{align}
\end{theorem}

This is proven in Appendix~\ref{sec:one}. 
As a consequence of the theorem, we see that for sufficiently small $\epsilon$, $e_n(\epsilon)\approx \epsilon/n$. 
Then, when $\epsilon$ is small enough, the protocols with larger $n$ are better at reducing distinguishability errors. 
This is visible in Figure~\ref{fig:err_reduction_fourier}. 
However, we note that as $\epsilon$ grows, the protocols with large $n$ tend to perform worse, 
due to the higher order terms not considered in Theorem~\ref{thm:error rate}. 
In particular, each protocol has a \emph{distillation threshold}, a value of $\epsilon$ after which it does not necessarily reduce the error rate. 
We numerically plot some representative threshold values in Figure~\ref{fig:epsilon_critical}; these values appear to decrease as $n$ increases, likely approaching $0$ in the limit $n\rightarrow\infty$. 
For practical values of $n$, however, the thresholds are reasonably large. For example, the $16$-photon protocols improve the error rate as long as $\epsilon< 0.179$ (for OBB and SBB error models). 

\begin{remark}
    We note that Theorem~\ref{thm:error rate} still holds even if, in Protocol~\ref{proto:distillation}, we post-select for only a \emph{subset} of the ideal patterns. 
    In particular, \cite{somhorst2024photon} (released simultaneously with the first version of the present work) considers the output error rates obtained by post-selecting on only a \emph{single} ideal pattern at a time. 
    Each pattern yields the same reduction in error up to first order; however, \cite{somhorst2024photon} observes numerically that, when considering \emph{higher order} terms in $\epsilon$, different patterns may lead to significantly different output error rates. 
    Thus one might consider pruning the set of ideal patterns and keeping only those with the smallest output error rates. 
    When $\epsilon$ is small relative to the distillation threshold, so that higher order terms are negligible, this will not be beneficial: a smaller set of ideal patterns leads to a smaller heralding rate and thus larger resource costs, without significantly improving the output error rate. (See Section~\ref{sec:iterated} for a discussion of resource costs.) 
    However, as seen in Figure~\ref{fig:err_reduction_fourier}, higher order terms become increasingly nontrivial as $\epsilon$ grows, and it is worth investigating whether such methods could lead to improved performance for large $\epsilon$. 
\end{remark}

In order to understand the utility of a distillation protocol, we must also consider the heralding rate, which determines the resource cost. 
We investigate this problem in Section~\ref{sec:heralding}.

\begin{figure}[h]
    \centering
    \includegraphics[width=\columnwidth]{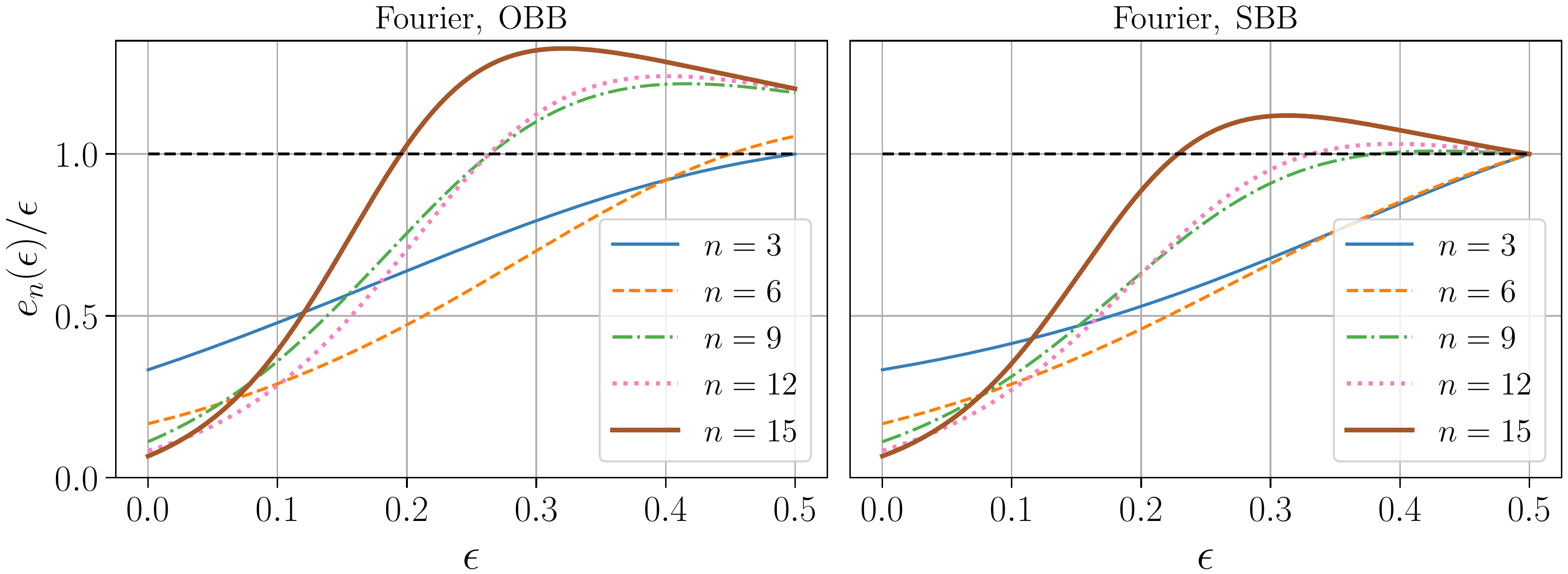}
    \caption{Plot of the error rate reduction with OBB (left) and SBB (right) error models, for the Fourier ($F_n$) distillation protocol. For $\epsilon\rightarrow 0$, $e_n(\epsilon)/\epsilon \rightarrow 1/n$. 
    Distillation is successful at reducing the error rate when the curve is less than $1$, marked by the horizontal dashed line. 
    (Also see Fig.~\ref{fig:epsilon_critical}).}
    \label{fig:err_reduction_fourier}
\end{figure}

\begin{figure}[h]
    \centering
    \includegraphics[width=0.7\columnwidth]{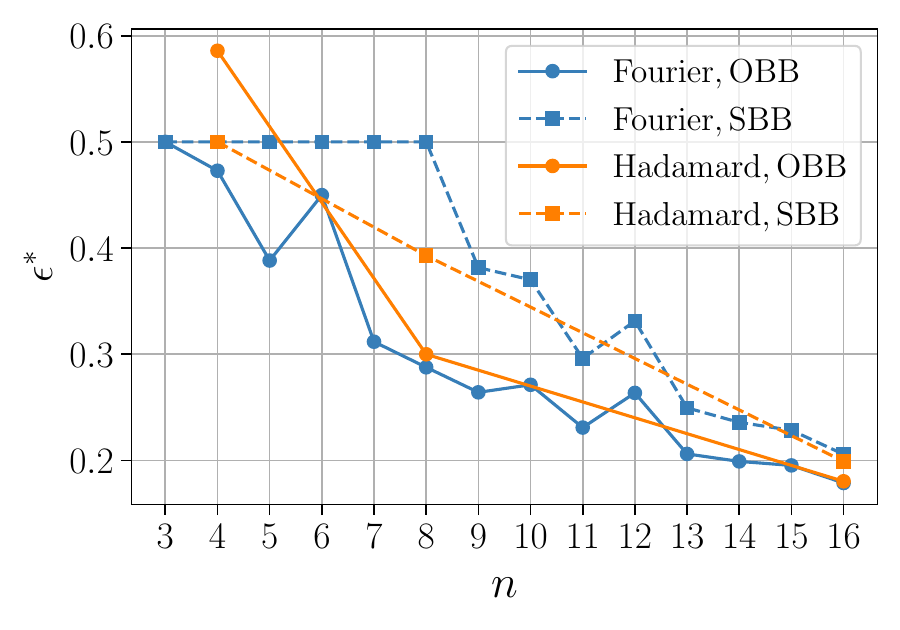}
    \caption{Plot of the distillation thresholds, the smallest value $\epsilon^*>0$ such that $e_n(\epsilon^*) = \epsilon^*$. 
    In particular, distillation is advantageous for the given protocol, as long as the initial error $\epsilon$ satisfies $\epsilon < \epsilon^*$.
    We plot both SBB and OBB noise models, for Hadamard and Fourier $(F_n)$ matrices. 
    Numerically we observe the Fourier OBB case to scale approximately as $\epsilon^* \sim 1/n^{0.61}$.}
    \label{fig:epsilon_critical}
\end{figure}

\subsection{Heralding rates}\label{sec:heralding}

For a distillation protocol with heralding probability $h_n(\epsilon)$, 
we expect to repeat it $1/h_n(\epsilon)$ times in order to successfully herald a distilled output photon. 
Since each iteration involves $n$ photons, the expected number of photons required for a single distilled output photon is then $n/h_n(\epsilon)$. 
Thus, we are interested in estimating the heralding rate, and especially obtaining lower bounds, in order to obtain upper bounds on the required resources.

Recalling the decomposition \eqref{eq:error decomp} for $\rho^{(n)}(\epsilon)$, we have
\begin{align}
    h_n(\epsilon) &= \sum_{k=0}^n \binom{n}{k}\epsilon^k (1-\epsilon)^{n-k} h_n(\Phi_k) \label{eq:herald_expanded_start}
    \\&= (1-\epsilon)^n h_n(\Phi_0) + n \epsilon (1-\epsilon)^{n-1} h_n(\Phi_1) + O(\epsilon^2) 
    \\&= h_n(\Phi_0) + \epsilon n(h_n(\Phi_1) - h_n(\Phi_0)) + O(\epsilon^2).\label{eq:herald_expanded_end}
\end{align}
Thus, for sufficiently small $\epsilon$, the heralding rate is determined by $h_n(\Phi_0)=h_n(0)$ and $h_n(\Phi_1)$, corresponding to heralding rates for input states with $0$ or $1$ errors. 
We have: 

\begin{theorem}\label{thm:herald first order}
Consider an $n$-photon distillation protocol with unitary $U=F_{(n_1, \dots, n_\ell)}$, $n=n_1\cdots n_\ell > 2$, 
and input state $\rho^{(n)}(\epsilon)$. 
(In particular, we may take $U=F_n, H_n$.)
\begin{enumerate}
    \item We have
    $h_n(\Phi_1) = \frac{1}{n}h_n(0).$
    \item The heralding rate is given by
    \begin{equation}
        h_n(\epsilon) = h_n(0)-(n-1)h_n(0)\epsilon + O(\epsilon^2).
    \end{equation}
\end{enumerate}
\end{theorem}
\begin{proof}
The first claim is proven in Appendix~\ref{sec:one}. 
The second claim follows from the first and \eqref{eq:herald_expanded_start}-\eqref{eq:herald_expanded_end}. 
\end{proof}

Then up to first order in $\epsilon$, the heralding rate is entirely governed by the ideal heralding rate $h_n(0)$. The ideal heralding rate in the cases of interest is characterized as follows. In Appendix~\ref{sec:haar}, we compare these results to the case of Haar random linear optical unitaries. 
\begin{theorem}\label{thm:herald ideal}
Let $U=F_{(n_1, \dots, n_\ell)}$, where $n=n_1\cdots n_\ell > 2$. More generally, we may take $U$ to be any $n\times n$ unitary with entries in the first row identically equal to $1/\sqrt{n}$. 
\begin{enumerate}
    \item \label{part:herald ideal 1}The ideal heralding rate has the following equivalent expressions: 
\begin{align}
    h_n(0) &= \left(\frac{-1}{n}\right)^{n-1}(n-1)! \,\sum_{t=0}^{n-1}  (n-t) \frac{(-n)^t}{t!} 
    \\&= {}_2 F_0(-(n-1), 2; 1/n), 
    \label{eq:hypergeom_hn}
\end{align}
where ${}_2 F_0(a,b; z) = \sum_{j\geq 0}\dfrac{(a)_j (b)_j}{j!}z^j$ is the generalized hypergeometric function, with $(a)_j = a (a+1) \cdots (a + (j-1))$.
\item \label{part:herald ideal 2}$\lim_{n\rightarrow \infty}h_n(0)=1/4$. 
\end{enumerate}
\end{theorem}
The proof of the first equality in Part~\ref{part:herald ideal 1} is given in Appendix~\ref{sec:ideal herald}; the second follows from reindexing $t\mapsto n-1-j$ and rewriting the terms. 
Part~\ref{part:herald ideal 2} follows from a result of Vaclav Kotesovec \cite{oeis}, proven using the Maple library \emph{gdev} \cite{salvy1991examples}. 
We adapt this proof to our setting in Appendix~\ref{sec:maple appendix}. 

\begin{remark}
We note that $\{n^{n-1}h_n(0)\}_{n\geq 1}$ is visibly a sequence of integers. Kotesovec entered the same sequence into The On-Line Encyclopedia of Integer Sequences as sequence A277458 in 2016 \cite{oeis}, as the coefficients of an exponential generating function related to the Lambert $W$ function. 
The asymptotic above is given in the encyclopedia listing, up to the factor of $n^{n-1}$. 
We note that the $W$ function appears in many physical problems \cite{valluri2000some}, and certain variants were recently used in proposed protocols for verification of Gaussian boson sampling \cite{mendes2022applications}. 
\end{remark}

\begin{figure}
    \centering
    \includegraphics[width=0.6\columnwidth]{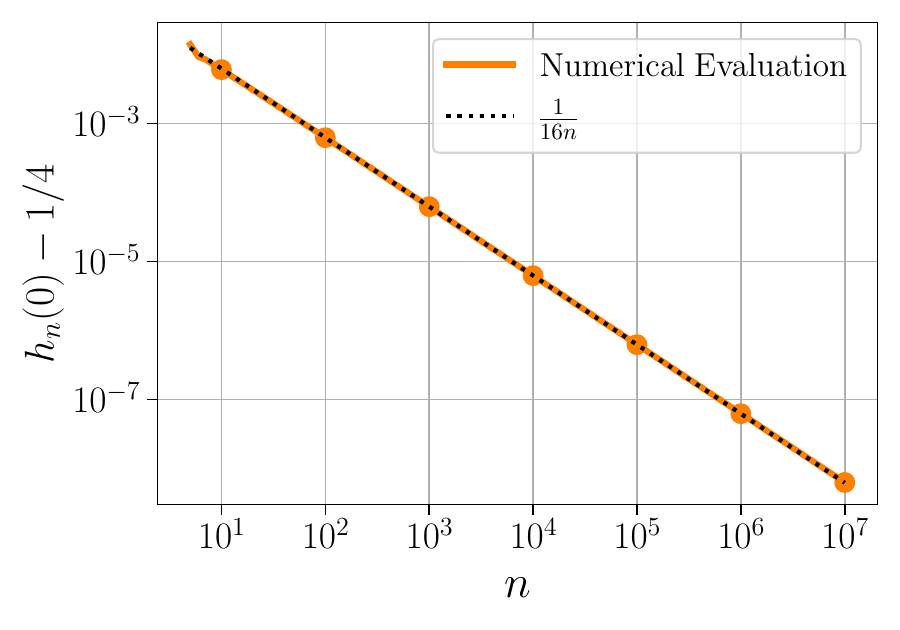}
    \caption{Numerical evaluation of $h_n(0)-1/4$ for $n=5$ up to $10^7$. The data (solid line) contains all values of $n\in[5, 10^5]$, as well as multiples of $10^5$ up to $10^6$, and multiples of $10^6$ up to $10^7$ (we include circular data points at powers of 10 to guide the eye). This was computed using the hypergeometric function, Eq.~\eqref{eq:hypergeom_hn}. We observed numerically that the deviation of $h_n(0)$ from $1/4$ falls as $1/(16n)$, plotted as the dashed line.}
    \label{fig:hn0}
\end{figure}

The hypergeometric expression allows for quick numerical estimation of the ideal heralding rate. 
The first few values are $h_3(0) = 1/3$, $h_4(0) = 1/4$, $h_5(0)\approx 0.264$. 
In Figure~\ref{fig:hn0}, we plot the difference $h_n(0) - 1/4$ and numerically observe that $h_n(0)$ quickly approaches $1/4$, with asymptotic $h_n(0) - 1/4 \sim \frac{1}{16n}$. 
In particular, we conjecture
\begin{conjecture}\label{conj:herald}
For $n\geq 5$, the sequence $h_n(0)$ monotonically decreases to $1/4$. 
\end{conjecture}
As an immediate consequence of Theorems \ref{thm:herald first order} and \ref{thm:herald ideal}, we have: 

\begin{proposition}
For $n\geq 4$ and sufficiently small $\epsilon$, 
\begin{equation}
    h_n(\epsilon) \approx \frac{1}{4} - \left(\dfrac{n-1}{4}\right)\epsilon.
\end{equation}
\end{proposition}

In Figure~\ref{fig:herald}, we numerically plot heralding rates as a function of $\epsilon$ for the Fourier case. 
This plot, and others in this work, are obtained by numerically calculating the terms $h_n(\Phi_k)$, $0\leq k\leq n$, and substituting into \eqref{eq:herald_expanded_start}-\eqref{eq:herald_expanded_end}. Thus we obtain $h_n(\epsilon)$ as a polynomial in $\epsilon$ with explicitly known coefficients, up to rounding errors in the calculations for $h_n(\Phi_k)$. 
No sampling or fitting with respect to particular values of $\epsilon$ is required. 
In Appendix~\ref{sec:numerics appendix}, we explicitly give a list of computed heralding rates and output error rates as functions of $\epsilon$. 
We see from Figure~\ref{fig:herald} that the protocol with $n=3$ has a significantly better heralding rate than the larger protocols; however, as observed in Theorem~\ref{thm:error rate}, the larger protocols have better output error rates for small $\epsilon$. 
Further, since the heralding rate $h_n(0)$ stabilizes at $1/4$, it seems that for small $\epsilon$ it is advantageous to increase the value of $n$. 

We use a similar expansion to the above, discussed in Appendix~\ref{sec:calculations}, to calculate the output error rates $e_n(\epsilon)$, plotted in Fig.~\ref{fig:err_reduction_fourier}. 
We note that as $\epsilon$ grows larger, it is not necessarily better to increase $n$; 
this is especially clear in Figs.~\ref{fig:epsilon_critical} and \ref{fig:opt_n}. 
We further discuss the tradeoffs between different protocols in Section \ref{sec:iterated}. 

\begin{figure}[h]
    \centering
    \includegraphics[width=\columnwidth]{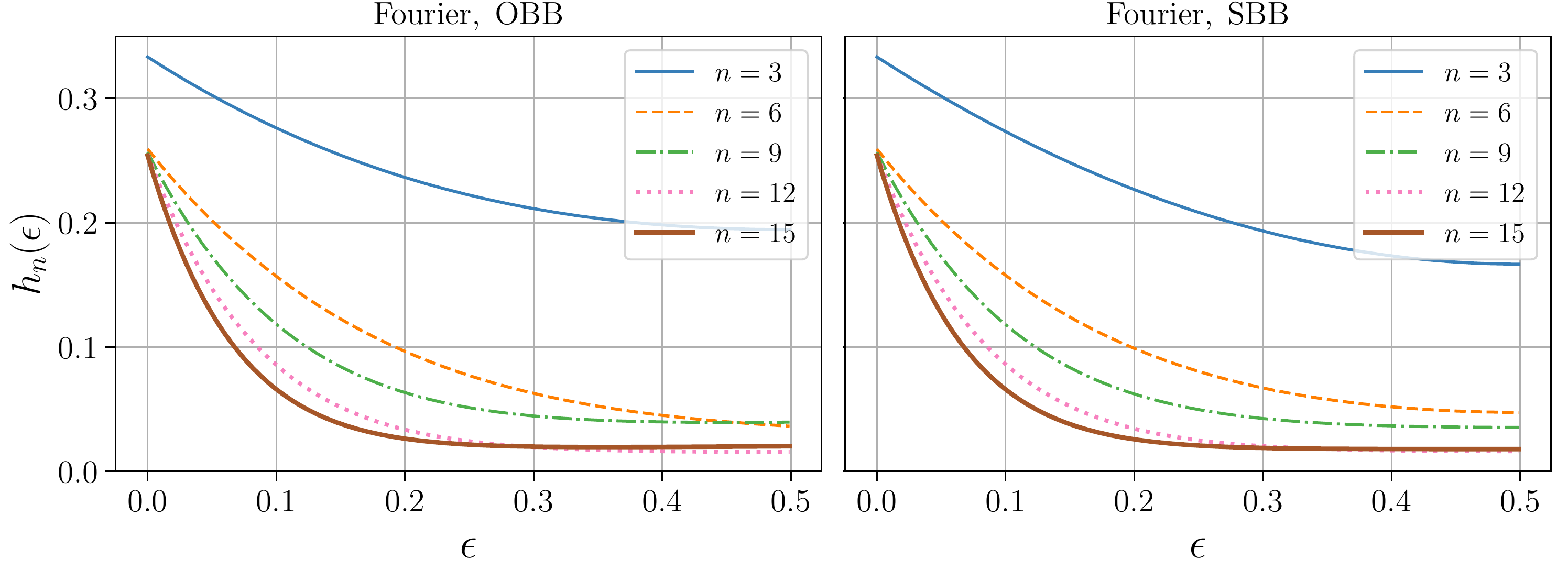}
    \caption{Plot of the heralding rate $h_n(\epsilon)$ with OBB (left) and SBB (right) error models for the Fourier distillation protocol. The zero-error heralding rate $h_n(0)$ is given analytically by Th.~\ref{thm:herald ideal}.}
    \label{fig:herald}
\end{figure}

\subsection{Comparison of resource costs}\label{sec:iterated}
We now compare the resource costs of our distillation protocols with previous work. In particular, we consider the expected number of photons required to obtain output error rate $\epsilon'\approx \epsilon/n$ in the small $\epsilon$ regime. 

The first photon distillation protocol, due to Sparrow and Birchall and called HOM filtering \cite{sparrow_quantum_2017}, requires $n$ photons per distillation attempt, with heralding probability decaying exponentially to $0$ \cite{marshall_distillation_2022}. 
Thus the expected number of photons required is exponential in $n$. If instead one iterates the 2-photon version of this scheme with active feed-forward, it is possible to have a scheme that scales cubically $O(n^3)$, as pointed out in \cite{marshall_distillation_2022}.

Marshall \cite{marshall_distillation_2022} proposed the $3$-photon Fourier distillation scheme and suggested iterating the scheme to arbitrarily reduce the error rate at the cost of additional photons. 
In particular, many batches of photons with error rate $\epsilon$ are put through a $3$-photon distillation protocol, giving output photons with error rate near $\epsilon/3$ upon successful heralding; these then undergo further rounds of distillation to obtain error rate approximately $\epsilon/3^2$, then $\epsilon/3^3$, and so on. 
Without active feed-forward, it is easy to see that the expected number of photons to obtain output error rate $\epsilon'\approx \epsilon/n$ (where $n$ is a power of $3$) 
is exponential in $n$, since it is highly unlikely for many distillation protocols to succeed simultaneously. 
With active feed-forward and the ability to perfectly store successfully heralded photons in memory, Marshall's iterated distillation protocol can be made more efficient: the output photons from one round of distillation can be stored until enough are available to use as input for subsequent rounds of distillation. 
Such an iterated protocol reduces the error rate by a factor of $n=3^r$ with an expected cost of $n^2$ photons \cite{marshall_distillation_2022}. 
In practice, however, these additional requirements may introduce more errors, for example due to dispersion introduced by optical delay lines. 

In our present setting, an $n$-photon protocol $F_{(n_1, \dots, n_\ell)}$ obtains output error rate $e_n(\epsilon)\approx \epsilon/n$ upon successful heralding in the small $\epsilon$ regime. 
The heralding rate satisfies $h_n(0)\approx 1/4$ for $n\geq 4$ (see Theorem~\ref{thm:herald ideal} and the following discussion). 
Then we have: 
\begin{theorem}\label{thm:resources}
For $n\geq 4$ and $\epsilon$ sufficiently small relative to $n$, the expected number of photons required to distill a single output photon with error rate $e_n(\epsilon)\approx \epsilon/n$ is 
\begin{equation}
    \dfrac{n}{h_n(\epsilon)}\approx \dfrac{n}{h_n(0)}\approx 4n.
\end{equation}
\end{theorem}
In particular, our $n$-photon distillation protocol requires only \emph{linearly scaling resources} to obtain the target error in the low $\epsilon$ regime, \emph{without} requiring feed-forward or memory. 
Thus the new protocols are the most resource-efficient known method for distilling the distinguishability error rate to an arbitrarily small value. 

\begin{remark}\label{remark:efficiency}
We now more closely consider the feasibility of the new distillation protocols relative to previous work, specifically the iterated $F_3$ protocol of \cite{marshall_distillation_2022} discussed above. For $n=3^r$, the iterated $F_3$ protocol requires $r$ concatenated rounds of $F_3$ distillation protocols, with an expected cost of $n^2$ photons, in order to reduce the error rate by a factor of $n$. Further, the iterated protocol of \cite{marshall_distillation_2022} (for $r>1$ rounds) requires active feed-forward and storing photons in memory, which in fact can \emph{increase} the distinguishability error rate, as discussed above. 
On the other hand, the new $F_n$ protocol instead requires around $4n$ photons for the same amount of error reduction, \emph{without} requiring memory or active feed-forward. 
For example, to reduce the error rate by a factor of $n=81$, the protocol of \cite{marshall_distillation_2022} requires $4$ rounds of the $F_3$ distillation scheme, resulting in the use of 
$n^2 = 6561$ photons. 
The present work requires approximately $4n=324$ photons, using a single $F_{81}$ protocol. This is more than $20$ times more efficient than the iterated $F_3$ protocol. 
In fact, if one is willing to use $6561$ photons, the number used by the iterated $F_3$ protocol, they may instead use a single iteration of $F_{1640}$, reducing the error rate by a factor of approximately $1640$ without increasing the cost. (And note that the difficulty of implementation is potentially \emph{decreased}, since memory and feed-forward are not required.) 
Thus, given a desired level of performance, the new protocols are both dramatically more efficient and simpler to implement than the work of \cite{marshall_distillation_2022}, which in turn was far more efficient than \cite{sparrow_quantum_2017}. 
\end{remark}

We recall, however, that our protocols for large $n$ have an ever-decreasing \emph{threshold} after which the protocol fails to reduce the distinguishability error rate. 
(See Figure~\ref{fig:epsilon_critical}.) 
In particular, in the discussion above, one must take $\epsilon$ sufficiently small so that it is below the threshold for all relevant values of $n$. 
This limitation does not apply to the iterated protocol of \cite{marshall_distillation_2022}, as the relevant threshold is that of the $3$-photon protocol. 
We note that these thresholds are not a major concern for moderately sized protocols: as discussed above, we expect realistic experimental protocols to have $\epsilon\leq 0.15$ or so, below threshold for all $n\leq 16$, as seen in Figure~\ref{fig:epsilon_critical}. 

However, even below the threshold, we may not always want to choose the largest value of $n$ available. 
This is discussed in Section~\ref{sec:numerics}, especially Figures~\ref{fig:opt_n} and \ref{fig:opt_n_both}, where, for each $\epsilon$, we numerically identify the value of $n$ with smallest output error rate $e_n(\epsilon)$. 
If the initial error rate $\epsilon$ is large, we note that one may iterate distillation protocols with gradually increasing numbers of photons. 
The first rounds will have small $n$, chosen according to the figures, allowing for the reduction of the distinguishability error rate. 
Once the error rate is sufficiently small, one may use larger protocols in later rounds, further decreasing the error rate at a lower cost. 
(However, as discussed above, iterative protocols come at the cost of requiring active feed-forwarding and quantum memory during distillation.) 

We briefly consider some practical examples. 
Recall that, from visibility estimates in the literature \cite{halder_high_2008, tsujimoto_high_2017,tambasco_quantum_2018, wang_research_2019, ollivier_hong-ou-mandel_2021, somhorst_quantum_2023, alexander_manufacturable_2024}, we approximate that the initial error rate $\epsilon_0$ is likely in the range $[0.005, 0.15]$. 
We first consider the high end of this range. 
Given initial $\epsilon_0 = 0.15$, applying the protocol with $U=F_6$ will give output error rate $\epsilon_0'= e_6(\epsilon_0)\approx 0.056$ in the OBB model. 
(Calculated using the exact formula for $e_6(\epsilon)$ in Appendix~\ref{sec:numerics appendix}.) 
If a smaller error rate is required, we may then insert the output photons with error rate $\epsilon_0'$ into the $U=F_{12}$ protocol, giving output error rate $\epsilon_0''= e_{12}(\epsilon_0')\approx 0.0097$. 
We now consider the low end of the range, corresponding to very small error rate $\epsilon_0=0.005$, roughly corresponding to the visibility found in \cite{alexander_manufacturable_2024}. 
Of the protocols we considered, the optimal one in the OBB model would be $U=H_{16}$. After one iteration, we obtain output error rate $\epsilon_0' \approx 0.0003565$, an improvement by a factor of around $14$. 
Thus distillation protocols can have a significant impact in both the low and high error regimes. 

\subsection{Symmetries and conjectures}\label{sec:symmetry body}
In this section, we briefly discuss the importance of certain symmetry properties to our results. 
We will return to this subject in much greater detail in Appendix~\ref{sec:symmetry}. 
Further, in Appendix~\ref{sec:simulation appendix}, we discuss how these and other symmetries can be leveraged to greatly reduce the complexity of calculating the $h_n(\Phi_k)$ and related quantities via simulation. 

We consider a distillation protocol determined by $n\times n$ matrix $U$. 
We refer to a \emph{mode symmetry} of $U$ to be an $n\times n$ permutation matrix $P$ such that there exists an $n\times n$ diagonal matrix $D$ with $UP = DU$. In a boson sampling experiment, this corresponds to the observation that permuting the \emph{modes} (not the photons) of the input state is equivalent to applying phase shifts on the output state. 
Let $G$ be an abelian group of such mode symmetries for $U$. 
In other words, $G$ is an abelian permutation group that is simultaneously diagonalized by $U$. 
One can show that such symmetries lead to suppression laws \cite{dittel2018totally}; in other words, we obtain a set $\mc{S}_G$ of patterns $(s_0, \dots, s_{n-1})$, determined by the symmetry group $G$, such that 
\begin{equation}\label{eq:suppression}
    \bra{s_0, \dots, s_{n-1}}\hat{U}\ket{1, \dots, 1}\neq 0 \implies (s_0, \dots, s_{n-1})\in\mc{S}_G.
\end{equation}
In a distillation setting, recalling \eqref{def:ideal pattern}, we see that the \emph{ideal patterns} must be a subset of $\mc{S}_G$. 
As discussed above, in the Hadamard case, the ideal patterns may be exactly identified as the elements of $\mc{S}_G$ (satisfying our additional restriction $s_0=1$) \cite{crespi2015suppression, dittel2018totally}. 
In the Fourier case, we say that patterns in $\mc{S}_G$ satisfy the \emph{Zero Transmission Law} (ZTL), and
the ideal patterns are in general a proper subset of the ZTL patterns \cite{tichy2010zero}. The first non-ideal pattern satisfying the ZTL occurs for $n=6$ (see Appendix~\ref{sec:fourier symm}). 
In general, not all suppression laws may be explained by symmetry conditions such as the ZTL \cite{bezerra2023families}: we expect that the extra suppression for $U=F_6$ is such a case. 
In Appendix~\ref{sec:fourier prime power}, we prove the following theorem: 
\begin{theorem}\label{thm:suppression body}
If $n$ is a prime power, the suppression laws for $U=F_n$ are exactly characterized by the Zero Transmission Law. In particular, the ideal patterns $(s_0, \dots, s_{n-1})$ for $U=F_n$ are precisely the elements of $\mc{S}_G$ with $s_0=1$. 
If $n>1$ is not a prime power, we may explicitly exhibit a pattern in $\mc{S}_G$ with $s_0=1$ that is \emph{not} an ideal pattern. 
\end{theorem}
We also discuss analogues of the ZTL for general $F_{(n_1, \dots, n_\ell)}$ in Appendix~\ref{sec:fourier abelian}. 

Beyond constraining the ideal patterns, these symmetries may be used to simplify calculations for heralding and output error rates. In particular, we have the following, proven in Appendix~\ref{sec:symmetry}. 

\begin{theorem}\label{thm:symm body}
Let $U$ be an $n\times n$ unitary with corresponding symmetry group $G$ and set $\mc{S}_G$ of symmetry-preserving patterns. 
Let $\ket{\eta}\in \hn$ be an arbitrary pure state, with normalized projection $\ket{\eta_+}$ onto the space of $G$-symmetric states. 
We have the following results. 
\begin{enumerate}
    \item Overlap with states in $\mc{S}_G$ may be calculated in terms of the symmetrized state  $\ket{\eta_+}$:\\
    $\bra{s_0, \dots, s_{n-1}}\hat{U}\ket{\eta} = \braket{\eta_+}{\eta}\bra{s_0, \dots, s_{n-1}}\hat{U}\ket{\eta_+}$. \label{step:overlap}
    \item \label{step:symm} The heralding and error rates may be calculated in terms of the symmetrization:
    \begin{align}
        h_n(\ketbra{\eta}) &= |\braket{\eta_+}{\eta}|^2 h_n(\ketbra{\eta_+}),
        \\ e_n(\ketbra{\eta}) &= e_n(\ketbra{\eta_+}).
    \end{align}
    \item Further assume that $U = H_n$, where $n$ is a power of $2$, or $U=F_n$, where $n$ is a prime power. 
    \label{step:symm2}
    Then 
    \begin{equation}\label{eq:herald symm norm}
        h_n(\ketbra{\eta_+}) = |\left(\bra{1}\otimes I\right)\hat{U}\ket{\eta_+}|^2.
    \end{equation}
    In particular, this gives a method for calculating heralding rates $h_n(\ketbra{\eta})$ while only checking whether output patterns have $1$ photon in the output mode, rather than verifying whether the entire pattern is ideal. 
    In Appendix~\ref{sec:app error}, we express $e_n(\ketbra{\eta})$ in a similar form involving only measurements on the $0$th mode. 
\end{enumerate}
\end{theorem}

\begin{remark}\label{remark:fourier prime assumption}
We briefly comment on the assumption on $U$ in part \ref{step:symm2} above. 
More generally, we can replace this with the assumption that all patterns $(s_0, \dots, s_{n-1})$ in $\mc{S}_G$ with $s_0 = 1$ and $\bra{s_0, \dots, s_{n-1}}\hat{U}\ket{\eta}\neq 0$ are ideal. 
As discussed above, this is automatically satisfied for $U$ as in the theorem, as all patterns in $\mc{S}_G$ with $s_0=1$ are ideal. 
For $\ket{\eta} = \ket{1, \dots, 1}$ perfectly indistinguishable, \eqref{eq:herald symm norm} follows from the definition of ideal pattern regardless of the unitary $U$, allowing for the straightforward computation of $h_n(0)$ as in Appendix~\ref{sec:ideal herald}. 
Further, we show in Appendix~\ref{sec:one} that we may use the indistinguishable case to calculate $h_n(\Phi_1)$ for any $U=F_{(n_1, \dots, n_\ell)}$
regardless of the value of $n=n_1\cdots n_\ell > 2$. 
In the Fourier case with $n$ not a prime power, one may modify Protocol~\ref{proto:distillation} to post-select for any patterns in $\mc{S}_G$ with $s_0=1$ rather than the smaller set of ideal patterns, so that the general assumption given above holds by definition. This intentionally increases the output error rate, but allows for the application of Theorem~\ref{thm:symm body} Part~\ref{step:symm2} and makes the behavior of the protocol more predictable for all $n$. 
Noting Figure~\ref{fig:opt_n}, however, we advise against doing this, as the Fourier protocols with $n$ not a prime power generally seem to be optimal without this modification. 
\end{remark}

These symmetry considerations are essential in the proofs given in the Appendix. 
Further, they lead to the following conjectures: 

\begin{conjecture}\label{conj:herald term bound}
As above, assume that 
$U = H_n$, where $n$ is a power of $2$, or $U=F_n$, where $n$ is a prime power. 
Let $\ket{\eta}, \ket{\Delta}$ be states of the form \eqref{eq:internal fock experiment}, where $\ket{\eta}$ has $k$ distinguishability errors and $|\Delta \rangle$ has $n$ (so that in the latter case, all photons are mutually fully distinguishable). 
Let $\ket{\eta_+}, \ket{\Delta_+}$ be the $G$-symmetrizations, as above. Then 
\begin{equation}
    h_n(\Phi_0)\leq h_n(\ketbra{\eta_+})\leq h_n(\ketbra{\Delta_+}).
\end{equation}
In particular, 
\begin{equation}
    \frac{1}{n}h_n(\Phi_0)\leq h_n(\Phi_k)\leq \dfrac{\textnormal{gcd}(k,n)}{n} h_n(\ketbra{\Delta_+}).
\end{equation}
\end{conjecture}
The heart of this conjecture is that for $G$-symmetrized states, more mode symmetry seems to \emph{decrease} the heralding rate. 
Thus the lower bound comes from the fully mode-symmetric state $\Phi_0 = \ketbra{1, \dots, 1}$ and the upper bound from a state $\ket{\Delta}$ with as little symmetry as possible. 
The coefficients come from Theorem~\ref{thm:symm body} Part \ref{step:symm} and trivial bounds on  $|\braket{\eta_+}{\eta}|^2$ when $\ket{\eta}$ has $k$ distinguishability errors. 

For $F_n$ where $n$ is not a prime power, Conjecture~\ref{conj:herald term bound} fails, with such protocols tending to have lower heralding rate than predicted: for example, in the OBB model we have $h_6(\Phi_2)/h_6(\Phi_0)\approx 0.132 < 1/6$. 
However, the conjecture seems to hold for general $F_n$ if we modify the protocol as suggested in Remark~\ref{remark:fourier prime assumption}. 
With this modification, we would for example have $h_6(\Phi_2)/h_6(\Phi_0)\approx 0.216 > 1/6$, and the $n=6$ case would then fit into the same patterns as the prime-power cases. (But this would negatively impact the error rate, e.g., increasing $e_6(\Phi_2)$ from $0.017$ to $0.026$). 

As a corollary of the previous conjecture, we obtain the following: 

\begin{conjecture}\label{conj:herald lower bound}
For the $n$-photon Fourier or Hadamard protocols, with $n$ a prime power as above, we have
\begin{equation}
    h_n(\epsilon)\geq h_n(0)\left((1-\epsilon)^n\left(1-\frac{1}{n}\right) + \frac{1}{n}\right).
\end{equation}
\end{conjecture}
In particular, if $\epsilon \leq \frac{1}{n}$, by Theorem~\ref{thm:herald ideal} we have
\begin{equation}
    \lim_{n\rightarrow\infty} h_n(\epsilon)\geq \lim_{n\rightarrow\infty}h_n(0) \left(\left(1+\frac{1}{n}\right)^{n+1} + \frac{1}{n}\right)= \frac{1}{4e}\approx 0.092.
\end{equation}
Then assuming the conjecture holds, we see that as long as $\epsilon$ is small relative to $n$, the heralding rates cannot get too small even for large $n$. 
Further, if Conjecture~\ref{conj:herald} holds, we have for all $n\geq 3$
\begin{equation}
    h_n(\epsilon)\geq \frac{1}{4e}.
\end{equation}

Recall from our earlier discussion of error thresholds that choosing $n$ and $\epsilon$ appropriately is also important to ensure that the distillation protocols reduce error rates. 
Here we note that the numerical threshold values obtained in Figure~\ref{fig:epsilon_critical} are all far greater than $1/n$. 
Thus it appears that if $\epsilon < 1/n$, we obtain both a large heralding rate and a distillation protocol that improves the output error rate.

\subsection{Numerics}\label{sec:numerics}
In this section, we give further numerics for our protocols. The details of how our numerical simulations are performed is discussed in App.~\ref{sec:simulation appendix}, with data provided in App.~\ref{sec:numerics appendix}.

First, we consider the question of optimal distillation protocols given an error rate $\epsilon$. In particular, as discussed in Section~\ref{sec:iterated}, we expect the heralding rates of our protocols to be reasonably large, with linearly growing resource requirements. 
The limitation is instead the growth of the error rate: protocols with large $n$ give the smallest output error rate when $\epsilon$ is small, but perform worse for large $\epsilon$. 
Thus, in Figure~\ref{fig:opt_n}, we consider all protocols with $U=F_n$, $3\leq n\leq 16$, and for each $\epsilon$ determine the value of $n$ with $e_n(\epsilon)$ as small as possible. 
We are most interested in $\epsilon\leq 0.15$, as discussed above. 
As expected, we find that the largest protocol is optimal for small $\epsilon$, and as $\epsilon$ grows the optimal protocol becomes smaller and smaller. 
Further, in the OBB model, there is a large plateau for $\epsilon$ roughly between $0.1$ and $0.4$ in which the $6$-photon protocol is optimal. 
We expect this is due to the fact for $n=6$, the set of ideal patterns is strictly smaller than the set of symmetry-preserving patterns, as discussed in Section~\ref{sec:symmetry body} and Appendix~\ref{sec:fourier symm}. 
In fact, for both error models and reasonably small values of $\epsilon$, 
the dominant protocols seem to have $n=6,12,15$, all of which are not prime powers and have a set of ideal patterns strictly smaller than the set of symmetry-preserving patterns \cite{tichy_sampling_2015} (recall Theorem~\ref{thm:suppression body}). 
(Note we exclude $n=16$ from the discussion here, as by Theorem~\ref{thm:error rate} it is guaranteed to be optimal for sufficiently small $\epsilon$.) 
The strong performance of $n=6,12,15$ is unsurprising, since a smaller set of ideal patterns leads to a smaller heralding rate: see the discussion below Conjecture~\ref{conj:herald term bound}. 
As explained in more detail below, this should generally reduce the error rate as well. 
We also note that these seemingly optimal protocols have $n$ divisible by $3$. 
Understanding this divisibility condition may explain why there is no point at which $n=14$ is preferable to $n=12$ or $n=15$, even though all three have fewer ideal patterns than symmetry-preserving patterns. 
In Figure~\ref{fig:opt_n_both}, we repeat this analysis while including the Hadamard protocols as well. 
We find that for reasonably small values of $\epsilon$, the same Fourier protocols ($n$ not a prime power and divisible by $3$) are generally preferred. 
(The fact that $H_{16}$ is preferred over $F_{16}$ in the OBB model seems to be related to the discussion of the $n=8$ case below.) 
We conjecture that if we included all $n\leq 18$ in the analysis, we would see the optimal protocol transition directly from $F_{18}$ to $F_{15}$ or $F_{12}$. 

\begin{figure}
    \centering
    \includegraphics[width=\columnwidth]{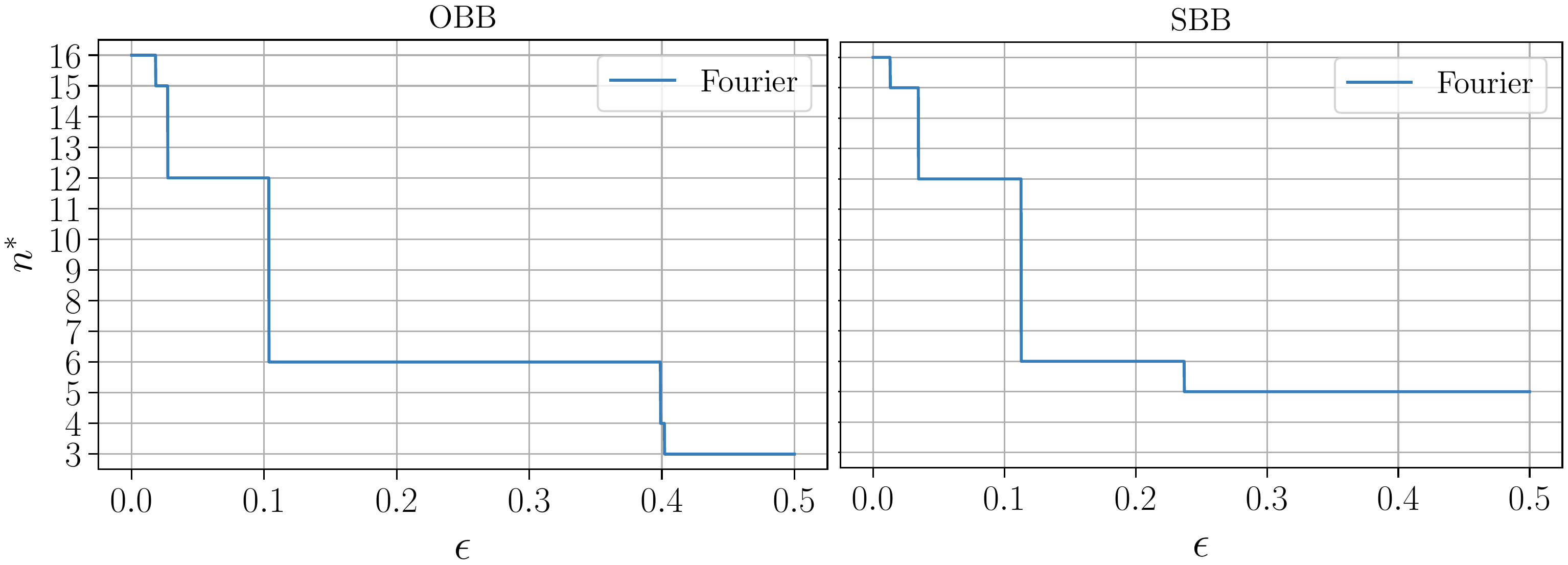}
    \caption{Given an initial error rate $\epsilon$, we plot the value $n^*$ of $n$ for which the protocol with $U=F_n$ gives the greatest error reduction (this does not take into account heralding rates or expected resource counts). We considered protocols up to size $n=16$. As expected, for small $\epsilon$ we have $n^*=16$, since $e_n(\epsilon)\approx \epsilon/n$.
    }
    \label{fig:opt_n}
\end{figure}

\begin{figure}
    \centering
    \includegraphics[width=0.5\columnwidth]{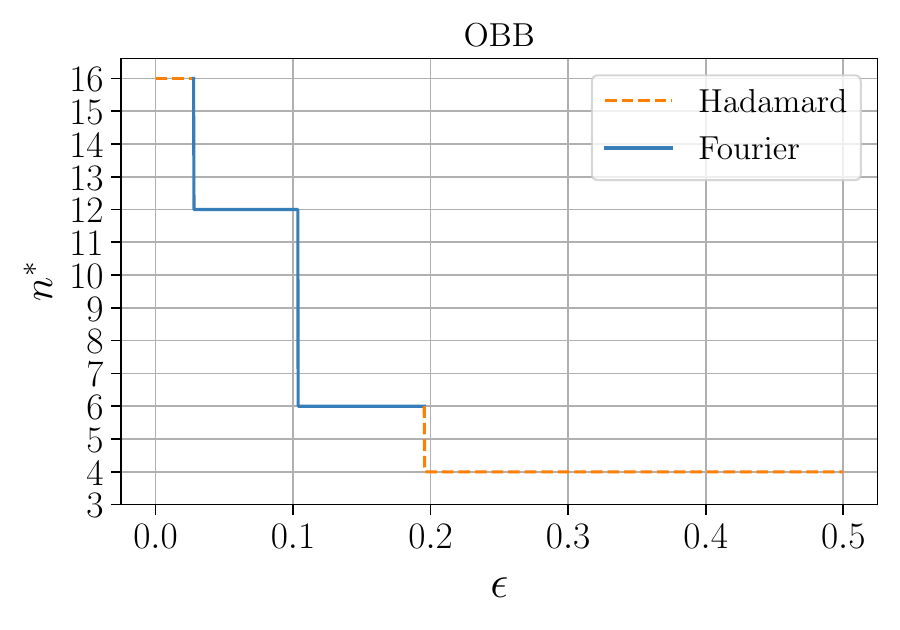}
    \caption{We extend the plot of Figure~\ref{fig:opt_n} (OBB) to include the Hadamard protocols $U=H_4, H_8, H_{16}$. 
    The different line styles indicate whether the protocol of interest is Fourier or Hadamard. Note that in the SBB case, the plot is unchanged from the one in Figure~\ref{fig:opt_n} (i.e., for the SBB error model, the Fourier protocol always has the greatest error reduction).}
    \label{fig:opt_n_both}
\end{figure}

In the preceding analysis, we compared protocols with different values of $n$. 
Next, we fix $n=8$ and numerically investigate the heralding and error rates of the distillation protocols using different unitaries and error models. 
We choose to study $n=8$ here because it is a power of $2$, allowing for comparison of corresponding Fourier and Hadamard distillation protocols. 
(We include a third protocol corresponding to $F_{(4,2)}$ in Appendix~\ref{sec:fourier abelian}.) 
The $n=4$ case was originally presented in \cite{marshall_distillation_2022}, 
so $n=8$ is the smallest new case relevant to both protocols. 
We note that similar results to the below hold for the $n=4$ case. 
For this discussion, we note the following formula for the output error rate, given in Appendix~\ref{sec:app error}: 
\begin{equation}\label{eq:error explicit}
    e_n(\epsilon) = \dfrac{\ole_n(\epsilon)}{h_n(\epsilon)} = \dfrac{\sum_{k=1}^n \binom{n}{k}\epsilon^k (1-\epsilon)^{n-k} e_n(\Phi_k)h_n(\Phi_k)}{\sum_{k=0}^n \binom{n}{k}\epsilon^k (1-\epsilon)^{n-k} h_n(\Phi_k)}. 
\end{equation}
We recall that $e_n(\Phi_k)$ is a conditional probability: the probability that, if a state with $k$ distinguishability errors leads to an ideal pattern, it outputs a non-ideal photon. 
Recalling Theorems \ref{thm:error rate}, \ref{thm:herald first order}, and \ref{thm:herald ideal}, the quantities $h_8(\Phi_0), h_8(\Phi_1), e_8(\Phi_1)$ are fixed and we have the first-order approximations 
\begin{equation}
    e_8(\epsilon) = e_8(\Phi_1)\epsilon + O(\epsilon^2), \,\,\,\, h_8(\epsilon) = h_8(\Phi_0)(1-7\epsilon) + O(\epsilon^2),
\end{equation}
regardless of the error model or the choice of $U\in\{F_8, H_8\}$. 
Then we expect much of the difference between protocols and models to be explainable via the second-order terms, which are determined by $h_8(\Phi_2)$, $e_8(\Phi_2)$ in addition to the above fixed quantities. 
In particular, $e_8(\Phi_2)$ and $h_8(\Phi_2)$ affect the numerator of \eqref{eq:error explicit} to second order; $h_8(\Phi_2)$ also affects the denominator, but only to third order. 
We give values for these second-order quantities in Table~\ref{table:8}. 
One may think of $\Phi_2$ in the OBB model as an average of all states with $2$ ``different" errors; in the SBB model, $\Phi_2$ is an average of states with $2$ of the ``same" error. 

\begin{table}[h]
\begin{tabular}{|l|l|l|}
\hline
              & $h_8(\Phi_2)$ & $e_8(\Phi_2)$ \\ \hline
Fourier, OBB  & $0.040$       & $0.377$       \\ \hline
Fourier, SBB  & $0.042$       & $0.328$       \\ \hline
Hadamard, OBB & $0.038$       & $0.338$       \\ \hline
Hadamard, SBB & $0.076$       & $0.338$       \\ \hline
\end{tabular}
\caption{We plot the values of $h_8(\Phi_2)$ and $e_8(\Phi_2)$ for the variant protocols and error models considered here. These are the dominant values that differ between cases, and they determine the behavior of $h_8(\epsilon)$, $e_8(\epsilon)$ up to second order. \label{table:8}}
\end{table}

In Figure~\ref{fig:n_8_dist}, we plot the heralding rates and output error rates in the Fourier and Hadamard cases, using both the OBB and SBB error models. 
We are particularly interested in the small $\epsilon$ regime. 
We note from the plot that $3$ of the $4$ cases considered have very similar heralding rates $h_8(\epsilon)$. 
The Hadamard case with SBB error model, however, has a significantly larger heralding rate. 
This is clearly reflected in the values of $h_8(\Phi_2)$ in Table~\ref{table:8}, with the Hadamard SBB case nearly doubling the other values. 
Since the only second-order effect of $h_8(\Phi_2)$ on \eqref{eq:error explicit} is in the numerator, an increase in this quantity should increase the overall error rate. 
This explains why this case has significantly worse error rate than the others, as seen in Figure~\ref{fig:n_8_dist}. Intuitively, the Hadamard protocol does not notice two of the ``same" (SBB) errors as often, so we are more likely to have distinguishable output photons. 

\begin{remark}
We reiterate the above observation: even though a large heralding rate reduces the overall resource requirements, it can also lead to larger error rates. This is what makes our protocols notable, especially as $n$ increases: they maintain a nontrivial heralding rate while arbitrarily decreasing the error rate. 
\end{remark}

Next, we fix an error model and compare error rates $e_n(\epsilon)$ for $F_n$ and $H_n$. 
For SBB, the Fourier case $U=F_8$ has smaller output error rate, translating to better performance. 
This is explained by the analysis above. 
When we consider the OBB error model, however, we see that the Hadamard protocol $U=H_8$ has the smaller error rate and better performance, with both $h_8(\Phi_2)$ and $e_8(\Phi_2)$ decreasing in this case. 
Then for a \emph{fixed} $n=2^r$, we cannot conclude whether Fourier or Hadamard unitaries are preferable in general: this depends strongly on the error model relevant to the experiment in question. 
Note that this is in contrast with the discussion of Figures~\ref{fig:opt_n} and \ref{fig:opt_n_both} above, in which we found that Fourier protocols often perform better when we are free to choose the value of $n$. 

Similarly, we may fix the unitary (Fourier or Hadamard) and compare the error rates $e_n(\epsilon)$ under the two error models. From Figure~\ref{fig:n_8_dist}, we see that the Hadamard protocol performs better with OBB errors, whereas the Fourier protocol performs better with SBB errors. 
In the Hadamard case, this may be explained by the increased heralding rate under the SBB model, as discussed above. 
For $U=F_8$, however, the heralding rate does not significantly differ between the two error models. 
(In fact, the protocol with smaller error rate $e_n(\epsilon)$ has slightly \emph{larger} heralding rate!) 
Then the discrepancy is caused by the change in $e_8(\Phi_2)$, which is proportionally more significant. 
Intuitively, for the $F_8$ protocol, states with $2$ 
errors give ideal output patterns at roughly the same rate regardless of error model, but when heralding occurs and the errors are ``different," the error photons are more likely to end up in the output mode. 

In summary, we note that Fourier protocols, especially those with $n$ divisible by $3$ and not a prime power, seem to be the most effective at reducing the error rate. 
For a fixed $n = 2^r$, however, whether $F_n$ or $H_n$ performs better seems to be highly dependent on the error model. 
Finally, the Hadamard protocols seem to struggle with filtering out SBB errors, while the Fourier protocols are equally capable of filtering out both types but, when heralding occurs, are more likely to direct OBB errors to the output mode. 

\begin{figure}
    \centering
    \includegraphics[width=\columnwidth]{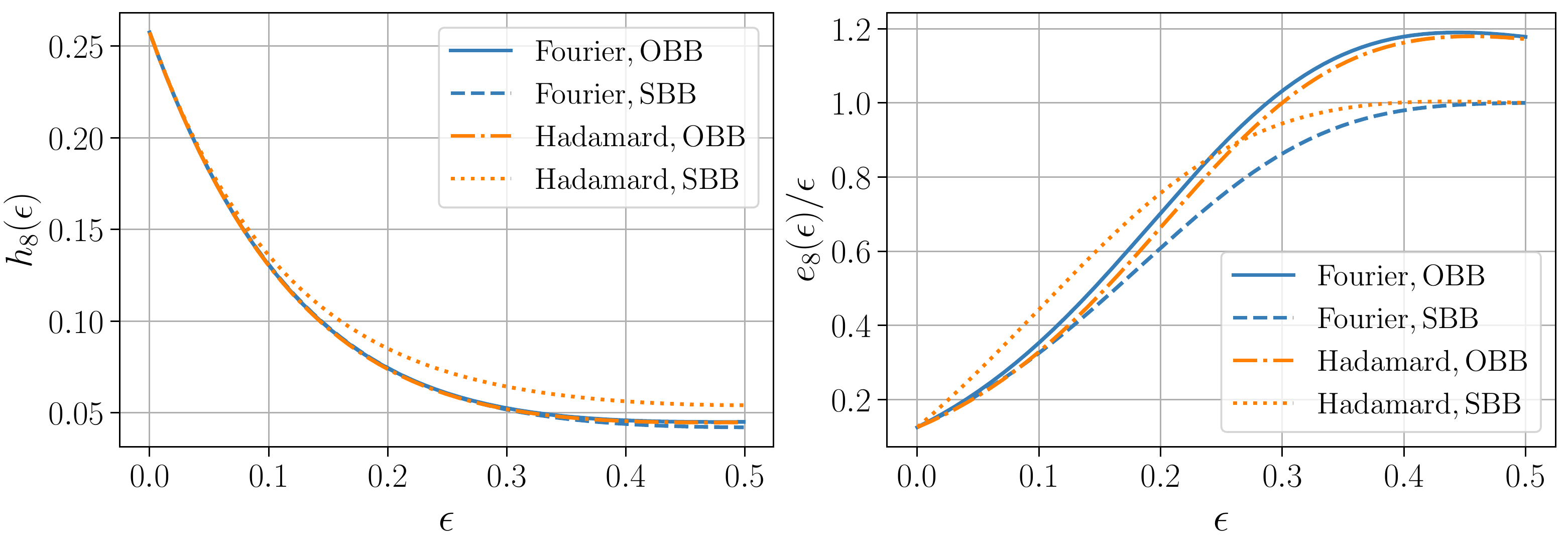}
    \caption{Heralding rate (left) and output error rate (right) for $n=8$ Fourier and Hadamard protocols, for both OBB and SBB noise models.}
    \label{fig:n_8_dist}
\end{figure}

\subsection{Incorporation of photon loss}\label{sec:loss}
We now discuss the incorporation of photon loss models into our analysis of the Fourier and Hadamard distillation protocols. 
We begin with a brief discussion of loss models. 

We initially consider the \emph{uniform beamsplitter loss model}, in which our linear optical circuit is made up of lossy beamsplitters, and each photon has probability $\lambda$ of loss at each beamsplitter. 
We model this as a photon loss channel with loss probability $\lambda$, occurring independently on each of the two relevant modes, just before the beamsplitter. 
We recall, however, that symmetric loss commutes through linear optics \cite{oszmaniec_classical_2018}. 
For the unitary $U=H_n$, we give an explicit recursive circuit decomposition in Figure~\ref{fig:hadamard-circuit}, and note that the described loss channels may be commuted to the end of the circuit. 
In particular, since each path from an input mode to an output mode involves exactly $\log_2 n$ beamsplitters, the loss probability on each mode is $\Lambda =1-(1-\lambda)^{\log_2 n}$. 
For simplicity, we consider the same loss model in the Fourier case. 
This may likely be derived as above; for example, when $n$ is a power of $2$, a circuit decomposition similar to that in Figure~\ref{fig:hadamard-circuit} is given in \cite{barak2007quantum}, and the same analysis holds. 
Then, going forward, we will use the \emph{simplified loss model} depicted in Figure~\ref{fig:loss-model}, in which partially distinguishable photons are given as input to the circuit, the linear optical unitary $\hat{U}$ is applied, each photon independently undergoes a loss channel with probability $\Lambda = 1-(1-\lambda)^{\log n}$, and the final $n-1$ modes are measured. 
We note that in this model, the probability that no photons are lost is $(1-\Lambda)^n = (1-\lambda)^{n\log n}$.

\begin{figure}
    \centering
    \includegraphics[width=0.65\columnwidth]{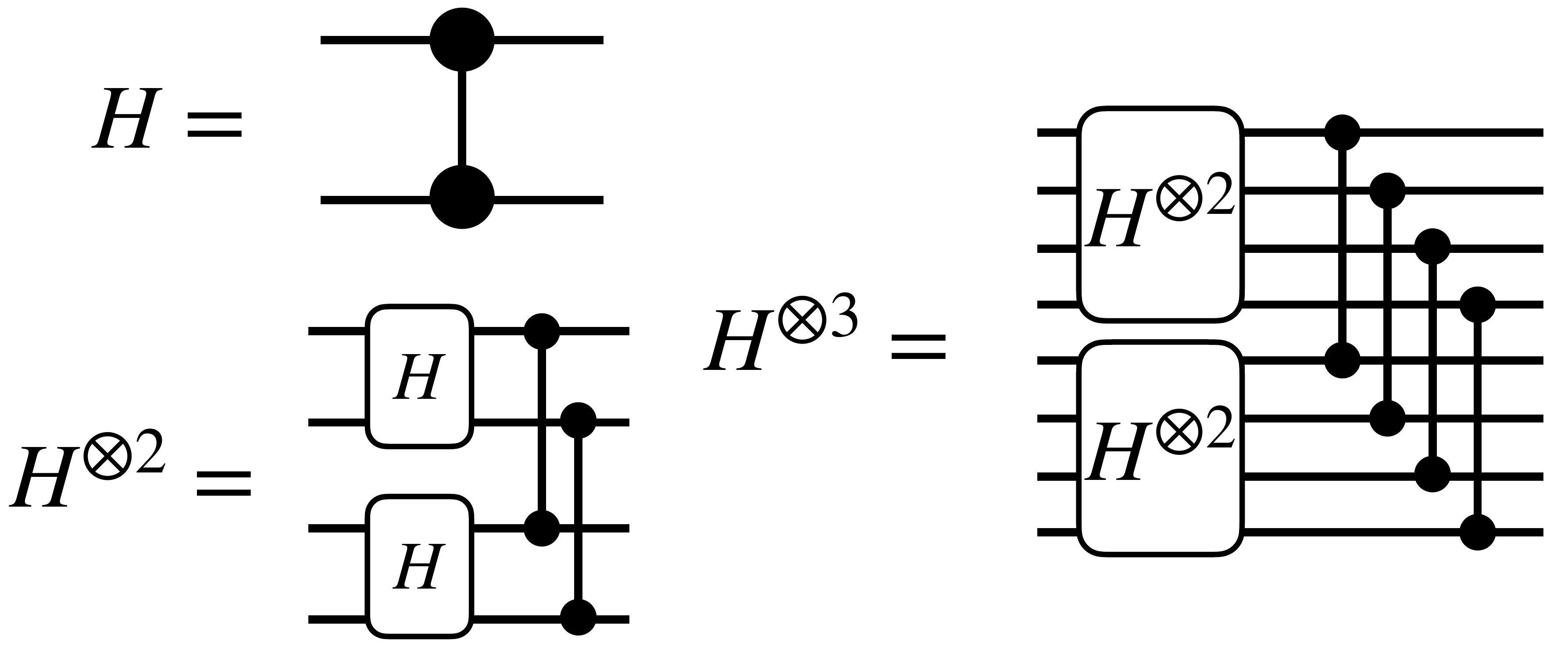}
    \caption{Illustration of construction of the circuit implementing the $n\times n$ Hadamard unitaries $H_{n} = H^{\otimes r}$, for $r=1,2,3$, where $n=2^r$. It generalizes in the obvious way. The circuit components are 50:50 beamsplitters. One can see directly that each mode interacts with $r$ beamsplitters.}
    \label{fig:hadamard-circuit}
\end{figure}

\begin{figure}
    \centering
    \includegraphics[width=0.65\columnwidth]{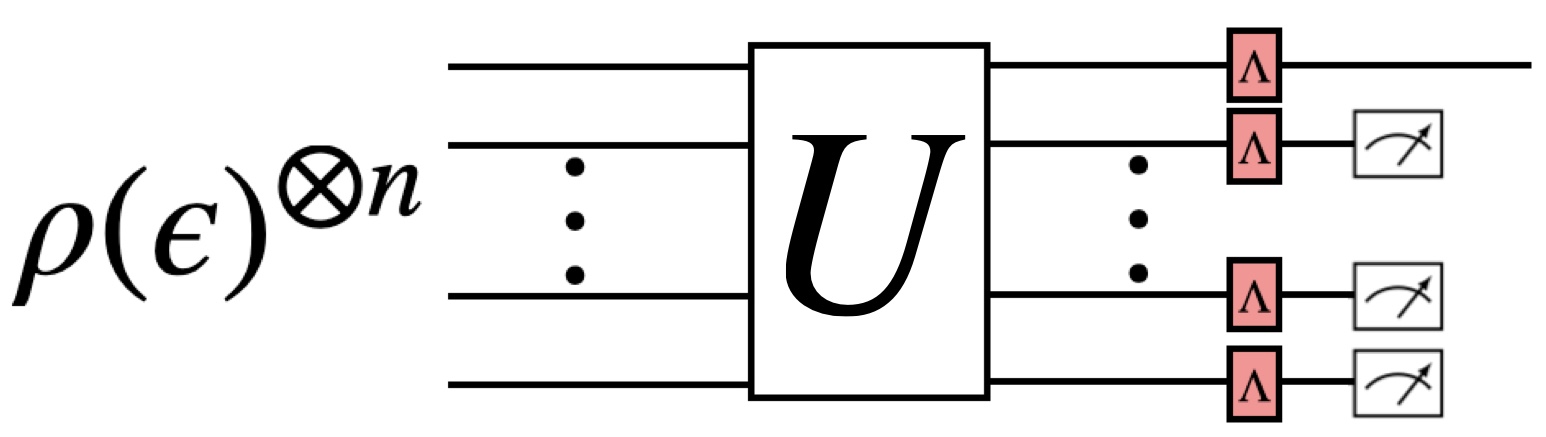}
    \caption{Illustration of our error model. As in the rest of the paper, $n$ partially distinguishable photons, each with internal state $\rho(\epsilon)$, are input into $n$ different external modes, and we apply the linear optical unitary $\hat{U}$ corresponding to an $n\times n$ unitary $U$. Unique to this section, we then apply a photon loss channel, with loss probability $\Lambda = 1-(1-\lambda)^{\log n}$, to each mode. We then perform PNRD on the final $n-1$ modes and post-select as usual. 
    For the Hadamard circuits of Figure~\ref{fig:hadamard-circuit}, this is equivalent to the uniform beamsplitter loss model with probability $\lambda$ of loss at each beamsplitter. 
    }
    \label{fig:loss-model}
\end{figure}

We begin by discussing the heralding rate, which we now write as $h_n(\epsilon; \lambda)$ to allow for dependence on the loss rate, with $h_n(\epsilon) = h_n(\epsilon; 0)$. 
Heralding requires the detection of $n-1$ photons in the latter $n-1$ modes; in particular, heralding is only possible if at most $1$ photon has been lost. 
We have the corresponding decomposition
\begin{equation}\label{eq:lossy decomposition eq}
    h_n(\epsilon; \lambda) = (1-\Lambda)^n h_n(\epsilon) + n\Lambda (1-\Lambda)^{n-1}h_n^{(1)}(\epsilon),
\end{equation}
where $h_n(\epsilon) = h_n(\epsilon; 0)$ encapsulates the heralding probability when no photons are lost and $h_n^{(1)}(\epsilon)$ is the heralding probability when exactly $1$ photon is lost (as mentioned above, $h_n^{(k)}(\epsilon)=0$ for $k>1$). 
Note that we have two possibilities for pre-loss output patterns contributing to $h_n^{(1)}(\epsilon)$: either we have an ideal pattern and the single photon in mode $0$ is lost, which occurs with probability $\frac{1}{n}h_n(\epsilon)$; or we have a \emph{non-ideal} pattern (with no photons in mode 0), and a photon in a nonzero mode is lost, and the resulting pattern on the final $n-1$ modes matches an ideal pattern. We write the latter probability as $c_n(\epsilon)$. In either case, no photons exit mode 0. 
From this description, we decompose 
$h_n^{(1)}(\epsilon)= \frac{1}{n}h_n(\epsilon) + c_n(\epsilon)$, 
and substituting into \eqref{eq:lossy decomposition eq} gives
\begin{equation}\label{eq: lossy decomposition detailed}
    h_n(\epsilon; \lambda) = (1-\Lambda)^{n-1} h_n(\epsilon) + n\Lambda (1-\Lambda)^{n-1}c_n(\epsilon).
\end{equation}

\begin{theorem}\label{thm:loss}
Let all notation be as above. For the 
protocol corresponding to $U=F_{(n_1, \dots, n_\ell)}$, $n=n_1\cdots n_\ell$, we have
\begin{equation}\label{eq:lossy decomposition ineq}
    h_n(\epsilon; \lambda)\geq (1-\Lambda)^{n-1}h_n(\epsilon) = (1-\lambda)^{(n-1)\log n} h_n(\epsilon).
\end{equation}
If $\epsilon=0$, the bound is tight: 
\begin{equation}
    h_n(0; \lambda)= (1-\Lambda)^{n-1} h_n(0) = (1-\lambda)^{(n-1)\log n} h_n(0).
    \label{eq:loss_herald_no_err}
\end{equation}
\end{theorem}
\begin{proof}
The first claim follows by bounding $c_n(\epsilon)\geq 0$ in \eqref{eq: lossy decomposition detailed}. 
For the second claim, we argue that $c_n(0) = 0$. 
We recall that when $\epsilon=0$, the (pre-loss) output patterns satisfy certain symmetry constraints, given in Section~\ref{sec:distillation}. 
In particular, in the notation used above, we have $\sum_i g_i \equiv 0 \mod n$ in the Fourier case and $\bigoplus_i g_i = 0$ in the Hadamard case. 
(We have similar restrictions for general $F_{(n_1, \dots, n_\ell)}$, discussed in Appendix~\ref{sec:fourier abelian}.) 
For a non-ideal pattern to contribute to $c_n$, there must exist two symmetry-preserving patterns of $n$ photons in $n$ modes, one ideal and one non-ideal, both of which may give the same $(n-1)$-photon pattern after a loss. 
The symmetry conditions given above make this impossible: given a pattern of $n-1$ photons in $n$ modes, there is a unique mode $j$ such that adding a photon to mode $j$ results in a symmetry-preserving pattern. 
This is clear for the Fourier and Hadamard cases given above and is discussed for the general case in Appendix~\ref{sec:fourier abelian}.
\end{proof}
This may be used to obtain an upper bound on the resources required for distillation in the presence of loss, namely
\begin{equation}\label{eq:lossy resource estimate}
    \dfrac{n}{h_n(\epsilon; \lambda)} \leq \dfrac{n}{(1-\lambda)^{(n-1)\log n}h_n(\epsilon)}\approx \dfrac{4n}{(1-\lambda)^{(n-1)\log n}},
\end{equation}
where the approximation follows from Theorem~\ref{thm:herald ideal}, as in Theorem~\ref{thm:resources}. 
With a nonzero loss rate $\lambda$ per beamsplitter, the heralding rate is exponentially suppressed in $n\log n$, leading to exponentially growing resource requirements. 
Given a fixed amount of resources and error rates $\epsilon,\lambda$, this effectively puts an upper bound on the maximum feasible size $n$ of a distillation protocol. 
For moderate values of $n$, however, the effect is not so drastic. 
We show numerical scaling of the heralding probability in Fig.~\ref{fig:herald_loss} for $n=8$, where we see that even a high loss rate of $\Lambda = 0.1$ only cuts the heralding probability in half.
As a further example, we consider distillation with $n=16$. 
Without loss, distillation would require $4$ runs on average, with a total expected cost of approximately $64$ photons. 
With a loss rate of $\lambda = 0.01$ ($\Lambda \approx 0.04$), 
the above implies that distillation would require an average of $7.3$ runs, with a total expected cost of approximately $117$ photons. 

\begin{figure}
    \centering
    \includegraphics[width=\columnwidth]{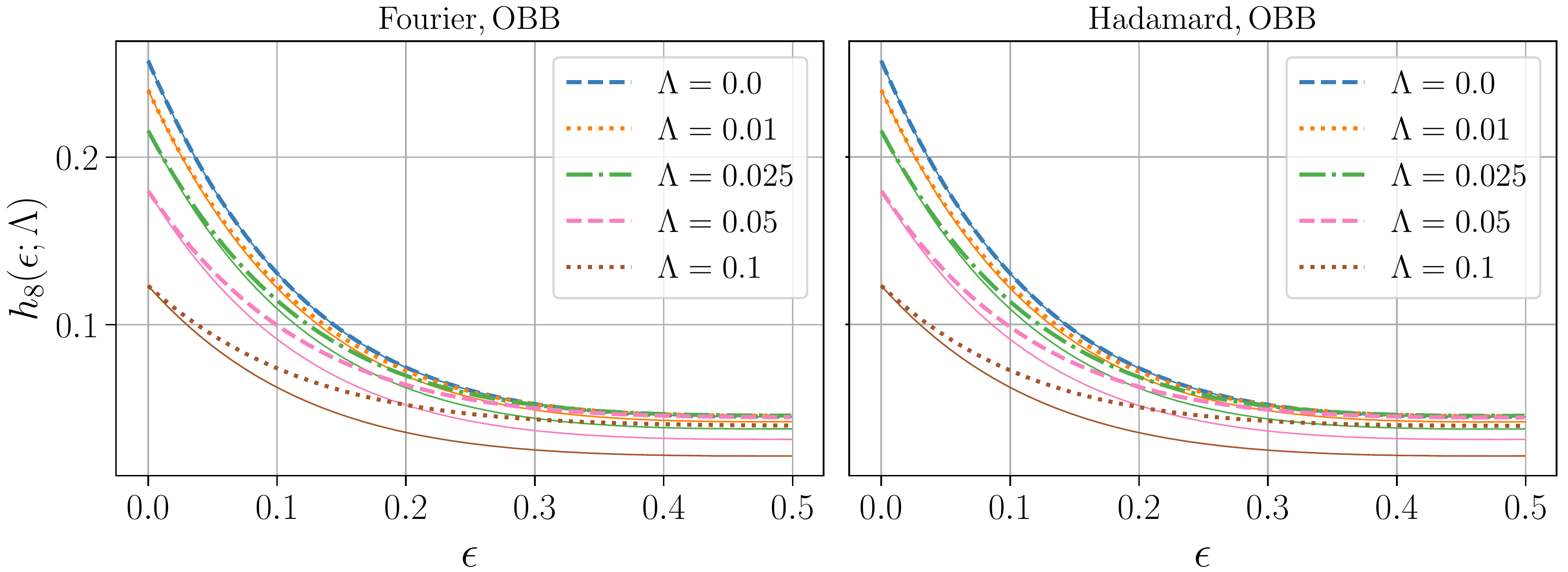}
    \caption{Heralding probability with photon loss in the Fourier (left) and Hadamard (right) protocols, for $n=8$, with the OBB error model. The SBB curves look similar. 
    Note that the parameter in the legend is $\Lambda = 1-(1-\lambda)^{\log n}$. 
    The thin solid lines represent the lower bound, Eq.~\eqref{eq:lossy decomposition ineq}, and notice they intersect the data curves at $\epsilon=0$ as per Eq.~\eqref{eq:loss_herald_no_err}. We see the lower bound is tighter for smaller $\Lambda$.}
    \label{fig:herald_loss}
\end{figure}

We now briefly discuss the output error rate. 
In the presence of photon loss, the definition of output error rate given in Section~\ref{sec:distillation} is no longer well-defined, as it assumes there is always $1$ output photon. 
In the more general setting, it is more natural to consider the 
\emph{output fidelity} $f_n(\epsilon; \lambda)$ to be the probability that, given the heralding of an ideal pattern, there is a single output photon with internal state $\ket{\xi_0}$. 
(In other words, we count the probability of outputs with no photon loss \emph{and} an ideal output photon.) 
The output error rate is then $e_n(\epsilon; \lambda) = 1-f_n(\epsilon;\lambda)$. 
We note that $f_n(\epsilon; \lambda)$ is straightforwardly computable in terms of the quantities we have already discussed. In particular, since successful output requires no photon loss, 
\begin{align}
    f_n(\epsilon;\lambda) &= P(\textnormal{no loss}|\textnormal{herald})f_n(\epsilon; 0) \label{eq:lossy fidelity def}
    \\&= \dfrac{(1-\Lambda)^{n}h_n(\epsilon; 0)}{h_n(\epsilon; \lambda)}f_n(\epsilon; 0)
    \\&= \dfrac{(1-\Lambda)f_n(\epsilon; 0)}{1+n\Lambda c_n(\epsilon)/h_n(\epsilon)},\label{eq:lossy fidelity quotient}
\end{align}
where the second equality follows from Bayes' theorem and the third follows from 
expanding $h_n(\epsilon; \lambda)$ according to \eqref{eq: lossy decomposition detailed}. 
To estimate the output fidelity, we recall that $h_n(0)\approx 1/4 \neq 0$ (see Theorem~\ref{thm:herald ideal}) and $c_n(0)=0$ (see the proof of Theorem~\ref{thm:loss}), and therefore $c_n(\epsilon)/h_n(\epsilon) = O(\epsilon)$. Expanding \eqref{eq:lossy fidelity quotient} as a Maclaurin series in $\epsilon$, we obtain
\begin{equation}\label{eq:lossy fidelity}
    f_n(\epsilon; \lambda) = \dfrac{(1-\Lambda)f_n(\epsilon; 0)}{1+n\Lambda c_n(\epsilon)/h_n(\epsilon)} = (1-\Lambda)f_n(\epsilon; 0) + O(n\Lambda\epsilon).
\end{equation}
Then up to \emph{second-order} corrections involving both a photon loss and a distinguishability error, the output fidelity is simply the output fidelity in the lossless case multiplied by the probability that the output photon is not lost. 

We may also phrase \eqref{eq:lossy fidelity def} as follows. Given successful heralding, at most one photon can be lost. 
If there is one loss, then there is no photon in the output mode. 
If there is no loss, then we reduce to the lossless case, Theorem~\ref{thm:error rate}. 

\begin{figure}[th]
    \centering
    \includegraphics[width=0.6\columnwidth]{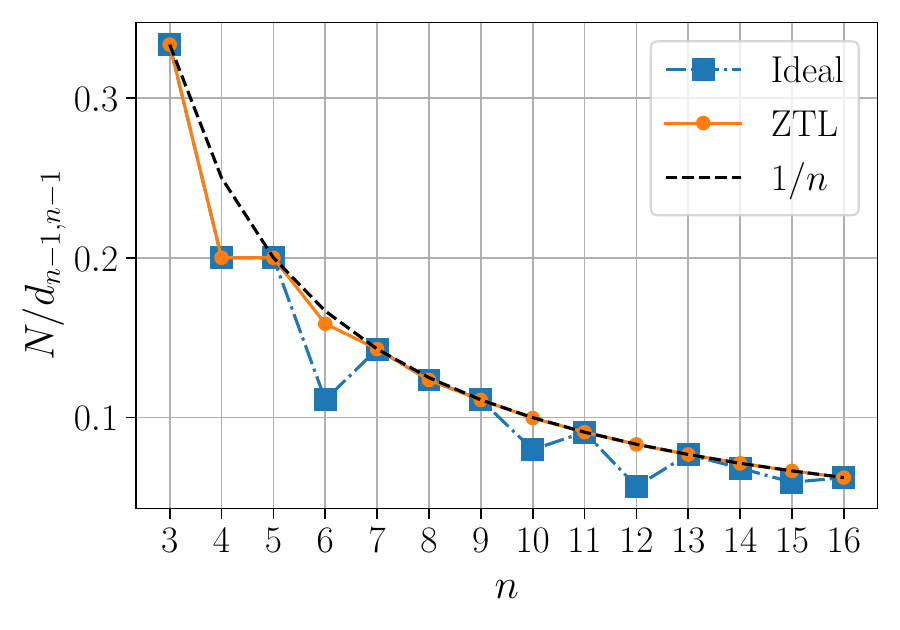}
    \caption{For the Fourier matrices $F_n$, we plot the ratio of the total number $N$ of ideal patterns (respectively, ZTL patterns with 1 photon in mode 0) to the total number of possible patterns $d_{n-1, n-1}$ (dimension of the space of $n-1$ photons in $n-1$ modes). As expected from \cite{tichy2010zero}, it scales like $1/n$.  We can also see here that for prime powers, the ZTL patterns match the ideal ones (see Th.~\ref{thm:suppression body}), but for non-prime powers there is additional suppression.}
    \label{fig:pattern_ration}
\end{figure}

\section{Discussion}

We have provided a set of protocols for reducing photonic distinguishability errors, using $n$ photons to reduce the error rate by a factor of $n$. 
This generalizes the work of  \cite{sparrow_quantum_2017, marshall_distillation_2022}. Our protocols, based on the Hadamard and Fourier unitaries (or more generally the Fourier transform on any finite abelian group, as discussed in Appendix~\ref{sec:fourier abelian}), 
can scale to arbitrarily many photons (in contrast to \cite{marshall_distillation_2022}) 
and retain a high heralding probability near $1/4$ even for large $n$ (in contrast to \cite{sparrow_quantum_2017}). These protocols are efficient in the sense that they require approximately $4n$ photons to reduce the error rate by a factor of $n$. This is a major improvement over the previous state-of-the-art protocol of \cite{marshall_distillation_2022}, which requires $O(n^2)$ photons to obtain similar error reduction. 
We consider a detailed example in Remark~\ref{remark:efficiency}, demonstrating the greatly reduced resource costs in our setting. 
Further, to reduce the error rate by more than a factor of $4$ without exponentially scaling costs, the protocols of \cite{marshall_distillation_2022} require an iterative scheme involving active feed-forwarding and quantum memory during distillation, which can potentially introduce new errors and negate the benefits of distillation. 
These additional assumptions are not required for our new protocols, which are therefore significantly more resource-efficient \emph{and} potentially simpler to implement experimentally. 

We note that a major potential application of distillation protocols is to fault-tolerant linear optical quantum computation, for example fusion-based quantum computation (FBQC). Distinguishability errors may often lead to measurement errors in FBQC and similar paradigms, and therefore reducing the distinguishability error rates can improve the thresholds for other errors. Increasing the photon loss threshold is of particular importance, since loss tends to be the dominant source of noise. 
In turn, these improved loss thresholds may result in significantly decreased resource requirements for FBQC. 
Since our distillation protocols are efficient, source-agnostic, and require only standard linear optical hardware, this motivates the use of distillation to improve distinguishability error rates beyond even those obtained by state-of-the-art single photon sources. 
We discuss these considerations in greater detail in Section~\ref{sec:error correction}.

In this work, we established the theory for these new and efficient distillation protocols. We proved and conjectured 
various results, including 
exact formulas and lower bounds for heralding rates in Sections~\ref{sec:heralding} and \ref{sec:symmetry body}, 
characterization of the effect of photon loss in Section~\ref{sec:loss}, 
and discussion of resource costs and optimal choices of distillation protocols in Sections~\ref{sec:iterated} and \ref{sec:numerics}. 
Much of this analysis was enabled by the symmetries present in the Fourier and Hadamard matrices \cite{dittel2018totally} and, more generally, the Fourier transform associated with a finite abelian group. This theoretical perspective explains earlier work on distillation \cite{marshall_distillation_2022}, relates it to the study of suppression laws, and provides insight into the efficiency of the new distillation protocols. 

We also show in App.~\ref{sec:haar} that while distillation is \textit{possible} using almost any unitary (in particular for Haar random unitaries), it does not yield an efficiently scaling protocol. 
First we observe that in the Haar case (as with the Fourier and Hadamard cases), the heralding rate for the ideal patterns asymptotically approaches 1/4 (the `ideal' patterns in the Haar case are typically just all possible patterns, since there is no symmetry). 
However, in the Haar case, post-selecting on all ideal patterns increases the error (approximately doubling it).
Instead, Haar random unitaries can be used for distillation only when one chooses the post-selection patterns as a small subset of ideal patterns with a high amount of constructive interference (high probability). 
Although the Haar case does not scale efficiently, these observations help to highlight how the Hadamard/Fourier protocols achieve efficiency; the symmetry causes a concentration on a relatively small number of possible patterns, as demonstrated in Fig.~\ref{fig:pattern_ration}. 
This figure shows that in the Fourier case, the total proportion of ideal patterns (out of all possible ones) scales like $1/n$, as expected from analysis in Ref.~\cite{tichy2010zero}. Remarkably, the weight of these patterns remains approximately constant, at $1/4$, as proven in Theorem~\ref{thm:herald ideal}. 

We proved in Theorem~\ref{thm:suppression body} that for $n$ a prime power, the Fourier protocols (like the Hadamard protocols \cite{crespi2015suppression}) have ideal patterns exactly characterized by symmetry, 
resolving an open problem due to \cite{tichy2010zero}. 
However, not all suppression laws come from symmetry conditions \cite{bezerra2023families}. In particular, 
our numerical analysis in Section~\ref{sec:numerics} indicates that the optimal distillation protocols likely correspond to $F_n$ where $n$ is \emph{not} a prime power. (See Figure~\ref{fig:opt_n}.) In those cases, again by Theorem~\ref{thm:suppression body}, the ideal patterns are a strict subset of those determined by the known symmetry conditions. 
Further, there seem to be additional properties determining which $n\neq p^r$ lead to optimal distillation protocols: for example, $F_n$ with $n$ divisible by $3$ seem to outperform other non-prime-power $n$. 
Thus we believe that future work on distillation should focus on understanding and taking advantage of the extra suppression occurring in these cases. 

\section{Acknowledgements}
We are grateful for support from the NASA SCaN program, from DARPA under IAA 8839, Annex 130, and from  NASA Ames Research Center.
J.M. is thankful for support from NASA Academic Mission Services, Contract No. NNA16BD14C. 
N.A. is a KBR employee working under the Prime Contract No. 80ARC020D0010 with the NASA Ames Research Center. 
We also thank Vaclav Kotesovec for pointing out the proof of Theorem~\ref{thm:herald ideal} Part~\ref{part:herald ideal 2} discussed in Appendix~\ref{sec:maple appendix} and Bojko Bakalov for suggesting the study of general Fourier transforms on finite abelian groups, as discussed in Appendix~\ref{sec:fourier abelian}. 
The United States Government retains, and by accepting the article for publication, the publisher acknowledges that the United States Government retains, a nonexclusive, paid-up, irrevocable, worldwide license to publish or reproduce the published form of this work, or allow others to do so, for United States Government purposes.

\bibliography{refs, report}

\appendix

\section{Symmetry}\label{sec:symmetry}

\subsection{General case}\label{sec:symmetry general}

In this section, we review an argument due to \cite{dittel2018totally} to show that certain types of symmetries of a unitary $U$ determine the transition probabilities of the associated linear optical unitary $\hat{U}$. 
Suppose we have an $n\times n$ unitary $U$. 
As above, we write $\hat{U}\ket{\psi}$ to denote the action of the corresponding linear optical unitary on the state $\ket{\psi}\in\mc{H}^{\otimes n}$. 
In this section, we only consider PNRD. As shorthand, we write $\bra{s_0, \dots, s_{n-1}}$ for the projector onto the subspace of states with $s_j$ photons in mode $j$. 
(Regardless of internal states.) 
Thus we will be interested in the transition probabilities $\bra{s_0, \dots, s_{n-1}}\hat{U}\ket{\psi}$. 

Suppose there is some $n\times n$ permutation matrix $P$ such that $UP = DU$, where $D$ is a diagonal matrix with diagonal entries $d_0, \dots, d_{n-1}$. (By unitarity, all $|d_i| = 1$.) 

\begin{remark}
Before continuing, we emphasize that the linear optical unitary $\hat{P}$ is a permutation of the \emph{external modes} of the input state, acting by $P^{\otimes n}\otimes I^{\otimes n}$ on $\mc{H}_{\textnormal{ext}}^{\otimes n}\otimes \mc{H}_{\textnormal{int}}^{\otimes n}$. 
These should not be confused with the permutations of the \emph{photons} discussed in Section~\ref{sec:linear optics}. 
\end{remark}

For simplicity, we begin with a special case in which the input state is fully indistinguishable and there are no internal modes. 
We have
\begin{equation}
    \hat{U}\ket{1,\dots, 1} = \hat{U}\hat{P}\hat{P}^\dagger \ket{1,\dots, 1} = \hat{D}\hat{U}\ket{1,\dots, 1},
\end{equation}
since $\ket{1, \dots,  1}$ is unaffected by permutations of the modes. 
Then for any $s_0, \dots, s_{n-1}$, we have
\begin{equation}
    \bra{s_0\cdots s_{n-1}}\hat{U}\ket{1, \dots,  1} = \bra{s_0\cdots s_{n-1}}\hat{D}\hat{U}\ket{1, \dots,  1} = d_0^{s_0}\cdots d_{n-1}^{s_{n-1}}\bra{s_0\cdots s_{n-1}}\hat{U}\ket{1, \dots,  1}.
\end{equation}
Since all $d_i$ are nonzero, we conclude that either 
\begin{equation}
    d^s = 1\textnormal{ or }\bra{s_0\cdots s_{n-1}}\hat{U}\ket{1, \dots,  1}=0,
\end{equation}
where $d^s := d_0^{s_0}\cdots d_{n-1}^{s_{n-1}}.$  
Thus, the output patterns arising with nonzero probability must satisfy $d^s = 1$. 
These precisely correspond to the known suppression laws for the Fourier and Hadamard cases \cite{tichy2010zero, crespi2015suppression}, discussed below. 

Note that the only property of $\ket{1, \dots,  1}$ we needed was its invariance under certain permutations of its modes. More generally, suppose that $\ket{\psi}\in\mc{H}^{\otimes n}$ is a general pure state of $n$ photons in $n$ external modes, potentially distinguishable. 
Assume that $P\ket{\psi} = e^{i\phi}\ket{\psi}$. Then the same argument gives
\begin{equation}
    \bra{s_0\cdots s_{n-1}}\hat{U}\ket{\psi} = \bra{s_0\cdots s_{n-1}}\hat{U}\hat{P}\hat{P}^\dagger\ket{\psi}
    = e^{-i\phi}\bra{s_0\cdots s_{n-1}}\hat{D}\hat{U} \ket{\psi}
    = e^{-i\phi} d^s \bra{s_0\cdots s_{n-1}}\hat{U} \ket{\psi}.
\end{equation}
We conclude that either $d^s = e^{i\phi}$ or $\bra{s_0\cdots s_{n-1}}\hat{U}\ket{\psi} = 0$. 
We summarize in the following lemma.  
\begin{lemma}\cite{dittel2018totally}\label{lem:suppression}
Let $U$, $P$, and $D$ be $n\times n$ unitary matrices, with $P$ a permutation matrix and $D$ diagonal with diagonal entries $d_0, \dots, d_{n-1}$. 
Suppose $UP = DU$. 
Let $\ket{\psi}$ be a (pure) state of $n$ photons in $n$ modes satisfying $\hat{P}\ket{\psi} = e^{i\phi}\ket{\psi}$. 
If $\bra{s_0, \dots, s_{n-1}}\hat{U}\ket{\psi}\neq 0$, we have
\begin{equation}
    d^s = e^{i\phi}.
\end{equation}
\end{lemma}
We now apply this result to the calculation of certain transition probabilities. 
Let 
\begin{equation}\label{def:S}
    \mc{S} = \{\ket{s_0, \dots, s_{n-1}}: \sum s_j = n, d^s = 1\}
\end{equation}
be the set of output patterns arising with nonzero probability from $P$-symmetric input. As before, we have that $d^s := d_0^{s_0}\cdots d_{n-1}^{s_{n-1}}$. We decompose the state space into eigenspaces for $P$ with eigenvalues $e^{i\phi}$. 
For any $\ket{\psi'}$ with eigenvalue $e^{i\phi}\neq 1$ and $\ket{s}\in \mc{S}$, the lemma implies $\bra{s}\hat{U}\ket{\psi'} = 0$. 
Then, for arbitrary $\ket{\eta}$, decompose 
\begin{equation}
    \ket{\eta} = \braket{\eta_+}{\eta}\ket{\eta_+} + \ket{\eta_\perp},
\end{equation}
where $\ket{\eta_+}$ is the unit vector projection of $\ket{\eta}$ into the $+1$-eigenspace and $\ket{\eta_\perp}$ collects the terms from other eigenspaces. We have, for $\ket{s}\in\mc{S}$,
\begin{equation}\label{eq:reduce to symm}
    \bra{s}\hat{U}\ket{\eta} =
    \bra{s}\hat{U}\left(\braket{\eta_+}{\eta}\ket{\eta_+} + \ket{\eta_\perp}\right)
    =\braket{\eta_+}{\eta}\bra{s}\hat{U}\ket{\eta_+}.
\end{equation}
More generally, consider an abelian group $G$ of $n\times n$ permutation matrices, where each $P\in G$ has corresponding diagonal matrix $D_P$ with $UP = D_P U$. 
(In other words, writing $UPU^\dagger = D_P$, we see that $G$ is a group of simultaneously diagonalizable permutation matrices, and the change of basis is given by $U$.) 
Let $d_{P,i}$ index the diagonal entries of $D_P$, with $d_{P}^s$ extending the $d^s$ notation above. 
We consider the set of output patterns that occur with nonzero probability for \emph{all} $P\in G$:
\begin{equation}
    \mc{S}_G = \{\ket{s}: \forall P\in G, d_P^s = 1\}.
\end{equation}
We may generalize the above argument to obtain the following, Part \ref{step:overlap} of Theorem~\ref{thm:symm body}: 
\begin{lemma}\label{lem:symmetrization}
Let all notation be as above. Let $\ket{\eta}$ be an arbitrary pure state, and let $\ket{\eta_+}$ be its projection into the joint $+1$-eigenspace of all $P\in G$, normalized to be a unit vector (if nonzero). 
For $\ket{s}\in \mc{S}_G$, we have
\begin{equation}
    \bra{s}\hat{U}\ket{\eta} = \braket{\eta_+}{\eta}\bra{s}\hat{U}\ket{\eta_+}.
\end{equation}
Further, if $\ket{s}\not\in\mc{S}_G$ and $\braket{\eta_+}{\eta}\neq 0$, then $\bra{s}\hat{U}\ket{\eta}=0$. 
\end{lemma}
\begin{proof}
Since $G$ is abelian, we may decompose $\hn$ into a direct sum of joint eigenspaces for $G$. We index these eigenspaces $V_\alpha$ by functions $\alpha: G\rightarrow\mathbb{C}^*$, such that for $P\in G$ and $\ket{\eta_\alpha}\in V_\alpha$, we have 
\begin{equation}
    P\ket{\eta_\alpha} = \alpha(P) \ket{\eta_\alpha}.
\end{equation}
Note that if $\alpha\neq \beta$, then $V_\alpha$ is orthogonal to $V_\beta$. 
In particular, given $\ket{\eta}$ as above, we may decompose
\begin{equation}
    \ket{\eta} = \sum_\alpha c_\alpha \ket{\eta_\alpha},
\end{equation}
where the set of all nonzero $\ket{\eta_\alpha}$ is an orthonormal set and $c_\alpha = \braket{\eta_\alpha}{\eta}$. 
The term $\ket{\eta_+}$ above corresponds to $\alpha$ trivial, $\alpha(P) = 1$ for all $P\in G$. 
Now consider $\ket{s}\in\mc{S}_G$, so that all $d_P^s = 1$. 
If $\alpha$ is nontrivial, find $P$ such that $\alpha(P)\neq 1$. 
We have $d_P^s = 1\neq \alpha(P)$, so by Lemma~\ref{lem:suppression}, $\bra{s}\hat{U}\ket{\eta_\alpha} = 0$. 
This gives the first claim; the second is a direct consequence of Lemma~\ref{lem:suppression}. 
\end{proof}
In other words, to calculate the output probabilities of patterns in $\mc{S}_G$, 
it suffices to consider the projection of $\ket{\eta}$ into the space of $G$-symmetric vectors. 
Since $\ket{1, \dots, 1}$ is fully symmetric, we observe the following: 
\begin{lemma}\label{lem:ideal}
Let all notation be as above. Let $\ket{s}$ satisfy $\bra{s}\hat{U}\ket{1, \dots, 1}\neq 0$. 
Then $\ket{s}\in \mc{S}_G$. 
\end{lemma}
In particular, the \emph{ideal patterns} discussed in Section~\ref{sec:distillation} satisfy $\ket{s} = \ket{s_0, \dots, s_{n-1}}$, with $s_0 = 1$ and $\bra{s}\hat{U}\ket{1, \dots, 1}\neq 0$. 
Then the ideal patterns are a subset of $\mc{S}_G$; to calculate $\bra{s}\hat{U}\ket{\eta}$ for an ideal pattern $\ket{s}$, it suffices to consider $\bra{s}\hat{U}\ket{\eta_+}$. 

For later use, we consider a special case. Assume that for all $P\in G$, we have $\bra{\eta}P\ket{\eta}\in\{0,1\}$, so that each permutation either fixes $\ket{\eta}$ or takes it to an orthogonal state. 
Let $K$ be the \emph{stabilizer} of $\ket{\eta}$, the subgroup of $G$ with $P\ket{\eta} = \ket{\eta}$ for $P\in K$. 
Letting $G/K$ be the set of left cosets of $K$ in $G$, 
we have
\begin{equation}
    \ket{\eta_+} = \dfrac{1}{\sqrt{|G/K|}}\sum_{g\in G/K} g \ket{\eta}.
\end{equation}
Note that 
\begin{equation}
    \braket{\eta_+}{\eta} = \dfrac{1}{\sqrt{|G/K|}} = \sqrt{|K|/|G|}\geq \dfrac{1}{\sqrt{|G|}}.
\end{equation}

\subsection{Fourier case}\label{sec:fourier symm}
As an important example, we consider the Fourier case, $U=F_n$. We use the notation $\omega = e^{2\pi i/n}$. 
Take the permutation matrix $P$ to correspond to the cyclic permutation sending column $j$ to column $j+1$ (modulo $n$). 
We get $UP = DU$, where $d_i = \omega^{i}$. 
The group $G$ of symmetries is the cyclic group generated by $P$. 
Then the set $\mc{S} = \mc{S}_G$ consists of $\ket{s_0, \dots, s_{n-1}}$ with
\begin{equation}
    1 = \prod_i \omega^{-i s_i} = \omega^{-\sum_i i s_i}, 
\end{equation}
or equivalently $\sum_i i s_i \equiv 0 \mod n$. This is the Zero Transmission Law (ZTL) for the Fourier transform \cite{tichy2010zero}. 
We call the patterns in $\mc{S}$ \emph{ZTL patterns}. 
This is typically rephrased as in Section~\ref{sec:distillation} above, by letting $g_0, \dots, g_{n-1}$ be the weakly decreasing sequence for which $s_i$ counts the number of $j$ with $g_j = i$.  
The ZTL then translates to $\sum_i g_i\equiv 0\mod n$. 

As an example, consider $n=2$, so that $U=F_2 = H_2 = H$. Only the patterns $s = (2,0), (0,2)$ satisfy the ZTL: this is precisely the HOM effect. 
(Although note that neither of these are ideal patterns, since neither has $s_0=1$.) 
Further, Lemma~\ref{lem:suppression} explains why $\hat{H}$ must map the symmetric states $\ket{1,1}$ and $\frac{1}{\sqrt{2}}(\ket{2,0} + \ket{0,2})$ to $\textnormal{span}\{\ket{2,0},\ket{0,2}\}$, 
while mapping the antisymmetric state $\frac{1}{\sqrt{2}}(\ket{2,0} - \ket{0,2})$ to $\ket{1,1}$ (the only non-ZTL pattern of $2$ photons in $2$ modes). 
Thus the ZTL may be viewed as a direct extension of the HOM effect. 

By Lemma~\ref{lem:ideal}, the ideal patterns must be a subset of $\mc{S}$. 
However, as discussed in Section~\ref{sec:distillation} above, there exist ZTL patterns $\ket{s}$ with $s_0 = 1$ that are \emph{not} ideal patterns. 
(This is discussed in the original paper \cite{tichy2010zero} without the $s_0=1$ restriction.) 
In other words, these patterns satisfy the required symmetry conditions but still cannot be obtained from the input state $\ket{1,\dots, 1}$. 
The first example arises for $n=6$, where there are $6$ ZTL patterns $\ket{s}$ with $s_0 = 1$ that are not ideal patterns. These are: 
\begin{equation}
    (1, 0, 1, 1, 2, 1),
 (1, 0, 2, 0, 1, 2),
 (1, 1, 0, 1, 1, 2),
 (1, 1, 2, 1, 1, 0),
 (1, 2, 1, 0, 2, 0),
 (1, 2, 1, 1, 0, 1).
\end{equation}
In the following section, we characterize exactly when this occurs. 

\subsubsection{Fourier case for prime powers}\label{sec:fourier prime power}
In this section, we prove the following theorem: 
\begin{theorem}\label{thm:fourier prime power}
Let $n\geq 2$ be a positive integer. 
Let $g_0, \dots, g_{n-1}$ be a list of integers in $0, \dots, n-1$, and let $\ket{s} = \ket{s_0, \dots, s_{n-1}}$ be the corresponding Fock state, with $s_i$ counting the number of $j$ with $g_j = i$. 
Let $A$ be the $n\times n$ matrix whose $i$th row is the $g_i$th row of $\sqrt{n}F_n$. 
\begin{enumerate}
    \item Further assume $n=p^r$, where $p$ is a prime. Then $\perm(A)\neq 0$ if and only if $\sum_i g_i \equiv 0 \mod n$. \label{part:fourier prime appendix part 1}
    \item Instead assume that $n$ is \emph{not} a prime power. Then there exists a choice of $g_0, \dots, g_{n-1}$ satisfying $\sum_i g_i \equiv 0 \mod n$ and $s_0 = 1$ such that $\perm(A) = 0$.  \label{part:fourier prime appendix part 2}
\end{enumerate}
\end{theorem}
By the permanent formulation of boson sampling \cite{aaronson_arkhipov}, we see that 
for $n$ a prime power and $\ket{s}$ satisfying the Zero Transmission Law, we have $\bra{s}U\ket{1, \dots, 1} \neq 0$. 
We immediately obtain
\begin{corollary}
If $n = p^r$, where $p$ is a prime, then for the distillation protocol corresponding to $U=F_n$, all ZTL patterns with $s_0=1$ are ideal patterns. 
If $n$ is not a prime power, then there exist ZTL patterns with $s_0=1$ that are \emph{not} ideal patterns. 
\end{corollary}
This gives Theorem~\ref{thm:suppression body}. 
Before the proof, we briefly introduce some additional notation. 
Let $x_0, \dots, x_{n-1}$ be indeterminates; for an $n$-tuple $s = (s_0, \dots, s_{n-1})$, let $x^s = x_0^{s_0} x_1^{s_1}\cdots x_{n-1}^{s_{n-1}}$. 
We will use the \emph{generic circulant matrix} $C$, the $n\times n$ matrix with $(i,j)$ entry equal to $x_{i+j}$ (where the indices are taken modulo $n$). 
We now prove Theorem~\ref{thm:fourier prime power}, a direct application of results of \cite{thomas2004number, colarte2019coefficients}. 
We note that these works make no explicit reference to boson sampling or the matrix $A$ defined above: instead, their results are related to our setting by \eqref{eq:circulant det} below. 
\begin{proof}
Let all notation be as in the theorem.  
The determinant of the circulant matrix, $\det(C)$, is a polynomial in $x_0, \dots, x_{n-1}$, and the coefficient of $x^s$ in $\det(C)$ is \cite{thomas2004number}
\begin{equation}\label{eq:circulant det}
    [x^s] \det(C) = \sum_{\sigma\in S_n} \omega^{\sum_{j=0}^{n-1}j g_{\sigma(j)}} = \perm(A). 
\end{equation}
We note that only the first equality is observed in \cite{thomas2004number}, in the proof of the main theorem; the second equality of \eqref{eq:circulant det} is immediate from the definition of the permanent. 
Now, \cite{thomas2004number} shows that for $n$ a prime power, $\sum_i g_i \equiv 0\mod n$ if and only if \eqref{eq:circulant det} is nonzero. This proves Part~\ref{part:fourier prime appendix part 1} of the theorem. 
For Part~\ref{part:fourier prime appendix part 2}, suppose that $n$ is not a prime power. 
Then Theorem 3.5 of \cite{colarte2019coefficients} explicitly constructs a list of nonnegative integers $g_0, \dots, g_{n-1}$ satisfying:   
\begin{enumerate}
    \item $\sum_i g_i \equiv 0\mod n$,
    \item the coefficient of $x^s$ in $\det(C)$ is $0$. 
\end{enumerate}
The first condition is the ZTL; the second shows that the corresponding permanent vanishes, by \eqref{eq:circulant det}. 
Then we must only show that there is some $g_i$ with multiplicity $1$, i.e., the corresponding $s_j$ is equal to $1$. 
Once this is proven, then by the cyclic symmetry of the Fourier transform we obtain an example with $s_0=1$. 
(More explicitly, we may note that if 
we perform the transformation $g_i\mapsto g_i + 1$ for all $i$, the condition $\sum_i g_i \equiv 0\mod n$ is unchanged, and \eqref{eq:circulant det} will differ only by a factor of $\omega^{\sum_{j=0}^{n-1} j} = \pm 1$. This operation on the $g_i$ corresponds to cyclically permuting the $s_j$; thus, we may do this repeatedly, without changing $|\perm(A)|$, until we obtain an example with $s_0 = 1$.)  
It remains to show the relevant properties of the $g_i$. 
For this, we must give more details of the construction of \cite{colarte2019coefficients}. 
We leave the proof that the state we construct has vanishing permanent to \cite{colarte2019coefficients}, which uses Corollary 6 of \cite{malenfant2015matrix}. 
Suppose that $n=q_1 q_2$, where $q_1,q_2\geq 2$ are relatively prime. By Bezout's theorem (exchanging the roles of $q_1,q_2$ if necessary), we find positive integers $c_1,c_2$ with  
$c_1 q_1 = 1+c_2 q_2$, where 
$1\leq c_1 < q_2$ and $1\leq c_2 < q_1$. 
We then define 
\begin{align*}
    M_1 &= c_2 q_2 - 1
    \\ M_0 &= n - c_2 q_2 - 2
    \\ A_{1} &= c_2 q_2 - c_2 c_1 + 1
    \\ A_{2} &= n - c_2 q_2
    \\ A_{3} &= n - c_2 q_2 + c_1 c_2.
\end{align*}
We may easily verify that $A_1>1$ and 
\begin{align*}
    A_3 - A_2 &= c_1 c_2 > 0,
    \\ A_2 - A_1 &= (q_2 - c_1)(q_1 - c_2) > 0,
\end{align*}
so we have $1 < A_1 < A_2 < A_3$. 
The state is then constructed by: 
\begin{align*}
    &g_0 = g_1 = \dots = g_{M_0-1} = 0,
    \\&g_{M_0} = g_{M_0+1} = \dots = g_{M_0+M_1-1} = 1,
    \\&g_{n-3} = A_1, \,\,g_{n-2} = A_2, \,\,g_{n-1} = A_3,
\end{align*}
where we note $M_0 + M_1 = n-3$. 
In particular, there are $M_0$ photons in mode $0$, $M_1$ photons in mode $1$, and additional photons in the three distinct modes  $A_1,A_2,A_3>1$. 
Then we have all $s_{A_j} = 1$, $j=1,2,3$. 
This completes the proof. 
\end{proof}

We now give the $n=6$ example of the construction of \cite{colarte2019coefficients}. 
We write $q_1=3$, $q_2 = 2$, so that $6 = q_1 q_2$. Note that with $c_1 = c_2 = 1$, we have $c_1 q_1 = 1+c_2 q_2 = 3$. 
We obtain $M_0=2, M_1=1, A_1 = 2, A_2 = 4, A_3 = 5$. 
This gives 
\begin{equation}
    g_0 = g_1 = 0, g_2 = 1, g_3 = 2, g_4 = 4, g_5 = 5,
\end{equation}
with corresponding Fock state $\ket{s} = \ket{2,1,1,0,1,1}$. 
We cyclically permute this state to see that $\ket{1,1,0,1,1,2}$ is another such example, now with $1$ photon in the $0$th mode. 
We note that this is one of the examples given above. 
As another example, for $n=18$ we obtain $\ket{s} = \ket{1,0,0,0,7,8,0,0,0,1,0,0,0,1,0,0,0,0}$. 

\begin{remark}
We note that the proof of Theorem 3.5 of \cite{colarte2019coefficients} has a typographical error in its first paragraph. 
In particular, the $a_i$ there (equivalent to our $g_i$) are said to satisfy $\sum_i i a_i \equiv 0\mod n$ and $\sum_i a_i = n$, which are in fact the conditions on our $s_i$. 
These conditions are not used in the proof, and from context it is clear that this was a simple typographical error: elsewhere in ibid., the notation $a_i$ is used for the $g_i$ and they are said to satisfy the usual ZTL $\sum_i a_i\equiv 0\mod n$; and the notation $\alpha_i$ or $M_i$ is used for our $s_i$. 
(See the discussion between equation (3) and Proposition 3.2 of ibid.) 
The proof of the theorem reduces to checking the conditions given in Corollary 6 of \cite{malenfant2015matrix}, which is in fact stated in terms of our $g_i$. 
\end{remark}

To obtain \eqref{eq:circulant det}, one simply notes that the eigenvalues of the generic circulant matrix $C$ have the form $\sum_j \omega^{ij}x_j$, so that 
\begin{equation}
    \det(C) = \prod_i \sum_j \omega^{ij}x_j.
\end{equation}
For terms involving $x^s$, one needs to sum over all possible products of the form 
\begin{equation}
    \omega^{0j_0}x_{j_0}\omega^{1j_1}x_{j_1}\cdots\omega^{(n-1)j_{n-1}}x_{j_{n-1}}
\end{equation}
in which there are $s_0$ copies of $x_0$, $s_1$ copies of $x_1$, and so on. 
The possible values of the sequence $j_k$ are precisely the permutations of $g_0, \dots, g_{n-1}$, giving \eqref{eq:circulant det}. 

We now note further consequences of \eqref{eq:circulant det}. 
Since $\det(C)$ must be a polynomial in the $x_i$ with integer coefficients, we immediately see that $\perm(A)$ is an integer. Then the weights $\bra{s}F_n\ket{1, \dots, 1}$ are all integers up to a known normalization factor of $n^{n}(s_0!)\cdots (s_{n-1}!)$. (Recall that we defined $A$ in terms of $\sqrt{n}F_n$, explaining the factor of $n^n$. The factorials come from the permanent formulation of boson sampling \cite{aaronson_arkhipov}.) 
Further, we note that for a fixed $n$, \eqref{eq:circulant det} allows one to compute \emph{all} $\bra{s}U\ket{1, \dots, 1}$, generally expressed in terms of permanents of $n\times n$ matrices, through the calculation of the determinant of a single $n\times n$ matrix (with symbolic entries). 
Eq. \eqref{eq:circulant det} also allows for the connection of boson sampling with the literature on the coefficients of $\det(C)$, which offers various combinatorial formulas such as those in Ref. \cite{malenfant2015matrix}, which were instrumental in proving Theorem~\ref{thm:fourier prime power}. 
We note that the relevant permanents $\perm(A)$ were previously related with $\det(C)$ in \cite{kocharovsky2015exact}, in the context of the three-dimensional Ising model. 

\subsection{Hadamard case}\label{sec:hadamard symm}
We now consider the Hadamard case. We let $n=2^r$ and 
identify $\mc{H}_{\textnormal{ext}} = \mathbb{C}^n$ with $(\mathbb{C}^2)^{\otimes n}$ by representing kets by binary strings. 
We take $U = H_n = H^{\otimes r}$, where $H$ is the $2\times 2$ Hadamard matrix. 
Let $X$ and $Z$ be the $2\times 2$ Pauli matrices. 
Of course, $H$ satisfies $HX = ZH$. 
This is, in fact, a special case of the framework above, since $X$ is a permutation matrix and $Z$ is diagonal. 
Following \eqref{def:S}, we find that in the case $n=2$, $\mc{S}$ consists of patterns $\ket{s_0, s_1}$ with $s_0 + s_1 = 2$ and $(-1)^{s_1} = 1$. In other words, $s_0, s_1\in\{0,2\}$, recovering the classic Hong-Ou-Mandel effect. 

We may extend this analysis to obtain symmetries of $H_n$ by taking tensor products. 
We consider 
the following symmetry group $G$ of $U$: 
\begin{equation}
    G = \{X^{i_1}\otimes X^{i_2}\otimes \cdots\otimes X^{i_r}: i_j\in\{0,1\}\}.
\end{equation}
These satisfy
\begin{equation}
    H_n (X^{i_1}\otimes X^{i_2}\otimes \cdots\otimes X^{i_r}) = (Z^{i_1}\otimes\cdots\otimes Z^{i_r}) H_n.
\end{equation}
We will compute $S_G$, defined as above. 
Letting $X_i$ and $Z_i$ be the actions of $X$ and $Z$ on the $i$th tensor factor, we note that 
$G$ is generated by $X_1, \dots, X_r$ and 
$H_n X_i = Z_i H_n$. 
To proceed, we will write down the map $\mathbb{H}_{\textnormal{ext}} = \mathbb{C}^n\rightarrow (\mathbb{C}^2)^{\otimes n}$ explicitly: 
for $k = \sum_{j=0}^{n-1} c_j 2^j$, we map $\ket{k}\mapsto \ket{c_{0}, \dots, c_{n-1}}$. 
Then $Z_i\ket{k} = (-1)^{c_i}\ket{k}$. 
We now consider what the $X_i$ symmetry tells us about 
$\ket{s} = \ket{s_0, \dots, s_n}\in\mc{H}^{\otimes n}_{\textnormal{ext}}$. 
As discussed in Section~\ref{sec:distillation}, we convert to "first quantization," letting $g_0, \dots, g_{n-1}$ be the weakly increasing sequence with $s_0$ copies of $0$, $s_1$ copies of $1$, etc. 
The $g_j$ can be viewed as the list of modes occupied by the $n$ photons. 
Let $(g_j)_i$ be the coefficient of $2^i$ in the binary expansion of $s_j$. 
By Lemma~\ref{lem:suppression}, we see that $\ket{s}\in S_G$ only if $\sum_j (g_j)_i \equiv 0 \mod 2$. 
(In other words, $(g_0)_i\oplus \cdots \oplus (g_{n-1})_i = 0$, where $\oplus$ represents the XOR of two bits.) 
Since we may consider symmetries arising from all $X_i$, this must hold for all $i$. 
Then we recover the characterization of $\mc{S}_G$ discussed in \cite{crespi2015suppression} and Section~\ref{sec:distillation} above: extending $\oplus$ to represent the bitwise XOR operation, $\ket{s}\in S_G$ if and only if $g_0\oplus\cdots\oplus g_{n-1} = 0$. 

Unlike the case of the Fourier transform, in the Hadamard case the symmetries given fully characterize the ideal patterns. In other words, the ideal patterns for $U=H_n$ are the elements of $\mc{S}_G$ with $s_0 = 1$. (This is proven in \cite{crespi2015suppression} without the additional condition $s_0=1$.)

\section{Detailed calculations}\label{sec:calculations}
Here we discuss the calculations for heralding probabilities and error rates in detail. 
Let $U$ be the $n$-mode Fourier or Hadamard matrix, with corresponding symmetry group $G$. 
Consider the URS model with parameter $R$, or the limiting OBB model. 
Recall the decomposition \eqref{eq:error decomp} of the $n$-photon state with error rate $\epsilon$: 
\begin{align}
    \rho^{(n)}(\epsilon) = \sum_{k=0}^n \binom{n}{k}\epsilon^k(1-\epsilon)^{n-k} \Phi_k, 
\end{align}
where $\Phi_k$ is the (normalized) uniform mixture of all states
\begin{equation}\label{eq:terms}
    \cre_0[\xi_{j_0}] \cdots \cre_{n-1}[\xi_{j_{n-1}}]\vac = \ket{1, \dots, 1}_{2Q}\otimes\ket{\xi_{j_0},\dots,\xi_{j_{n-1}}}
\end{equation}
(recalling \eqref{eq:internal fock experiment}) with $k$ distinguishability errors, meaning that $k$ of the indices $j_0, \dots, j_{n-1}$ are nonzero. 
For example, in the $R=1$ (SBB) case, $\Phi_k$ has $\binom{n}{k}$ terms, corresponding to a choice of $k$ distinguishable photons. 
For general $R$, $\Phi_k$ has $\binom{n}{k}k^R$ terms; first we choose the $k$ distinguishable photons, then each one has $R$ choices of internal state. 
In the OBB limit, since each photon has a unique error state, $\Phi_k$ again has $\binom{n}{k}$ terms to consider, corresponding to a choice of $k$ distinguishable photons. 

In the following subsections, we will consider detailed computations with the heralding and error rates, of both theoretical and computational importance. 
For practical purposes, we will see that in order to calculate these rates as an \emph{exact} function of $\epsilon$, one must only compute, for all $\ket{\eta}$ of the form \eqref{eq:terms}, $h_n(\ketbra{\eta})$ and $e_n(\ketbra{\eta})$. 

\subsection{Heralding rate}\label{sec:heralding appendix}

First we discuss the heralding rate. 
As discussed in Section~\ref{sec:heralding}, we can calculate the heralding rate via
\begin{equation}
    h_n(\epsilon) = \sum_{k=0}^n \binom{n}{k}\epsilon^k (1-\epsilon)^{n-k} h_n(\Phi_k).
\end{equation}
Since $\Phi_k$ is a uniform mixture, $h_n(\Phi_k)$ is simply the average heralding rate for states of the form \eqref{eq:terms} with $k$ distinguishability errors. 
We discuss the efficient calculation of the $h_n(\Phi_k)$ via simulation in Appendix~\ref{sec:simulation appendix}. 

As observed in Lemma~\ref{lem:ideal}, the ideal patterns are a subset of those in $\mc{S}_G$ with $s_0 = 1$. 
By Lemma~\ref{lem:symmetrization}, then, we have
\begin{equation}\label{eq:h symm}
    h_n(\ketbra{\eta}) 
    = \sum_{\ket{s}\text{ ideal}} \left| \bra{s}\hat{U}\ket{\eta}\right|^2 
    = |\braket{\eta_+}{\eta}|^2 \sum_{\ket{s}\text{ ideal}} \left| \bra{s}\hat{U}\ket{\eta_+}\right|^2 = h_n(\ketbra{\eta_+}).
\end{equation}
This gives the heralding rate statement in Part \ref{step:symm} of Theorem~\ref{thm:symm body}. 
To prove the corresponding statement in Part~\ref{step:symm2}, we consider the sum over ideal patterns. 
In the Hadamard case, the ideal patterns are exactly characterized by the conditions $\ket{s}\in\mc{S}_G$, $s_0 = 1$. 
In the Fourier case, this no longer holds in general, although in many cases, the patterns satisfying $\bra{s}\hat{U}\ket{\eta}\neq 0$ and $s_0 = 1$ are guaranteed to be ideal. 
This includes when $\ket{\eta}$ has $0$ or $1$ distinguishability errors, as in Appendix~\ref{sec:one} below.  This also holds when $n$ is a prime power, as proven in Appendix~\ref{sec:fourier prime power}. 
(Also see the discussion of Remark~\ref{remark:fourier prime assumption}.) 
Thus, going forward, we will assume that for all patterns $\ket{s}\in\mc{S}_G$ with $s_0 = 1$ and $\bra{s}\hat{U}\ket{\eta} \neq 0$, we have $\ket{s}$ ideal. 
Then from this assumption and Lemma~\ref{lem:suppression}, we may include extra vanishing terms to extend the sum to all $\ket{s}$ with $s_0 = 1$: 
\begin{align}
    h_n(\ketbra{\eta}) &= |\braket{\eta_+}{\eta}|^2 \sum_{\ket{s}: s_0=1} \left| \bra{s}\hat{U}\ket{\eta_+}\right|^2 
    \\&= |\braket{\eta_+}{\eta}|^2 \left| (\bra{1}\otimes I)\hat{U}\ket{\eta_+}\right|^2
    \\&= |\braket{\eta_+}{\eta}|^2 \bra{\eta_+}\hat{U}^\dagger \left(\ketbra{1}\otimes I\right)\hat{U}\ket{\eta_+}.\label{eq:heralding as operator}
\end{align}
This allows for the computation of heralding probabilities without needing to check for ideal patterns. 
This gives the heralding rate statement in Part~\ref{step:symm2} of Theorem~\ref{thm:symm body}. 

We may further simplify as follows. 
We have $\bra{1} = \bvac a_0$, where $a_0 = \sum_i a_0[\xi_i]$ is the external annihilation operator (which is summed over all $\xi_i$ so that it annihilates photons regardless of distinguishability). 
Then $\ketbra{1}\otimes I = \cre_0 (\ketbra{0}\otimes I) a_0$. 
Since the entries in the first row of $U$ are uniformly equal to $1/\sqrt{n}$, we have
\begin{equation}
    \hat{U}^\dagger \cre_0 = \frac{1}{\sqrt{n}}\sum_i \cre_i \hat{U}^\dagger
\end{equation}
and thus
\begin{equation}
    h_n(\ketbra{\eta}) = \frac{1}{n}|\braket{\eta_+}{\eta}|^2 \sum_{i,j} \bra{\eta_+}\cre_i \hat{U}^\dagger (\ketbra{0}\otimes I) \hat{U} a_j \ket{\eta_+}. 
\end{equation}
We note that $\hat{U}^\dagger (\ketbra{0}\otimes I) \hat{U}$ is a product of three operators, each of which acts in a ``linear optical" way, transforming each photon independently (and all acting only on external modes). 
For $\hat{U}, \hat{U}^\dagger$ this is by definition; meanwhile, $\ketbra{0}\otimes I$ acts on each creation operator by $\cre_i\mapsto \delta_{i\neq 0}\cre_i$. 
Then the product has the same linear optical form. 
By the unitarity of $U$ and (again) the fact that the entries in its first row are uniformly equal to $1/\sqrt{n}$, we may write
\begin{equation}
    \hat{U}^\dagger (\ketbra{0}\otimes I) \hat{U} = \hat{\Pi},
\end{equation}
recalling the notation $\hat{T}$ for an operator (not necessarily unitary) acting independently on each photon, as introduced in Section~\ref{sec:linear optics}. 
Here, the operator $\Pi$ on $\mc{H}$ is defined by $\Pi = I - \frac{1}{n}\mathbf{1}$, where 
$\mathbf{1}$ satisfies $\mathbf{1}\cre_i = \left(\sum_{j=0}^{n-1} \cre_j\right)\mathbf{1}$. 
In particular, $\frac{1}{n}\mathbf{1}$ is the projection into the mode-symmetric subspace of $\mc{H}$, and $\Pi$ is the complementary projection. 
We have
\begin{equation}\label{eq:heralding simplified fully}
    h_n(\ketbra{\eta}) = \frac{1}{n}|\braket{\eta_+}{\eta}|^2 \sum_{i,j} \bra{\eta_+}\cre_i \hat{\Pi} a_j \ket{\eta_+}. 
\end{equation}

\subsubsection{Ideal heralding}\label{sec:ideal herald}

We now compute $h_n(0)$ by considering the special case of \eqref{eq:heralding simplified fully} with $\ket{\eta} = \ket{\eta_+} = \ket{1, \dots, 1}_{2Q}$, a perfectly indistinguishable Fock state. More generally, these results apply to any pure tensor of $\ket{1, \dots, 1}_{2Q}$ with an internal state. This will be used in Appendix~\ref{sec:one} below. 
We also note that the above assumptions involving the relationship between ideal and symmetry-preserving patterns are not needed here. In fact, we will only assume that $U$ is an $n\times n$ unitary matrix with all entries in its first row equal to $1/\sqrt{n}$. 
By definition, the ideal patterns are those with $s_0 = 1$ and $\bra{s}\hat{U}\ket{\eta} \neq 0$, so we may immediately express the heralding rate in the form \eqref{eq:heralding as operator}, then simplify as above to obtain \eqref{eq:heralding simplified fully}. 

Now, we write 
\begin{equation}\label{eq:herald proof 1}
    a_j \ket{1, \dots, 1} =  \prod_{r\neq j}\cre_r \vac.
\end{equation}
Applying $\hat{\Pi}$, we obtain 
\begin{equation}\label{eq:herald proof 2}
    \hat{\Pi}a_j \ket{1, \dots, 1} = \prod_{r\neq j} \left(\cre_r - \frac{1}{n}\sum_{v}\cre_v\right)\vac.
\end{equation}
Then the heralding rate becomes
\begin{equation}
    h_n(0) = \frac{1}{n}\sum_{i,j} \bvac\prod_{r'\neq i} a_{r'} \prod_{r\neq j}\left(\cre_r - \frac{1}{n}\sum_{v}\cre_v\right)\vac.
\end{equation}
We note that the terms in this sum are invariant up to relabeling of the modes. Then each term may be reduced to one of two cases, $i=j=0$ or $i=1, j=0$. 
The former occurs $n$ times and the latter $n^2 - n = n(n-1)$ times. 
We obtain
\begin{align}
    h_n(0) =&  \bvac\prod_{r'\neq 0} a_{r'} \prod_{r\neq 0}\left(\cre_r - \frac{1}{n}\sum_{v}\cre_v\right)\vac \label{eq:herald i=j}
    \\ &+ (n-1)\bvac\prod_{r'\neq 1} a_{r'} \prod_{r\neq 0}\left(\cre_r - \frac{1}{n}\sum_{v}\cre_v\right)\vac.\label{eq:herald i!=j}
\end{align}
Labeling the terms in \eqref{eq:herald i=j} and \eqref{eq:herald i!=j} as $B_n$ and $L_n$ respectively, we have
\begin{lemma}
\begin{align}
    B_n &= \left(\frac{-1}{n}\right)^{n-1}(n-1)! \,\sum_{t=0}^{n-1}  \frac{(-n)^t}{t!}
    \\ L_n &= \left(\frac{-1}{n}\right)^{n-1}(n-1)! \,\sum_{t=0}^{n-1}  (n-t-1) \frac{(-n)^t}{t!}.
\end{align}
\end{lemma}
Note that the lemma immediately gives
\begin{equation}
    h_n(0) = B_n + L_n = \left(\frac{-1}{n}\right)^{n-1}(n-1)! \,\sum_{t=0}^{n-1}  (n-t) \frac{(-n)^t}{t!},
\end{equation}
which proves Theorem~\ref{thm:herald ideal}. 
In the remainder of this section, we will prove the lemma, which is a straightforward counting argument. We focus on $B_n$, with $L_n$ being similar. 

We view $B_n$ as the inner product of $\cre_1\cdots \cre_{n-1}\vac$ and $\prod_{r>0}\left(\cre_r - \frac{1}{n}\sum_{v}\cre_v\right)\vac$. 
In particular, we expand the latter and sum the coefficients involving exactly one $\cre_{r'}$ for all $r'>0$ (and no $\cre_0$). 
Each term in the expansion is a product of some number $t$ of the $\cre_r$ and $n-1-t$ other $\cre_v$ chosen from some $\frac{1}{n}\sum_v \cre_v$. 
These terms vanish in the inner product unless the $\cre_r$ all have $r>0$ and the $\cre_v$ involve the remaining $n-1-t$ nonzero indices. 
Then we expand as follows, where $t$ is the number of $\cre_r$ used, as above, and $c_t$ is the number of such nonvanishing terms: 
\begin{align}
    B_n &= \bvac a_1\cdots a_{n-1}\prod_{r>0}\left(\cre_r - \frac{1}{n}\sum_{v}\cre_v\right)\vac
    \\&= \bvac a_1\cdots a_{n-1}\sum_{t=0}^{n-1} \left(\dfrac{-1}{n}\right)^{n-1-t} c_t \cre_1\cdots \cre_{n-1}\vac
    \\&= \sum_{t=0}^{n-1} \left(\dfrac{-1}{n}\right)^{n-1-t} c_t.
\end{align}
We note there are $\binom{n-1}{t}$ ways to choose $t$ distinct values of $\cre_r$, then $(n-1-t)!$ ways to choose which factors the remaining $\cre_v$ come from. 
This gives 
\begin{equation}
    \left(\dfrac{-1}{n}\right)^{n-1-t} c_t = \left(\dfrac{-1}{n}\right)^{n-1-t}\binom{n-1}{t}(n-1-t)! = \left(\dfrac{-1}{n}\right)^{n-1-t} \dfrac{(n-1)!}{t!} = \left(\dfrac{-1}{n}\right)^{n-1} (n-1)! \dfrac{(-n)^t}{t!},
\end{equation}
proving the result for $B_n$. 

\subsubsection{Ideal heralding asymptotics}\label{sec:maple appendix}
We now prove Theorem~\ref{thm:herald ideal} Part~\ref{part:herald ideal 2} giving the asymptotic formula for $h_n(0)$, following the proof of Kotesovec. We let $W(z)$ be the Lambert $W$ function, the principal branch of the multi-valued solution to the equation $W(z)e^{W(z)} = z$. 

First, by Theorem~\ref{thm:herald ideal} Part~\ref{part:herald ideal 1}, we have
\begin{equation}
    n^{n-1} h_n(0) = (-1)^{n-1} (n-1)! \sum_{t=0}^{n-1}(n-t)\dfrac{(-n)^t}{t!}.
\end{equation}
For $n\geq 1$, elementary algebraic manipulation recovers the formula for the exponential generating function of $-1/(1-W(-z))$ given on the OEIS \cite{oeis}. 
In other words, expanding
\begin{equation}
    \dfrac{-1}{1-W(-z)} = \sum_{n\geq 0} w_n z^n,
\end{equation}
we have
\begin{equation}
     h_n(0) = \dfrac{n!}{n^{n-1}}w_n.
\end{equation}
By Stirling's formula, this gives
\begin{equation}\label{eq:hw}
    h_n(0) \sim \dfrac{\sqrt{2\pi}n^{3/2}}{e^n} w_n.
\end{equation}
Now, following Kotesovec, we use the Maple library \emph{gdev}. Specifically, we use the procedure \emph{equivalent}, which applies saddle point methods to determine the asymptotic behavior of coefficients of generating functions \cite{salvy1991examples}. 
From the command 
\begin{equation}
    equivalent(-1/(1-LambertW(-z)), z,n,1),
\end{equation}
we find
\begin{equation}
    w_n \sim \dfrac{e^n}{4\sqrt{2\pi}n^{3/2}}.
\end{equation}
Then by \eqref{eq:hw}, 
\begin{equation}
    h_n(0)\sim \dfrac{\sqrt{2\pi}n^{3/2}}{e^n} \cdot \dfrac{e^n}{4\sqrt{2\pi}n^{3/2}} = \frac{1}{4},
\end{equation}
giving Theorem~\ref{thm:herald ideal} Part~\ref{part:herald ideal 2}.

\subsection{Error rate}\label{sec:app error}

We now discuss the calculation of error rates. 
The error rate is a conditional probability: the probability that the output photon is distinguishable, given successful heralding. 
In our setting, this is calculated as follows. 
First (extending notation to arbitrary input states as usual), we may write the conditional probability $e_n(\rho)$ as a quotient of $\overline{e}_n(\rho)$, the probability of successful heralding \emph{and} distinguishable output, divided by $h_n(\rho)$, the probability of successful heralding. 
In particular, we have
\begin{equation}\label{eq:error and herald}
    e_n(\epsilon) = \dfrac{\overline{e}_n(\epsilon)}{h_n(\epsilon)}.
\end{equation}
The denominator may be calculated as in Appendix~\ref{sec:heralding appendix}, so we focus on $\overline{e}_n(\epsilon)$ here. 
As with the other notation, we extend the notation $\ole_n$ to arbitrary input states and write
\begin{equation}\label{eq:error and herald decomposed}
    \ole_n(\epsilon) = \sum_{k=1}^n \binom{n}{k}\epsilon^k (1-\epsilon)^{n-k} \ole_n(\Phi_k)
\end{equation}
(noting that $\ole_n(\Phi_0) = 0$). 
Of course, $\ole_n(\Phi_k)$ is the average of $\ole_n(\ketbra{\eta})$ over all $\ket{\eta}$ of the form \eqref{eq:terms} with $k$ distinguishability errors. 
Let $\ket{\eta}$ be such a term. 
To calculate $\ole_n(\ketbra{\eta})$, 
we perform an internal-external measurement on $\hat{U}\ket{\eta}$ (see Section~\ref{sec:linear optics}). 
That is, we project onto the states \eqref{eq:internal fock} and 
calculate the probability that the external Fock state is ideal and the photon in the output mode $0$ has internal state $\ket{\xi_i}$, $i>0$. 

As above, by Lemma~\ref{lem:symmetrization} (or rather, its slight extension to internal-external measurements) we may write
\begin{equation}
    \ole_n(\ketbra{\eta}) = |\braket{\eta_+}{\eta}|^2 \ole_n(\ketbra{\eta_+}).
\end{equation}
Dividing both sides by $h_n(\ketbra{\eta})$,  applying \eqref{eq:h symm}, and rewriting $\ole_n(\rho) = e_n(\rho) h_n(\rho)$, we obtain
\begin{equation}\label{eq:symmetrized error rate}
    e_n(\ketbra{\eta}) = e_n(\ketbra{\eta_+}).
\end{equation}
This completes Part~\ref{step:symm} of Theorem~\ref{thm:symm body}. 

Similarly to the heralding rate, we may rewrite the expression for $\ole_n(\ketbra{\eta})$ in a way that removes the burden of checking that the output pattern satisfies the appropriate symmetries: 
\begin{equation}
    \ole_n(\ketbra{\eta}) = |\braket{\eta_+}{\eta}|^2 \sum_{i\geq 1} |(\bra{0}\otimes I) a_0[\xi_i] \hat{U} \ket{\eta_+}|^2.
\end{equation}
This requires the assumption that the relevant patterns in $\mc{S}_G$ are ideal, as discussed in Remark~\ref{remark:fourier prime assumption} and Appendix~\ref{sec:heralding appendix}. 
This completes the proof of Part~\ref{step:symm2} of Theorem~\ref{thm:symm body}. 

Again writing $\ole_n(\Phi_k) = e_n(\Phi_k)h_n(\Phi_k)$, we note that we have the following expression for the error rate: 
\begin{equation}\label{eq:error expanded}
    e_n(\epsilon) = \dfrac{\sum_{k=1}^n \binom{n}{k}\epsilon^k (1-\epsilon)^{n-k} e_n(\Phi_k)h_n(\Phi_k)}{\sum_{k=0}^n \binom{n}{k}\epsilon^k (1-\epsilon)^{n-k} h_n(\Phi_k)} 
    = \dfrac{\sum_{k=1}^n \binom{n}{k}\epsilon^k (1-\epsilon)^{n-k} e_n(\Phi_k)\widetilde{h}_n(\Phi_k)}{1 + \sum_{k=1}^n \binom{n}{k}\epsilon^k (1-\epsilon)^{n-k} \widetilde{h}_n(\Phi_k)},
\end{equation}
where $\widetilde{h}_n(\Phi_k) = h_n(\Phi_k)/h_n(\Phi_0)$. These expressions $\widetilde{h}_n(\Phi_k)$ are given conjectured bounds in Conjecture~\ref{conj:herald term bound}. 
Thus the error rate can be described in terms of these relatively well understood ratios of heralding rates and the quantities $e_n(\Phi_k)$. 
We discuss the efficient calculation of these quantities via simulation in Appendix~\ref{sec:simulation appendix}. 

\subsection{First order approximations}\label{sec:one}

We now return to the setting of the URS (or OBB) model of photon distinguishability. 
We assume that $\epsilon$ is independent of $n$ and consider approximations of $h_n(\epsilon)$ and $e_n(\epsilon)$ up to first order in $\epsilon$. 
As discussed in Section~\ref{sec:heralding}, we have 
\begin{equation}
    h_n(\epsilon) = h_n(0) + \epsilon n (h_n(\Phi_1) - h_n(0)) + O(\epsilon^2),
\end{equation}
where $\Phi_1$ is the evenly weighted probabilistic mixture of all terms of the form \eqref{eq:internal fock experiment} involving exactly one distinguishability error. 
Further, from \eqref{eq:error expanded}, we see that 
\begin{equation}\label{eq:error first order expansion}
    e_n(\epsilon) = \epsilon n e_n(\Phi_1)\widetilde{h}_n(\Phi_1) + O(\epsilon^2).
\end{equation}
Then the first-order approximations require only knowledge of $h_n(0)$, discussed above, and $h_n(\Phi_1)$, $e_n(\Phi_1)$. 
Further, we will find that all terms of $\Phi_1$ exhibit the same behavior with regard to distillation protocols. 
Without loss of generality, we will consider 
\begin{equation}
    \ket{\eta} = \cre_0[\xi_{1}]\cre_1[\xi_0]\cdots \cre_{n-1}[\xi_{0}]\vac.
\end{equation}

We will consider $U = F_n$ or $U=H_n$. 
Note that we will \emph{not} place any additional assumptions on the relationship between ideal patterns and those in $\mc{S}_G$, so we allow all $n = 2^r$ in the Hadamard case and all $n\geq 3$ in the Fourier case. 
More generally, we may consider any $n\times n$ unitary $U$ with first row identically equal to $1/\sqrt{n}$ (so that Appendix~\ref{sec:ideal herald} applies) and with a symmetry group $G$ of order $n$ acting transitively on the $n$ modes. 
In particular, as discussed in Appendix~\ref{sec:fourier abelian}, these results hold when $U$ is the Fourier transform corresponding to any abelian group $G$ of order $n$. 
(Note that $F_n$ and $H_n$ are special cases.) 
In particular, for $P\in G$ and $\ket{\eta}$ as above, all $P\ket{\eta}$ are mutually orthogonal, 
corresponding to permuting the single ``distinguishable" photon to each of the modes $0, \dots, n-1$. 
By the discussion in Section~\ref{sec:symmetry general}, we obtain the associated $G$-symmetrization
\begin{equation}\label{eq:eta plus 1}
    \ket{\eta_+} = \dfrac{1}{\sqrt{n}}\sum_{i=0}^{n-1} \cre_0[\xi_0]\cdots \cre_{i-1}[\xi_0]\cre_i[\xi_1]\cre_{i+1}[\xi_0]\cdots \cre_{n-1}[\xi_0]\vac,\textnormal{ with }\braket{\eta_+}{\eta} = \frac{1}{\sqrt{n}}.
\end{equation}
We note that $\ket{\eta_+}$ is invariant under \emph{arbitrary} permutations of the external modes, not just those in $G$. 
In fact, by direct computation we see that
\begin{align}\label{eq:pure tensor 1 error}
    \ket{\eta_+} &= \textnormal{symm}\left(\ket{0, 1, \dots, n-1}_{1Q}\right)\otimes \textnormal{symm}\left(\ket{\xi_1, \xi_0, \dots, \xi_0}\right)
    \\&= \ket{1, \dots, 1}_{2Q}\otimes\left(\ket{\xi_1, \xi_0, \dots, \xi_0} + \ket{\xi_0, \xi_1, \xi_0, \dots, \xi_0} + \dots + \ket{\xi_0, \dots, \xi_0, \xi_1}\right).
\end{align}
In other words, $\ket{\eta_+}$ is a perfectly indistinguishable state, a pure tensor of the Fock state $\ket{1, \dots, 1}_{2Q}$ and some symmetric internal state. 
Then the results of Appendix~\ref{sec:ideal herald} apply here, and we have 
\begin{equation}
    h_n(\ketbra{\eta}) = \frac{1}{n}h_n(\ketbra{\eta_+}) = \frac{1}{n}h_n(0).
\end{equation}
Further, since all terms of $\Phi_1$ have the same symmetrization (up to relabeling of the $\ket{\xi_i}$, $i>0$), we have 
\begin{equation}\label{eq:heralding phi1}
    h_n(\Phi_1) = \frac{1}{n}h_n(0),
\end{equation}
or equivalently $\widetilde{h}_n(\Phi_1) = 1/n$. 
This gives Theorem~\ref{thm:herald first order}. 

Next we consider the output error rate. In this setting, we may characterize $e_n(\ketbra{\eta_+})$ as the probability that, given successful heralding of a single output photon, the output photon has internal state $\ket{\xi_1}$. Since $\ket{\eta_+}$ is a pure tensor as described above, and linear optics and PNRD only act on the external part of the state, the post-heralding state is still a pure tensor with the same internal part given in \eqref{eq:pure tensor 1 error}. 
This internal part is fully symmetric, so the probability of any given photon having internal state $\ket{\xi_1}$ is exactly $1/n$. 
Then by \eqref{eq:symmetrized error rate}, 
\begin{equation}
    e_n(\Phi_1) = e_n(\ketbra{\eta_+}) = \frac{1}{n}.
\end{equation}
Combined with \eqref{eq:error first order expansion} and \eqref{eq:heralding phi1}, this proves Theorem~\ref{thm:error rate}. 

\section{Fourier transform for finite abelian groups}\label{sec:fourier abelian}
In this section, we consider distillation protocols for a more general family of unitaries than the Fourier and Hadamard matrices above. 
In particular, we consider
\begin{equation}\label{eq:fourier general}
    F_{(n_1, \dots, n_\ell)} = F_{n_1}\otimes \cdots\otimes F_{n_\ell},
\end{equation}
where $n=n_1\cdots n_\ell$ and the $n_i\geq 2$ are arbitrary. 
We have $F_n = F_{(n)}$, by definition. Further, since $F_2 = H$, we have $H_{2^r} = F_{(2, 2, \dots, 2)}$. 
More generally, the unitary \eqref{eq:fourier general} corresponds to the Fourier transform of a finite abelian group $G$ of order $n = n_1 \cdots n_\ell$, namely
\begin{equation}
    G\cong \mathbb{Z}_{n_1} \times \cdots \times \mathbb{Z}_{n_\ell}.
\end{equation}
Here $\times$ is the direct product of groups and $\mathbb{Z}_m$ is the cyclic group of order $m$. 
In particular, $F_n$ corresponds to $\mathbb{Z}_n$ and $H_{2^r}$ corresponds to $\mathbb{Z}_2\times\cdots\times\mathbb{Z}_2$. 
We observe that as written, there are some redundancies for different choices of $n_1, \dots, n_\ell$; we discuss this further below.

Note that in Appendices~\ref{sec:ideal herald} and \ref{sec:one}, where we compute the $0$th and $1$st order terms of $h_n(\epsilon)$ and $e_n(\epsilon)$, very few properties of the unitaries $F_n$ and $H_n$ are used. 
In particular, to calculate $h_n(0)$, we observe in 
Appendix~\ref{sec:ideal herald} that we only need $U$ to be an $n\times n$ unitary matrix with entries in the first row identically equal to $1/\sqrt{n}$. 
To calculate $h_n(\Phi_1), e_n(\Phi_1)$ as in Appendix~\ref{sec:one}, we only need $U$ to have a symmetry group $G$ of order $n$ (in the sense of Appendix~\ref{sec:symmetry general}) that acts transitively on the $n$ modes. 
Then the following theorem implies that these results apply to general $F_{(n_1, \dots, n_\ell)}$ as well: 

\begin{theorem}\label{thm:abelian}
Let $n\geq 2$ be a positive integer with $n_1\cdots n_\ell = n$. 
We have the following: 
\begin{enumerate}
    \item With a suitable labeling of the modes (i.e., a change of basis via a permutation matrix), the first row of $U = F_{(n_1, \dots, n_\ell)}$ is identically equal to $1/\sqrt{n}$. 
    \item $F_{(n_1, \dots, n_\ell)}$ has a symmetry group $G$ isomorphic to $\mathbb{Z}_{n_1} \times \cdots \times \mathbb{Z}_{n_\ell}$. 
    \item $G$ transitively permutes the modes. 
\end{enumerate}
In particular, up to first order in $\epsilon$, the heralding and error rates of the distillation protocol for $U=F_{(n_1, \dots, n_\ell)}$ depend only on $n$. 
\end{theorem}
\begin{proof}
For each $i$ with $1\leq i\leq \ell$, we view $F_{n_i}$ as an operator on $\mathbb{C}^{n_i}$. 
Then $F_{(n_1, \dots, n_\ell)}$ is naturally an operator on $\mathbb{C}^{n_1}\otimes\cdots\otimes\mathbb{C}^{n_\ell}\cong \mathbb{C}^n$, which has basis $\ket{m_1}\otimes\cdots\otimes\ket{m_\ell}$ $(0\leq m_i < n_i)$.  
This may be identified with the usual basis for $\mathbb{C}^n$ by 
\begin{equation}\label{eq:fourier basis map}
    \ket{m_1}\otimes\cdots\otimes\ket{m_\ell}\mapsto \ket{m_1 + n_1 m_2 + n_1 n_2 m_3 + \dots + (n_1\cdots n_{\ell-1})m_\ell},
\end{equation}
where we interpret the ket on the right-hand side modulo $n$. 
In this basis, $\ket{0}$ corresponds to $\ket{0}^{\otimes \ell}$, so for $g = m_1 + n_1 m_2 + n_1 n_2 m_3 + \dots + (n_1\cdots n_{\ell-1})m_\ell$, it is clear that 
\begin{equation}
    \bra{0}F_{(n_1, \dots, n_\ell)}\ket{g} = \prod_i \bra{0}F_{n_i}\ket{m_i} = \prod_i \frac{1}{\sqrt{n_i}} = \frac{1}{\sqrt{n}},
\end{equation}
proving the first claim. 
For the second, we note that by Appendix~\ref{sec:symmetry}, each $F_{n_i}$ factor has a symmetry group $G_i\cong \mathbb{Z}_{n_i}$ with generating permutation matrices $P_i\in G_i$ and diagonal matrices $D_i$ satisfying $F_{n_i}P_i = D_i F_{n_i}$. 
As with the Hadamard case, this directly extends to the tensor product, with 
\begin{equation}\label{eq:symm tensor}
    \left(F_{n_1}\otimes \cdots\otimes F_{n_\ell}\right)\left(P_1^{q_1}\otimes\cdots\otimes P_{\ell}^{q_\ell}\right) = \left(D_1^{q_1}\otimes\cdots\otimes D_{\ell}^{q_\ell}\right)\left(F_{n_1}\otimes \cdots\otimes F_{n_\ell}\right).
\end{equation}
This proves the second claim, with $G = G_1\times\cdots\times G_\ell$. 
(Note this same reasoning gives a version of the ZTL for the general case, discussed below.) 
Finally, we observe that the action of $G$ on the modes is transitive. 
This is clear from the description $\ket{m_1}\otimes\cdots\otimes\ket{m_\ell}$ of the basis for $\mathbb{C}^n$. In particular, $G_i$ freely permutes the $i$th factor without affecting any of the others; then for arbitrary basis vectors $\ket{j_1}, \ket{j_2}$, one may find a suitable element $P\in G$ such that $P\ket{j_1} = \ket{j_2}$.
\end{proof}

We now give the analogue of the ZTL for general $F_{(n_1, \dots, n_\ell)}$. 
Let $g_0, \dots g_{n-1}\in\{0, \dots, n-1\}$ describe a Fock state $\ket{s}$ as usual, and by \eqref{eq:fourier basis map} identify each $g_j$ with a tuple $(m_1^{(j)}, m_2^{(j)}, \dots, m_\ell^{(j)})$. 
Consider \eqref{eq:symm tensor} with $q_i=1$ and all other $q_j=0$. 
Let $\omega_{n_i}$ be our usual choice of primitive $n_k$th root of unity. 
Recall from Appendix~\ref{sec:fourier symm} that $D_i$ is diagonal with $k$th diagonal entry $\omega_{n_i}^{k}$. 
Then 
\begin{align}
    \bra{s}\left(F_{n_1}\otimes \cdots\otimes F_{n_\ell}\right)\ket{1, \dots, 1} &= \bra{s}\left(F_{n_1}\otimes \cdots\otimes F_{n_\ell}\right)\left( I^{\otimes (i-1)} \otimes P_i \otimes I^{\ell-i}\right)\ket{1, \dots, 1}
    \\&= \bra{s}\left( I^{\otimes (i-1)} \otimes D_i \otimes I^{\ell-i}\right)\left(F_{n_1}\otimes \cdots\otimes F_{n_\ell}\right)\ket{1, \dots, 1}
    \\&= \omega_{n_i}^{-\sum_{j=0}^{n-1} m_i^{(j)}}\bra{s}\left(F_{n_1}\otimes \cdots\otimes F_{n_\ell}\right)\ket{1, \dots, 1},
\end{align}
so as before, for this quantity to be non-vanishing we require, for all $i$, 
\begin{equation}\label{eq:fourier ztl componentwise}
    \sum_{j=0}^{n-1} m_i^{(j)} \equiv 0 \mod n_i.
\end{equation}
(Note that the sum has $n$ terms, not $n_i$ terms.) 
This is the generalized Zero Transmission Law. 
Note that this is very natural in the tensor product basis, but less so in terms of the standard basis for $\mathbb{C}^n$ constructed in \eqref{eq:fourier basis map}. 
Later in this section we will discuss the compatibility of this result with the ZTL for $F_n$, using a different basis for $\mathbb{C}^n$. 

We note that in the case with all $n_i = 2$, we precisely recover the description of the ZTL in terms of decompositions into binary strings for the Hadamard unitaries, as discussed in \cite{crespi2015suppression} and Appendix~\ref{sec:hadamard symm} above. 
In particular, the decomposition of an arbitrary mode $g_j$ into the form $m_1^{(j)} + 2m_2^{(j)} + 2^2 m_3^{(j)} + \dots + 2^{\ell-1} m_\ell^{(j)}$, where all $m_i^{(j)}\in\{0,1\}$, is precisely its representation in base $2$, and \eqref{eq:fourier ztl componentwise} demands that all $\sum_j m_i^{(j)} \equiv 0\mod 2$. 

Note that for any $F_{(n_1, \dots, n_\ell)}$, if we are given the locations $g_0, \dots, g_{n-2}$ of only $n-1$ photons, there is a unique mode $g_{n-1}$ that completes the pattern to one satisfying the generalized ZTL. In particular, expressing each $g_j$ for $0\leq j < n-1$ in terms of the $m_i^{(j)}$ as above, \eqref{eq:fourier ztl componentwise} uniquely determines $m_i^{(n-1)}$ for all $i$. This then uniquely determines $g_{n-1}$ by \eqref{eq:fourier basis map}. 
This is used in the proof of Theorem~\ref{thm:loss}. 

We now discuss different ways of expressing the unitaries of \eqref{eq:fourier general}. 
We introduce an additional bit of notation: for $c$ relatively prime to $n_j$, define $F_{n_j}^{(c)}$ to have $(a,b)$ entry $\omega_{n_j}^{abc}$; in other words, we make a different choice of primitive $n_j$th root of unity, replacing $\omega_{n_j} = \exp(2\pi i/n_j)$ with $\omega_{n_j}^c = \exp(2\pi i c/n_j)$. 
These variants are not similar in general (in the technical sense that they are not equivalent up to a change of basis); however, we have $F_{n_j}^{(c)}Q = F_{n_j}$ for some permutation matrix $Q$. (The permutation taking column $c$ to column $1$, column $2c$ to column $2$, etc., with indices taken modulo $n_j$.) 
In terms of the corresponding linear optical unitary, as long as we only apply $\hat{F}_{n_j}^{(c)}$ to mode-symmetric states such as $\ket{1, \dots, 1}$, we see that the possible output patterns (and their amplitudes) are unchanged. 
In the context of Theorem~\ref{thm:abelian}, for the cyclic permutation $P_j$ generating $G_{i}$, we have $F_{n_j}^{(c)}Q^\dagger P_j Q = D_j F_{n_j}^{(c)}$. 
(Note $D_j$ is unchanged.) 
Then the corresponding symmetry group is still isomorphic to $G_j$, with the isomorphism given by $P\mapsto Q^\dagger P Q$, and we get exactly the same ZTL and suppression laws. 
For our purposes, then, we may view the $F_{n_j}^{(c)}$ as interchangeable as long as $c$ is relatively prime to $n_j$. 

Earlier, we commented that $F_{n_j}$ corresponded to the Fourier transform on $\mathbb{Z}_{n_j}$, which is a map between two vector spaces of dimension $n_j$. These spaces are isomorphic, but not canonically so; 
to write $F_{n_j}$ as a matrix over $\mathbb{C}^{n_j}$ requires one to arbitrarily choose a primitive $n_j$th root of unity. 
Different choices give the various $F_{n_j}^{(c)}$. 
Thus when discussing the correspondence between finite abelian groups and matrices for their discrete Fourier transforms, we will need to be flexible about the choice of root of unity. 
In particular, the unitary in \eqref{eq:fourier general} may be described by different choices of $n_1, \dots, n_\ell$, as long as we are comfortable with additional permutation matrices $Q$ as above. 
By the fundamental theorem of finite abelian groups, we may always express the finite abelian group $G$ in the form $G\cong \mathbb{Z}_{p_1^{r_1}}\otimes\cdots\otimes\mathbb{Z}_{p_t^{r_t}}$, where $p_1^{r_1}\geq\dots\geq p_t^{r_t}$ and the $p_i$ are (not necessarily distinct) primes. 
Correspondingly, up to a relabeling of the modes we may express $F_{(n_1, \dots, n_\ell)}$ as $F_{(p_1^{r_1}, \dots, p_{t}^{r_t})}Q$, where $Q$ is a permutation matrix, corresponding to potentially choosing different roots of unity as above. 

To be concrete, we consider the correspondence between $F_{n_1}\otimes F_{n_2}$ and $F_{n}$, where $n_1, n_2$ are relatively prime and $n=n_1 n_2$. (This gives the general case by induction.) 
By Bezout's Theorem, we find $c_1, c_2$ such that $c_1 n_1 + c_2 n_2 = 1$. 
We note that $c_1$ is uniquely determined modulo $n_2$ and $c_2$ is uniquely determined modulo $n_1$. 
We have a canonical ring isomorphism $\mathbb{Z}_{n_1}\times \mathbb{Z}_{n_2}\rightarrow \mathbb{Z}_n$ by 
\begin{equation}\label{eq:basis relpr}
    (k_1, k_2)\mapsto k_1 c_2 n_2 + k_2 c_1 n_1.
\end{equation}
The inverse is simply $k\mapsto (k\mod n_1, k\mod n_2)$.  
We then have $F_{n_1}^{(c_2)}\otimes F_{n_2}^{(c_1)}$ similar to $F_{n}$, 
with the change of basis simply being a relabeling of the modes according to \eqref{eq:basis relpr}, specifically 
$\ket{k_1}\otimes\ket{k_2}\leftrightarrow \ket{k_1 c_2 n_2 + k_2 c_1 n_1}$. 
(Note this is a different labeling from \eqref{eq:fourier basis map} in general!) 
In particular, this gives a \emph{canonical} way of identifying $F_{n_1}^{(c_2)}\otimes F_{n_2}^{(c_1)}$ and $F_{n}$. 
Further, as discussed above, the suppression laws are the same as for $F_{n_1}\otimes F_{n_2}$ when expressed in terms of the natural basis for $\mathbb{C}^{n_1}\otimes\mathbb{C}^{n_2}$. 
To express them in terms of the natural basis for $\mathbb{C}^n$ using \eqref{eq:basis relpr}, let $g_0, \dots, g_{n-1}$ be mode labels in $0, \dots, n-1$, with each $g_j = k_1^{(j)}c_2 n_2 + k_2^{(j)}c_1 n_1$. By \eqref{eq:fourier ztl componentwise}, we have $\sum_{j=0}^{n-1} k_i^{(j)} \equiv 0\mod n_i$. 
Then we see that
\begin{equation}
    \sum_{j=0}^{n-1} g_j = \left(\sum_{j=0}^{n-1} k_1^{(j)}\right)c_2 n_2 + \left(\sum_{j=0}^{n-1} k_2^{(j)}\right)c_1 n_1 \equiv 0\mod n.
\end{equation}
Then if we label the modes according to \eqref{eq:basis relpr}, we recover the standard ZTL for $F_n$. 

For example, we consider the case of $F_6$. 
Writing $n_1=2, n_2=3$, we have corresponding $c_1=-1, c_2=1$, so that $F_6$ and $F_2\otimes F_3^{(-1)}$ are similar, with both having the same suppression laws as $F_2\otimes F_3$. 
We may explicitly check that for $Q$ fixing column $0$ and permuting columns $1$ and $2$, 
we have $F_3^{(-1)} Q = F_3$. The similarity between $F_6$ and $F_2\otimes F_3^{(-1)}$ is seen by choosing $\ket{j}$ to correspond to the $j$th element of the sequence
\begin{equation}
    (0,0), (1,2), (0,1), (1,0), (0,2), (1,1).
\end{equation}

In summary, it suffices to consider the case of \eqref{eq:fourier general} where the $n_i$ are weakly decreasing prime powers. 
From this perspective, for each of $n=8,9,12$, we obtain one new protocol beyond the standard $F_n$ and $H_n$ (the latter only applicable when $n$ is a power of $2$), namely $F_{(4,2)}$, $F_{(3,3)}$, and $F_{(6,2)}$. 
For $n=16$, there are $5$ total options, with $3$ of them new: $F_{(8,2)}$, $F_{(4,4)}$, and $F_{(4,2,2)}$. 

In Fig.~\ref{fig:n_8_all_protocols}  we show the performance for all three $n=8$ protocols (i.e., including the $F_{(4,2)}$ protocol as well as $F_8=F_{(8)}$ and $H_8=F_{(2,2,2)}$ studied in the main text). Note that the $F_{(4,2)}$ protocol falls in between $F_8$ and $H_8$ in terms of performance, under both error models. Therefore, if one is agnostic to the noise model, it may in fact be advantageous to use the $F_{(4,2)}$ protocol (if one knows the noise model, then either $F_8$ or $H_8$ is optimal).

\begin{figure}
    \centering
    \includegraphics[width=0.75\columnwidth]{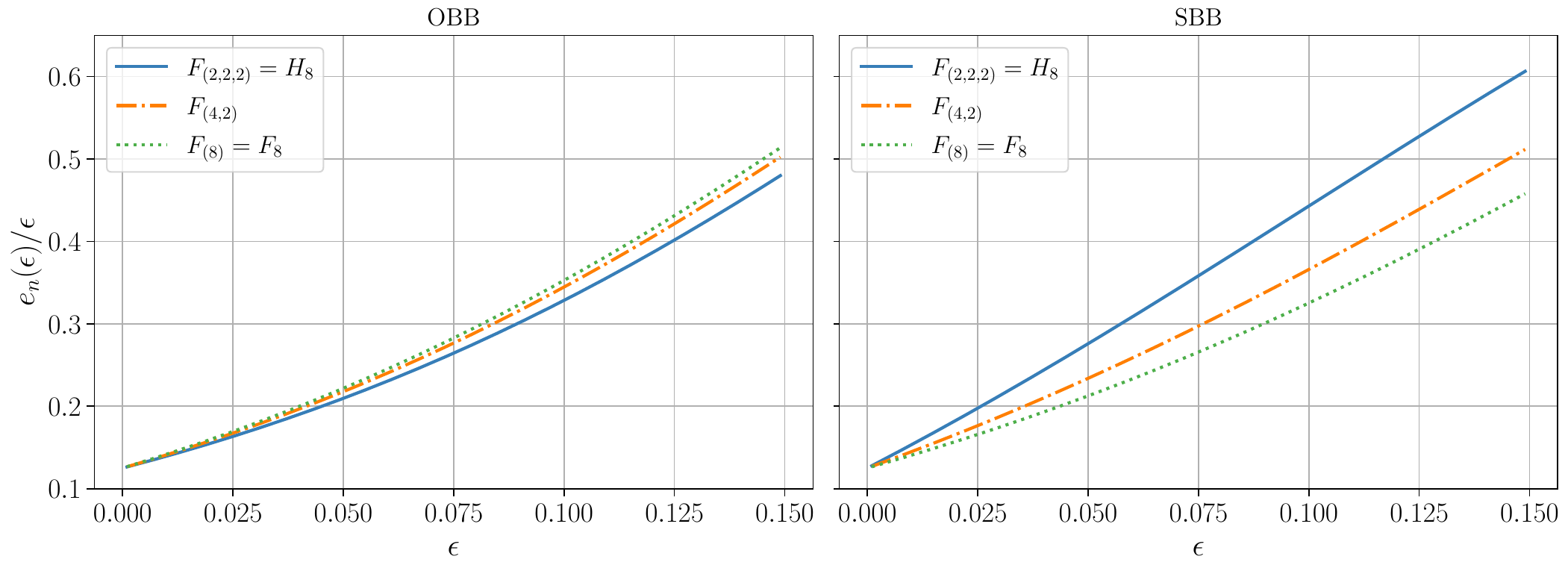}
    \caption{Performance comparison of all $n=8$ distillation protocols under the OBB and SBB noise models. The heralding rates $h_n(\epsilon)$ do not vary much between the models (see Fig.~\ref{fig:n_8_dist} for $F_8, H_8$).}
    \label{fig:n_8_all_protocols}
\end{figure}

\section{Haar random case}\label{sec:haar}

For comparison with the distillation schemes outlined in the main text and how they utilize interference to assist in distilling errors, we provide analysis in the case where random unitaries are considered for distillation. Interestingly, although these unitaries have no special structure, they can still be used for distillation purposes. However, as we will see, the protocols based around random matrices are significantly less 
resource efficient
than the Fourier protocol of the main text. This highlights the importance of the specific constructive and destructive interference in the latter scheme.

We consider generating Haar random unitaries from $U(n)$ (where the number of photons, $n$, is the same as the number of modes) and, as a warm up, running the same distillation protocols as outlined in this work. 
In particular, we evolve the all 1 state $|\bar{1}\rangle = \ket{1, \dots, 1}$ under a sampled random matrix, compute the ideal heralding patterns that have exactly 1 photon out in mode 0 (which as the matrix is random, it will typically include all possible patterns), and study how the heralding probabilities scale. By concentration of measure \cite{ledoux_concentration_2005_2}, for large enough $n$ we expect a typical sample to be well described by the average, and so we can use Haar averaging to compute quantities of interest.

First we can ask what the heralding rate $h_n(0)$ is. This is the probability that a single photon exits in mode 0, from the initial all 1 state. We compute this as follows:
\begin{equation}
    h_n(0) =  \int d\mu_{\hat{U}} \langle \bar{1}| \hat{U}^\dag \left( |1\rangle \langle 1|\otimes \mathbf{I}_{n-1,n-1} \right) \hat{U} |\bar{1}\rangle.
    \label{eq:haar_herald_integral}
\end{equation}
That is, the state $|\bar{1}\rangle$ is evolved under unitary $\hat{U}$, 
we compute the probability of PNRD giving an output pattern with a single photon in mode $0$, 
then we average over all unitaries. 
Note that here, the $\hat{U}$ are $d_{n, n }\times d_{n,n}$ matrices, where $d_{n,m}$ is the dimension of the space of $n$ photons in $m$ modes [$d_{n,m} = {n+m-1 \choose n}$]. The identity matrix $\mathbf{I}_{n-1,n-1}$ on $n-1$ modes and $n-1$ photons has dimension $d_{n-1, n-1}$.

Whilst the underlying linear optical unitary is drawn from a Haar random distribution on $U(n)$, it does not imply the $d_{n,n} \times d_{n,n}$ matrices $\hat{U}$ are Haar random in $U(d_{n,n})$. However, as shown in Ref.~\cite{saied2024advancing}, linear optical unitaries are a continuous 1-design. As such, we can still treat the integral in the above equation as a Haar random integral ($\int d\mu_U U X U^\dag = \mathrm{Tr}[X]\mathbb{I}/d$), which we can analytically evaluate as:
\begin{equation}
    h_n(0) = \frac{\mathrm{Tr}[\mathbf{I}_{n-1, n-1}]}{d_{n,n}} = \frac{d_{n-1, n-1}}{d_{n,n}} = \frac{1}{4}\frac{1}{1 - (2n)^{-1}} > 1/4.
    \label{eq:haar_herald}
\end{equation}
Of course this makes intuitive sense; for a random unitary, one expects any particular state is equally likely with probability $1/d_{n,n}$ (when averaged over unitaries), and the total number of states with 1 photon in mode 0 is $d_{n-1, n-1}$.
Notice that as with Theorem~\ref{thm:herald ideal} and Conjecture \ref{conj:herald} (pertaining to the Fourier and Hadamard matrices), this is decreasing monotonically in $n$ and tends to 1/4. We additionally conjecture that Eq.~\eqref{eq:haar_herald} is an upper bound for the same sized protocol with the Fourier matrices. See Fig.~\ref{fig:haar_vs_fourier_herald} for some numerics.

\begin{figure}
    \centering
    \includegraphics[width=0.6\columnwidth]{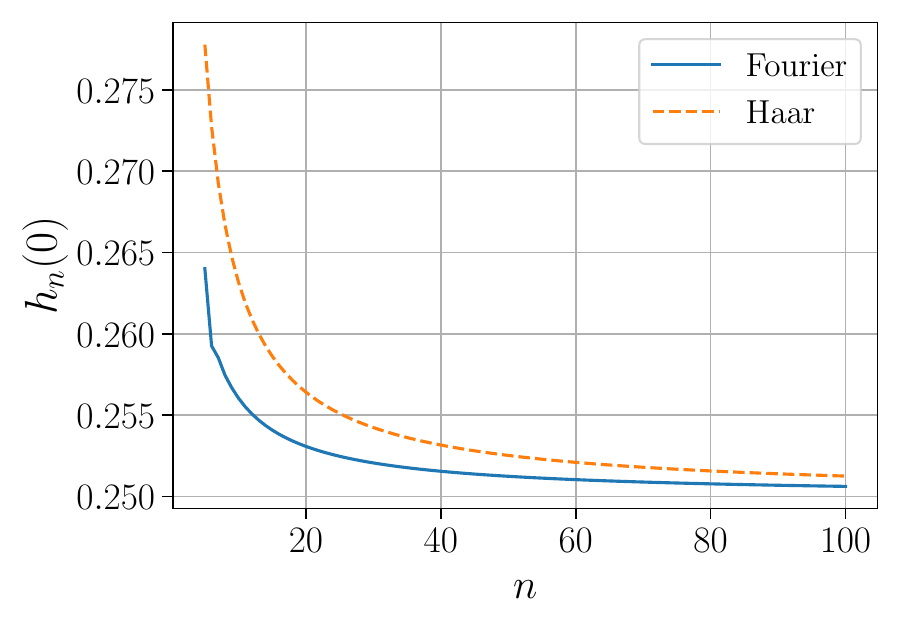}
    \caption{Haar (Eq.~\eqref{eq:haar_herald}) and Fourier (Th.~\ref{thm:herald ideal}) distillation schemes heralding rate at 0 error, as a function of $n$ (i.e., the probability a single photon exits from mode 0, in the absence of error).}
    \label{fig:haar_vs_fourier_herald}
\end{figure}

We can further calculate what happens at first order errors (where there are $n-1$ identical photons, and 1 in an orthogonal internal state). Now there are four integrals similar to Eq.~\eqref{eq:haar_herald_integral}; we care about getting 1 or 0 photons out starting with $n-1$ identical photons in $n$ modes, and conversely 0 or 1 photons out starting with 1 (error) photon. Multiplying these as written gives respectively the probability of observing an ideal photon out of the scheme, or an error photon, in the case of exactly one error. These can be computed in the same manner as Eq.~\eqref{eq:haar_herald_integral}, and combining the results gives two probabilities:
\begin{equation}
    P_{ideal} = \frac{d_{n-2, n-1}}{d_{n-1, n}} \frac{d_{1,n-1}}{d_{1,n}} = \frac{1}{4} \frac{(1-1/n)^2}{1 - \frac{3}{2n}},\;\;P_{err} = \frac{d_{n-1, n-1}}{d_{n-1, n}} \frac{d_{0,n-1}}{d_{1,n}} = \frac{1}{2n}.
\end{equation}
The first term corresponds to the probability an ideal photon exits mode 0, and the second term the probability the error (distinguishable) photon exits mode 0, given exactly 1 distinguishability error at the input.

We can compute the output fidelity of the scheme at first order using $h_n(0), P_{ideal}, P_{err}$:
\begin{equation}
    f_n(\epsilon):=1-e_n(\epsilon) \approx \frac{h_n(0)(1-\epsilon)^n + n\epsilon(1-\epsilon)^{n-1}P_{ideal}}{h_n(0)(1-\epsilon)^n + n\epsilon(1-\epsilon)^{n-1}(P_{ideal} + P_{err})}  \approx 1 - \frac{\epsilon}{2h_n(0)}.
    \label{eq:haar_fidel}
\end{equation}
Since $h_n(0) < 0.5$ for $n\ge 3$, we see the output error $e_n(\epsilon) \approx \epsilon/2h_n(0)$ has actually \textit{increased}, and for large enough $n$, $e_n(\epsilon) \approx 2\epsilon$. We verify this scaling numerically in Fig.~\ref{fig:haar_dist}.
Note that we numerically obtain similar scaling $e_n(\epsilon)\approx 2\epsilon$ in the Fourier case if we post-select on all possible outcomes with $1$ output photon, as opposed to only ideal patterns (see Fig.~\ref{fig:haar_dist}). 
Therefore, this result is more about the fact that we are post-selecting on all possible patterns, not necessarily related to interference or lack thereof in the Haar vs. Fourier case. 

\begin{figure}
    \centering
    \includegraphics[width=0.7\columnwidth]{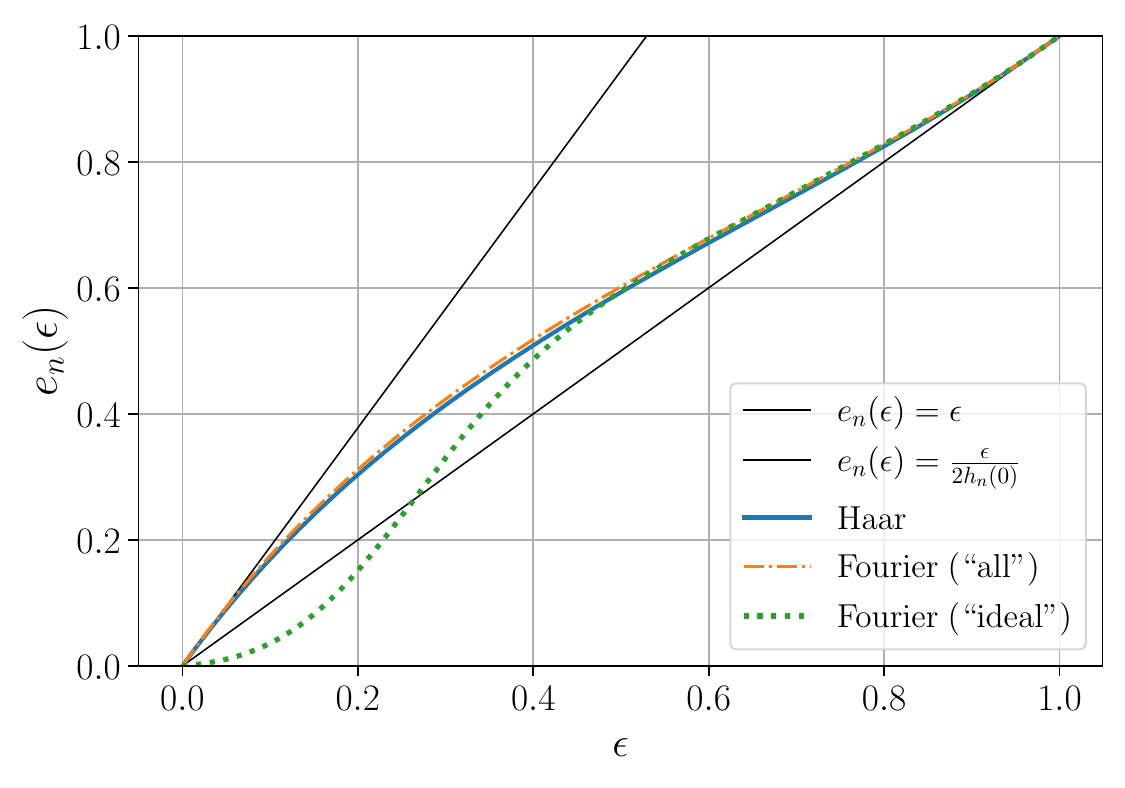}
    \caption{`Distillation' from a single instance drawn from the Haar random distribution for $n=9$, with the OBB error model, where we post-select on all possible patterns (solid blue line). Distillation is useful when $e_n(\epsilon) < \epsilon$ (lower solid black line), which is never the case here. For small $\epsilon \lesssim 0.1$ we see it follows the upper solid line,  Eq.~\eqref{eq:haar_fidel}. For comparison, we also include the case where we post-select on the ideal patterns in the Fourier protocol at $n=9$ (dotted line), which we see for $\epsilon\lesssim 0.25$ results in $e_n(\epsilon)<\epsilon$. If we run the Fourier protocol but allow one to post-select on all valid patterns (dash-dot), the scaling is similar to the Haar case.}
    \label{fig:haar_dist}
\end{figure}

A more interesting comparison is where we take a particular (random) unitary, and use only the highest weight heralding patterns resulting from the evolution of the initial all 1 state (such that there is 1 photon out in mode 0). The intuition here is that from a random linear optical unitary, one expects the probabilities to be roughly Porter-Thomas distributed, as demonstrated in Fig.~\ref{fig:haar_pt} (also see \cite{aaronson_arkhipov,nezami_permanent_2021_1, go_exploring_2024}). As such, there will be a set of patterns that will have a relatively high weight (i.e., above the mean probability, $1/d$), distributed according to $de^{-d p}$, where $p$ is the probability and $d=d_{n,n}$ (the expected number of states with probability greater than $p^*$ is $\approx de^{-dp^*}$). The high weight patterns undergo a degree of constructive interference, akin to the patterns used for Fourier/Hadamard distillation. As such, one may expect these patterns could be used for photon distillation, despite the underlying unitary being random.

\begin{figure}
    \centering
    \includegraphics[width=0.65\columnwidth]{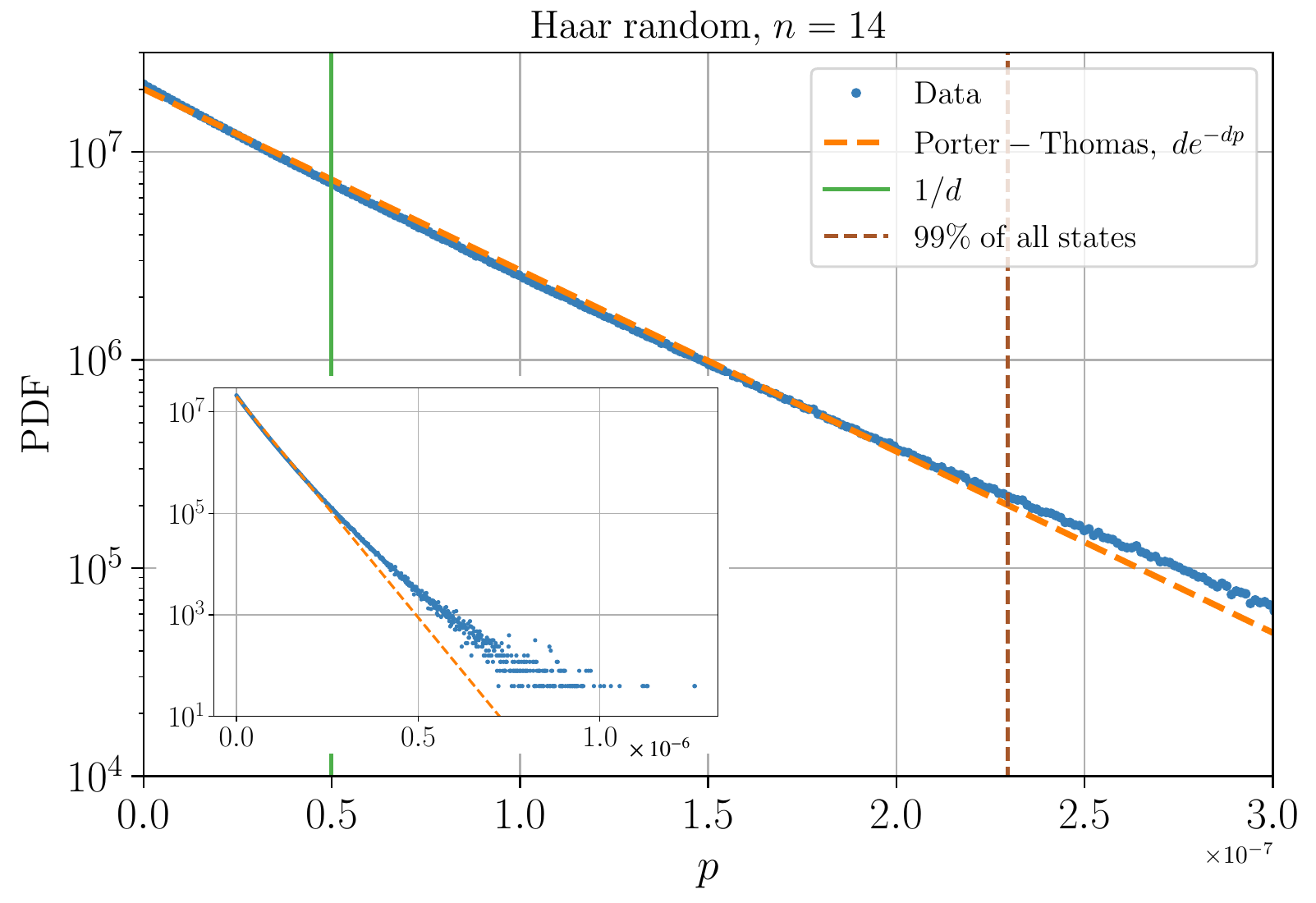}
    \caption{Probability density function (PDF) of the output probabilities $p$ from linear optical evolution of a sampled Haar random unitary in $U(n)$. Here the number of photons and modes is $n=14$, with initial state $|1\rangle^{\otimes n}$. We see the data points very closely follow the conjectured Porter-Thomas distribution. We also plot for reference the vertical solid line $p=1/d$ (where $d=d_{n,n}$ is the dimension), as well as the vertical dash line, left of which contains $99\%$ of all states, approximately where the data begins to diverge from Porter-Thomas (likely due to random sampling/small size effects). In the main figure we exclude the long sparse tail to focus on the majority of the data. The inset contains all data (the points that visibly diverge from the Porter-Thomas line account for less than 1\% of all states). The histogram data is constructed using 1000 equally sized bins over the range of probabilities $p$.}
    \label{fig:haar_pt}
\end{figure}

We provide numerics for this protocol in Fig.~\ref{fig:haar_high_weight}, which shows that if one post-selects on only the several highest-weighted patterns for that particular unitary, up to around $30\%$ of the total weight, 
distillation does work (i.e. it reduces the error, when the initial error $\epsilon$ is small). 
However, as can also be seen in the figure (right), the heralding probability decreases as we reduce the number of post-selection patterns and thus the error $\tilde{e}_n(\epsilon)$ (we put `tilde' above to differentiate this from the protocols of the main text). This is in contrast to the Fourier protocol in the main text, where at low error, the heralding probability is close to $1/4$ for all $n$, while the error reduces by a factor of $n$. Moreover, here the error reduction at small $\epsilon$ seems to be at most a factor of $\approx 2$ (i.e., at small $\epsilon$, $\tilde{e}_n(\epsilon) \approx \epsilon/2$). Again, we can contrast this to the Fourier protocol, which has $e_n(\epsilon) \approx \epsilon/n$.

We can provide some analytics on this protocol too. First, consider the case where one uses only the highest weight pattern for post-selection (this corresponds to the x-axis at 0 in Fig.~\ref{fig:haar_high_weight}). From the Porter-Thomas distribution, we can estimate the largest expected probability to be around $p^* = \frac{1}{d}\log d$ (i.e., the expected number of states with probability greater than $p^*$ is 1). However, since the fraction of states that have 1 photon out in mode 0 is $d_{n-1,n-1}/d_{n,n}$ (see Eq.~\eqref{eq:haar_herald}) we need to solve instead the following equation:
\begin{equation}
    N_{p>p^*} = d_{n,n}e^{-d_{n,n} p^*} =\frac{d_{n,n}}{d_{n-1,n-1}}\; \implies p^* = \frac{1}{d_{n,n}} \log d_{n-1,n-1}.
\end{equation}
This is the expected probability of the highest weight state with 1 photon out in mode 0. Whilst the distillation scheme that selects only the highest weighted pattern has the greatest error reduction, the heralding probability decreases with the dimension.

If instead we select the highest weighted $k$ patterns (with 1 photon out in mode 0), by similar reasoning as above, the cumulative probability is:
\begin{equation}
    \tilde{h}_n^{(k)}(0) = \frac{1}{d_{n,n}}\sum_{i=1}^k \log \frac{d_{n-1,n-1}}{i} = \frac{1}{d_{n,n}}(k \log d_{n-1, n-1} - \log k!) \approx r\frac{d_{n-1, n-1}}{d_{n,n}} (1 - \log r).
    \label{eq:general-haar-heralding}
\end{equation}
For the fraction of all possible patterns $r = k/d_{n-1,n-1}$ ``large enough", the equation can be re-written via Stirling's approximation as on the right hand side. For $r\rightarrow 1$ this reproduces $\tilde{h}_n^{(k)} \approx \frac{d_{n-1, n-1}}{d_{n,n}}$ as expected from Eq.~\eqref{eq:haar_herald} ($r=1$ is equivalent to post-selecting on all possible patterns).
We show how the analytic function compares to data for $n=7$ in Fig.~\ref{fig:haar_analytic}.

This analysis shows that distillation is not a unique phenomena, in fact, it is typically possible for an arbitrary unitary, so long as there are a subset of patterns that have relatively high constructive interference. However, distillation from random unitaries is generally much more expensive (compared to say the Fourier protocol) as shown in Fig.~\ref{fig:haar_high_weight} and the subsequent analysis above; the heralding rates are lower and the error is reduced by a lesser amount. The nature of the constructive interference in the Fourier (and Hadamard) unitaries is what allows for an efficient protocol, where the weight concentrates on a relatively small subset of all possible patterns, as per the zero-transmission laws, ultimately as a result of symmetry in the matrices (Sect.~\ref{sec:symmetry general}).

\begin{figure}
    \centering
    \includegraphics[width=\columnwidth]{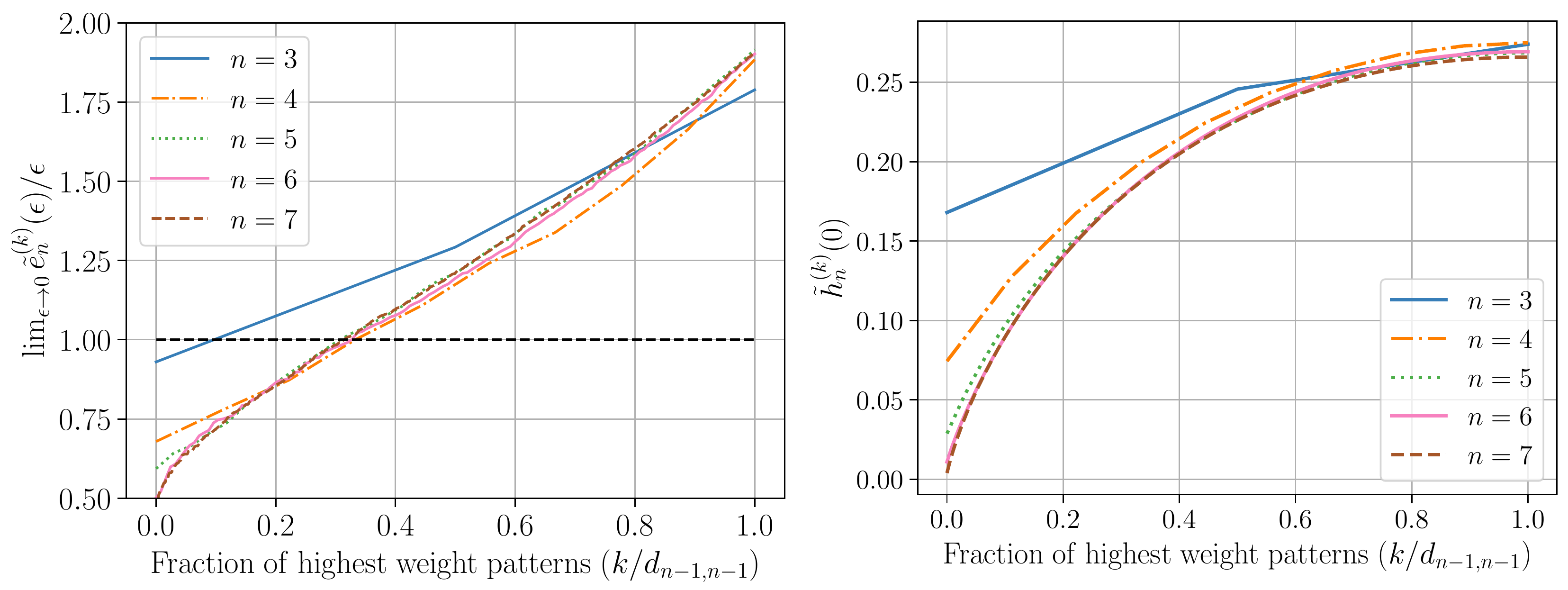}
    \caption{Distillation using Haar random unitaries, where post-selection is on a fraction of the highest weight (probability) output patterns (i.e., if one enumerated all states, those with the most constructive interference would be positioned towards the left on the x-axis). To show the data for different $n$ on a consistent axis, $x=0$ corresponds to including the single highest weight pattern, and $x=1$ to including all patterns. Each curve is an average of at least 100 Haar random unitaries. Left: Error reduction relative to the initial error $\epsilon$, at `small' $\epsilon$. We see for $n=4,5,6,7$, there is a distillation effect when the top $\lesssim 30\%$ of the highest weight patterns are used. As expected from Eq.~\eqref{eq:haar_fidel}, this tends to $\tilde{e}_n(\epsilon)/\epsilon \approx 2$ when all patterns are used. Right: Heralding rate at 0 error for the protocol that selects a certain percentage of the highest weight patterns. As expected from Eq.~\eqref{eq:haar_herald} it tends to a little over $1/4$, when using all possible patterns.}
    \label{fig:haar_high_weight}
\end{figure}

\begin{figure}
    \centering
    \includegraphics[width=0.75\columnwidth]{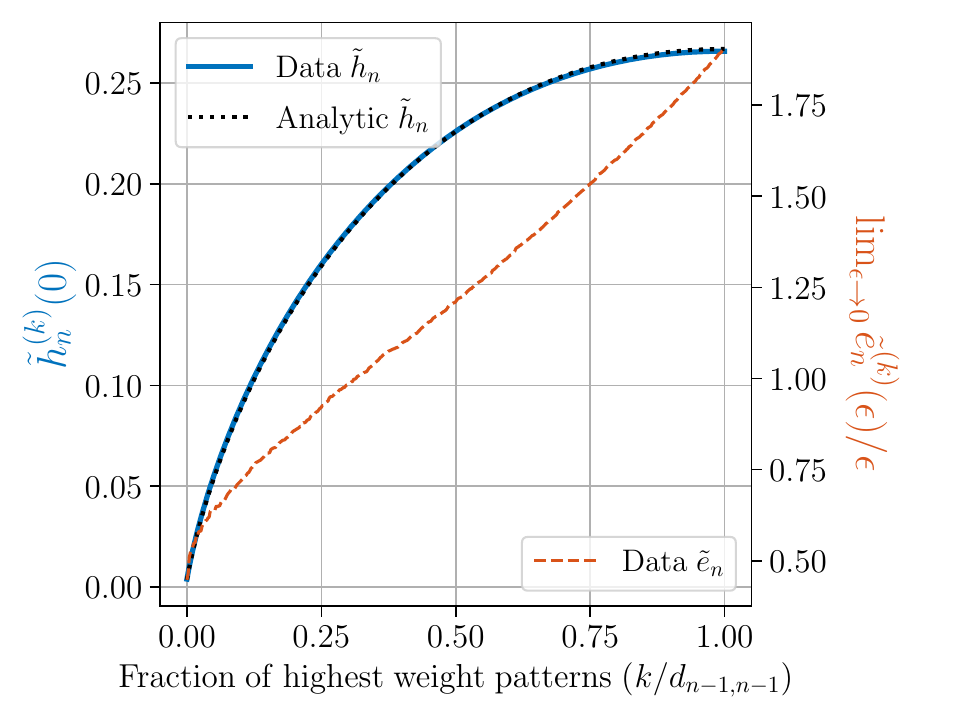}
    \caption{The same data as in Fig.~\ref{fig:haar_high_weight} for the $n=7$ case, with the analytically computed heralding rate, Eq.~\eqref{eq:general-haar-heralding} (dotted line). The dashed line (red) pertains to the right y-axis.}
    \label{fig:haar_analytic}
\end{figure}

\section{Simulations}\label{sec:simulation appendix}

Here we discuss how we performed our numerics for the data in the main text and appendix (e.g., Sect.~\ref{sec:numerics appendix}). 
For simplicity, we restrict our attention to the OBB and SBB error models. 
First we set up some basics. We call $d_{n,m}={n+m-1\choose n}$ the dimension of the Hilbert space of $n$ photons in $m$ modes.
To compute an evolved Fock state under linear optics, in dimension $d_{n,m}$, requires time scaling like $O(nd_{n,m})$ and space (i.e., memory) like $O(d_{n,m})$ \cite{marshall_simulation_2023}.
For an $n$ photon distillation protocol, the full state space is of dimension $d_{n,n}$; however, as discussed below, we only need to consider outcomes with 0 and 1 photon in mode 0, and hence the relevant dimension is $d_{n, n-1}+d_{n-1, n-1}$ (for $n=16$, this is around 220 million).
As the simulations are carried out by adding one photon at a time, as soon as a configuration has more than 2 photons in mode 0, we can discard it, thus constraining the dynamics in the relevant subspace.
We will now describe how we simulate distinguishability errors in this work.
Though not necessary, some additional details on this topic can be found in in Ref.~\cite{saied2024advancing}.

Naively to compute the distillation properties (i.e. $h_n(\epsilon), e_n(\epsilon)$) requires one to evaluate all $2^n$ initial states with $n-k$ `ideal' photons and $k$ `error' photons. 
(See the discussions of Appendix~\ref{sec:calculations}.) 
For example, for $n=3$ the relevant input states can be represented (1, 1, 1), (1, 1, $1'$), (1, $1'$, 1), ($1'$, 1, 1), (1, $1'$, $1'$), ($1'$, 1, $1'$), ($1'$, $1'$, 1), ($1'$,$1'$,$1'$), where $1$ indicates an `ideal' photon, and $1'$ an error. 
Since the errors do not interfere with the ideal photons, we can split the simulation into different subspaces, and combine results in post-processing.

In the SBB case each initial state splits into two independent simulations, one of $n-k$ identical photons, and another with $k$ identical photons, in the complementary positions (by symmetry of the states, this reduces the number of unique initial states to $2^{n-1}$). In each case we only care about getting 0 or 1 photon out in mode 0, since anything else would be heralded as a non-ideal pattern.
In particular, this is equivalent to projection of the full evolved state onto the subspace with 0 or 1 photons in mode 0. 
The two evolved (unnormalized) `states', $|\psi_{n-k}\rangle, |\phi_k\rangle$, can be combined to compute the relevant quantities.
Let's write each as a sum over configurations, where we separate out those that have 0 or 1 photon in mode 0:
\begin{equation}
\begin{split}
    & |\psi_{n-k}\rangle =  \sum_{r^{(0)}} b_r^{(0)} |r^{(0)}\rangle +  \sum_{r^{(1)}} b^{(1)}_r |r^{(1)}\rangle \\
    & |\phi_{k}\rangle =  \sum_{s^{(0)}} c_s^{(0)} |s^{(0)}\rangle +  \sum_{s^{(1)}} c^{(1)}_s |s^{(1)}\rangle, 
\end{split}
\end{equation}
where $|r^{(0)}\rangle$ ($|r^{(1)}\rangle$) is an $n-k$ photon configuration with 0 (1) photons in mode 0. Likewise for $|s^{(0)}\rangle, \ket{s^{(1)}}$, but for $k$ photons. 

As these two states do not interfere, we can classically combine their outcomes.
Note that only pairs of the form $(r^{(0)}, s^{(1)})$ and $(r^{(1)}, s^{(0)})$ can combine to give a heralded pattern with 1 photon out. In particular, we can compute two relevant probabilities
\begin{equation}
    P_{01}:=\sum_{r^{(0)}, s^{(1)}: r^{(0)} + s^{(1)} \in \mathcal{I}} |b_r^{(0)}|^2 | c_s^{(1)}|^2,\;\;P_{10} = \sum_{r^{(1)}, s^{(0)}: r^{(1)} + s^{(0)} \in \mathcal{I}} |b_r^{(1)}|^2 |c_s^{(0)}|^2,
    \label{eq:simprobs0}
\end{equation}
where $\mathcal{I}$ denotes the set of ideal heralding patterns, and $r+s$ is notation to add the length $n$ tuples element-wise (since PNRD does not distinguish between ideal and error photons).
The first term corresponds to the probability that an error photon exits mode 0, and the second term an ideal photon. Therefore, the heralding probability given this particular initial state is $P_{01}+P_{10}$.

The numerical time cost to compute these terms by checking membership of $\mathcal{I}$ for all pairs is respectively (order of)
\begin{equation}
    d_{n-k, n-1} d_{k-1, n-1},\,\, d_{n-k-1, n-1} d_{k, n-1},
\end{equation}
which in fact dominates the total simulation cost, i.e., this is generally more expensive than computing the evolved states themselves $|\psi_{n-k}\rangle, |\phi_k\rangle$ (discussed above).
Note, for the largest size we considered, $n=16$, we have $|\mathcal{I}|\approx 2^{22}$; however, membership can be checked in time $O(1)$ by hashing.
For $n=16$, computing these probabilities in the worst case requires around $2^{36}$ numerical evaluations, which is achievable on a single processor given a few hours.

If we further denote $P_{01}^{x_k}, P_{10}^{x_k}$ to be the two above probabilities given the initial state $x_k$ of $k$ errors (with ${n \choose k}$ total states for each $k=0, \dots, n$), we can compute the relevant rates as (see Eq.~\eqref{eq:error explicit}):
\begin{equation}
\begin{split}
    & h_n(\epsilon) = \sum_{k=0}^n \epsilon^k (1-\epsilon)^{n-k} \sum_{x_k} (P_{01}^{x_k} + P_{10}^{x_k}) \\
    & \bar{e}_n(\epsilon) = \sum_{k=0}^n \epsilon^k (1-\epsilon)^{n-k} \sum_{x_k} P_{01}^{x_k}.
    \label{eq:sim_probs}
\end{split}
\end{equation}

In the OBB case, since the $k$ errors are all independent, the statistics are classical. In particular, the probability to observe an output `error' pattern $s=(s_0, \dots, s_{n-1})$ from an initial state with $k=\sum_i s_i$ independent errors is given by the multinomial expression:
\begin{equation}
    p_{s} = \frac{k!}{n^k} \prod_{i}\frac{1}{s_i!}.
    \label{eq:multinomial}
\end{equation}
Since the final calculation only requires probabilities as per Eq.~\eqref{eq:simprobs0} (i.e., as opposed to the quantum amplitudes), we do not per se require any simulation for the $k$ error state; given the evolved $n-k$ identical photon state $|\psi_{n-k}\rangle$ as above, along with Eq.~\eqref{eq:multinomial}, we can compute Eqs.~\eqref{eq:simprobs0}, \eqref{eq:sim_probs}.

Thus far we have considered a brute force approach, where we evaluate all $2^n$ error configurations. However as remarked several times in this work, the systems of interest for distillation have a great deal of symmetry, which we can use to reduce the total number of unique error configurations.
Consider the Hadamard case first. Recall the discussion App.~\ref{sec:hadamard symm}, where it is noted the $2\times 2$ Hadamard matrix $H=F_2$ has the symmetry $HX=ZH$ where $X,Z$ are Pauli matrices (i.e., swapping the columns is equivalent to multiplication of row 1 by -1). 
This symmetry extends amongst the tensor factors of $H^{\otimes r}$ in the obvious way. Let's write these symmetries as before: $UP=DU$ where $P$ is a permutation matrix (swapping columns), and $D$ a diagonal matrix (of phases).
Notice that this implies initial states that are equivalent up to permutation, i.e., $|y_k\rangle = P|x_k\rangle$, will have the same output amplitudes, up to phase:
$\langle s|U|y_k\rangle = \langle s |UP|x_k\rangle =\langle s |DU|x_k\rangle = e^{i\phi} \langle s|U|x_k\rangle$. Since eventually we only care about the probabilities (Eq.~\eqref{eq:simprobs0}), the phase is unimportant for computing heralding rates for distillation.

We can find another kind of symmetry relevant for distillation; if there exists a column swap that is equivalent to swapping rows (apart from the 0th), the permuted input state will have the same output amplitudes on configurations with modes permuted. 
In other words: we look for cases in which swapping certain input modes is equivalent to swapping certain corresponding output modes. If we denote the column and row swap by matrices $C, R$ respectively, so that $UC=RU$, we note that for any pattern $s$ with $\bra{s}U\ket{\bar{1}}\neq 0$, 
its permuted version $C|s\rangle$ 
leads to the same amplitude: 
$0 \neq \langle \bar{1}|U|s\rangle = \langle \bar{1}|R U|s \rangle = \langle \bar{1}|UC|s\rangle$.
If such a transformation exists, this implies the relevant statistics for an input configuration $x_k$ will be the same as the version permuted by $C$. Intuitively this is saying that the amplitudes don't care precisely which mode a photon ends up in; they only care if the full configuration is an ideal pattern. Note that our definition of ideal pattern $s\in\mc{I}$ \emph{does} care that exactly $1$ photon is in the $0$th mode; therefore, to ensure this process maps ideal patterns to other ideal patterns, we must consider only $C$ and $R$ fixing the $0$th mode. 

We provide the symmetries for $H_4=H\otimes H$ as an example (we ignore normalization for convenience). Recall:
\begin{equation}
    H_4 = \left(
    \begin{array}{cccc}
        1 &  1 & 1 &  1\\
        1 &  -1 & 1 & -1 \\
        1 &  1 & -1 &  -1\\
        1 &  -1 & -1 & 1 
    \end{array}
    \right).
\end{equation}
There are three nontrivial symmetries of the usual $UP = DU$ form, arising from applying $HX=ZH$ on either or both tensor factors (i.e., $X\otimes I, I\otimes X, X\otimes X$ with $Z\otimes I, I\otimes Z, Z\otimes Z$). However there are three further symmetries of the form $RU = UC$ that can be seen, involving swapping columns (1,2), (1,3), and (2,3); these are respectively equivalent to row swaps (1,2), (2,3), (1,3). 
For $H_{16} = H_4\otimes H_4$, this results in only 32 unique error configurations (compared to $2^{16}$ naively).

In the Fourier case, $F_n$, as described in App.~\ref{sec:fourier symm} there is a cyclic symmetry of the form $F_nP=DF_n$ where $D$ is a diagonal matrix (of phases) and $P$ a cyclic permutation. This implies any initial states that are cyclically equivalent will have the same statistics as discussed above.
This reduces the number of unique states by around a factor of $n$.
We can find another symmetry to further reduce the terms, however, swapping both columns and rows similarly to the Hadamard case discussed above. Recall that the matrix elements of $F_n$ (again ignoring normalization) are of the form $\omega_n^{ij}$. 
For $r$ relatively prime to $n$, a unique multiplicative inverse $r^{-1}$ exists (modulo $n$). Therefore if we map $i\mapsto r^{-1} i, j\mapsto r j$, the matrix elements are unchanged: $ij = r^{-1} i r j \mod n$. This corresponds to sending column $j$ to column $r j$ (permuting modes of the input state), and then sending the output row  $i$ to $r^{-1} i$. (Note that we index the modes modulo $n$, so the $0$th row and column are untouched.) 
For $n=16$, these symmetries reduce the total number of unique errors from $2^{16}$ to 693. The sequence of integers counting the number of unique errors for different $F_n$ can be found in the OEIS, sequence number A002729 \cite{oeis2}.

\section{Numerics}\label{sec:numerics appendix}

Here we list the functions for the heralding rate $h_n(\epsilon)$ and  $\ole_n(\epsilon):=h_n(\epsilon)e_n(\epsilon)$ (see Eq.~\eqref{eq:error explicit} for reference), where the coefficients are given up to 6 decimal places (for some of the smaller protocols we write it as a fraction). The superscript denotes Hadamard ($H$) or Fourier ($F$) protocol, along with the error model (SBB or OBB). We go up to protocols of size $n=16$.

\let\clearpage\relax
\begin{equation*}
\begin{split}
  & h_{3}^{(F, OBB)}(\epsilon) = \frac{1}{3}(1-\epsilon)^{3} + \frac{1}{9}{3\choose 1}\epsilon^{1}(1-\epsilon)^{2} + \frac{2}{9}{3\choose 2}\epsilon^{2}(1-\epsilon)^{1} + \frac{2}{9}\epsilon^{3} \\ 
  & \bar{e}_{3}^{(F, OBB)}(\epsilon) = \frac{1}{27}{3\choose 1}\epsilon^{1}(1-\epsilon)^{2} + \frac{4}{27}{3\choose 2}\epsilon^{2}(1-\epsilon)^{1} + \frac{2}{9}\epsilon^{3} 
\end{split}
\end{equation*}

\begin{equation*}
\begin{split}
  & h_{4}^{(F, OBB)}(\epsilon) = \frac{1}{4}(1-\epsilon)^{4} + \frac{1}{16}{4\choose 1}\epsilon^{1}(1-\epsilon)^{3} + \frac{1}{12}{4\choose 2}\epsilon^{2}(1-\epsilon)^{2} + \frac{3}{32}{4\choose 3}\epsilon^{3}(1-\epsilon)^{1} + \frac{3}{32}\epsilon^{4} \\ 
  & \bar{e}_{4}^{(F, OBB)}(\epsilon) = \frac{1}{64}{4\choose 1}\epsilon^{1}(1-\epsilon)^{3} + \frac{5}{96}{4\choose 2}\epsilon^{2}(1-\epsilon)^{2} + \frac{9}{128}{4\choose 3}\epsilon^{3}(1-\epsilon)^{1} + \frac{3}{32}\epsilon^{4} 
\end{split}
\end{equation*}

\begin{equation*}
\begin{split}
  & h_{5}^{(F, OBB)}(\epsilon) = 0.264000(1-\epsilon)^{5} + 0.052800{5\choose 1}\epsilon^{1}(1-\epsilon)^{4} + 0.073600{5\choose 2}\epsilon^{2}(1-\epsilon)^{3}\\ & + 0.076800{5\choose 3}\epsilon^{3}(1-\epsilon)^{2} + 0.083200{5\choose 4}\epsilon^{4}(1-\epsilon)^{1} + 0.083200\epsilon^{5} \\ 
  & \bar{e}_{5}^{(F, OBB)}(\epsilon) = 0.010560{5\choose 1}\epsilon^{1}(1-\epsilon)^{4} + 0.039680{5\choose 2}\epsilon^{2}(1-\epsilon)^{3} + 0.051840{5\choose 3}\epsilon^{3}(1-\epsilon)^{2}\\ & + 0.066560{5\choose 4}\epsilon^{4}(1-\epsilon)^{1} + 0.083200\epsilon^{5} 
\end{split}
\end{equation*}

\begin{equation*}
\begin{split}
  & h_{6}^{(F, OBB)}(\epsilon) = 0.259259(1-\epsilon)^{6} + 0.043210{6\choose 1}\epsilon^{1}(1-\epsilon)^{5} + 0.034156{6\choose 2}\epsilon^{2}(1-\epsilon)^{4}\\ & + 0.032562{6\choose 3}\epsilon^{3}(1-\epsilon)^{3} + 0.029630{6\choose 4}\epsilon^{4}(1-\epsilon)^{2} + 0.028292{6\choose 5}\epsilon^{5}(1-\epsilon)^{1}\\ & + 0.028292\epsilon^{6} \\ 
  & \bar{e}_{6}^{(F, OBB)}(\epsilon) = 0.007202{6\choose 1}\epsilon^{1}(1-\epsilon)^{5} + 0.016872{6\choose 2}\epsilon^{2}(1-\epsilon)^{4} + 0.021451{6\choose 3}\epsilon^{3}(1-\epsilon)^{3}\\ & + 0.022085{6\choose 4}\epsilon^{4}(1-\epsilon)^{2} + 0.023577{6\choose 5}\epsilon^{5}(1-\epsilon)^{1} + 0.028292\epsilon^{6} 
\end{split}
\end{equation*}

\begin{equation*}
\begin{split}
  & h_{7}^{(F, OBB)}(\epsilon) = 0.258523(1-\epsilon)^{7} + 0.036932{7\choose 1}\epsilon^{1}(1-\epsilon)^{6} + 0.047072{7\choose 2}\epsilon^{2}(1-\epsilon)^{5}\\ & + 0.050204{7\choose 3}\epsilon^{3}(1-\epsilon)^{4} + 0.052988{7\choose 4}\epsilon^{4}(1-\epsilon)^{3} + 0.055079{7\choose 5}\epsilon^{5}(1-\epsilon)^{2}\\ & + 0.056660{7\choose 6}\epsilon^{6}(1-\epsilon)^{1} + 0.056660\epsilon^{7} \\ 
  & \bar{e}_{7}^{(F, OBB)}(\epsilon) = 0.005276{7\choose 1}\epsilon^{1}(1-\epsilon)^{6} + 0.019997{7\choose 2}\epsilon^{2}(1-\epsilon)^{5} + 0.027561{7\choose 3}\epsilon^{3}(1-\epsilon)^{4}\\ & + 0.034779{7\choose 4}\epsilon^{4}(1-\epsilon)^{3} + 0.041589{7\choose 5}\epsilon^{5}(1-\epsilon)^{2} + 0.048566{7\choose 6}\epsilon^{6}(1-\epsilon)^{1}\\ & + 0.056660\epsilon^{7} 
\end{split}
\end{equation*}

\begin{equation*}
\begin{split}
  & h_{8}^{(F, OBB)}(\epsilon) = 0.257446(1-\epsilon)^{8} + 0.032181{8\choose 1}\epsilon^{1}(1-\epsilon)^{7} + 0.039616{8\choose 2}\epsilon^{2}(1-\epsilon)^{6}\\ & + 0.042520{8\choose 3}\epsilon^{3}(1-\epsilon)^{5} + 0.044800{8\choose 4}\epsilon^{4}(1-\epsilon)^{4} + 0.046653{8\choose 5}\epsilon^{5}(1-\epsilon)^{3}\\ & + 0.048085{8\choose 6}\epsilon^{6}(1-\epsilon)^{2} + 0.049087{8\choose 7}\epsilon^{7}(1-\epsilon)^{1} + 0.049087\epsilon^{8} \\ 
  & \bar{e}_{8}^{(F, OBB)}(\epsilon) = 0.004023{8\choose 1}\epsilon^{1}(1-\epsilon)^{7} + 0.014948{8\choose 2}\epsilon^{2}(1-\epsilon)^{6} + 0.021359{8\choose 3}\epsilon^{3}(1-\epsilon)^{5}\\ & + 0.027055{8\choose 4}\epsilon^{4}(1-\epsilon)^{4} + 0.032378{8\choose 5}\epsilon^{5}(1-\epsilon)^{3} + 0.037566{8\choose 6}\epsilon^{6}(1-\epsilon)^{2}\\ & + 0.042951{8\choose 7}\epsilon^{7}(1-\epsilon)^{1} + 0.049087\epsilon^{8} 
\end{split}
\end{equation*}

\begin{equation*}
\begin{split}
  & h_{9}^{(F, OBB)}(\epsilon) = 0.256678(1-\epsilon)^{9} + 0.028520{9\choose 1}\epsilon^{1}(1-\epsilon)^{8} + 0.034616{9\choose 2}\epsilon^{2}(1-\epsilon)^{7}\\ & + 0.036834{9\choose 3}\epsilon^{3}(1-\epsilon)^{6} + 0.038728{9\choose 4}\epsilon^{4}(1-\epsilon)^{5} + 0.040313{9\choose 5}\epsilon^{5}(1-\epsilon)^{4}\\ & + 0.041613{9\choose 6}\epsilon^{6}(1-\epsilon)^{3} + 0.042628{9\choose 7}\epsilon^{7}(1-\epsilon)^{2} + 0.043305{9\choose 8}\epsilon^{8}(1-\epsilon)^{1}\\ & + 0.043305\epsilon^{9} \\ 
  & \bar{e}_{9}^{(F, OBB)}(\epsilon) = 0.003169{9\choose 1}\epsilon^{1}(1-\epsilon)^{8} + 0.012104{9\choose 2}\epsilon^{2}(1-\epsilon)^{7} + 0.017083{9\choose 3}\epsilon^{3}(1-\epsilon)^{6}\\ & + 0.021689{9\choose 4}\epsilon^{4}(1-\epsilon)^{5} + 0.025991{9\choose 5}\epsilon^{5}(1-\epsilon)^{4} + 0.030110{9\choose 6}\epsilon^{6}(1-\epsilon)^{3}\\ & + 0.034208{9\choose 7}\epsilon^{7}(1-\epsilon)^{2} + 0.038493{9\choose 8}\epsilon^{8}(1-\epsilon)^{1} + 0.043305\epsilon^{9} 
\end{split}
\end{equation*}

\begin{equation*}
\begin{split}
  & h_{10}^{(F, OBB)}(\epsilon) = 0.256038(1-\epsilon)^{10} + 0.025604{10\choose 1}\epsilon^{1}(1-\epsilon)^{9} + 0.025161{10\choose 2}\epsilon^{2}(1-\epsilon)^{8}\\ & + 0.026194{10\choose 3}\epsilon^{3}(1-\epsilon)^{7} + 0.027467{10\choose 4}\epsilon^{4}(1-\epsilon)^{6} + 0.028630{10\choose 5}\epsilon^{5}(1-\epsilon)^{5}\\ & + 0.029633{10\choose 6}\epsilon^{6}(1-\epsilon)^{4} + 0.030501{10\choose 7}\epsilon^{7}(1-\epsilon)^{3} + 0.031206{10\choose 8}\epsilon^{8}(1-\epsilon)^{2}\\ & + 0.031657{10\choose 9}\epsilon^{9}(1-\epsilon)^{1} + 0.031657\epsilon^{10} \\ 
  & \bar{e}_{10}^{(F, OBB)}(\epsilon) = 0.002560{10\choose 1}\epsilon^{1}(1-\epsilon)^{9} + 0.007890{10\choose 2}\epsilon^{2}(1-\epsilon)^{8} + 0.011238{10\choose 3}\epsilon^{3}(1-\epsilon)^{7}\\ & + 0.014343{10\choose 4}\epsilon^{4}(1-\epsilon)^{6} + 0.017258{10\choose 5}\epsilon^{5}(1-\epsilon)^{5} + 0.020034{10\choose 6}\epsilon^{6}(1-\epsilon)^{4}\\ & + 0.022780{10\choose 7}\epsilon^{7}(1-\epsilon)^{3} + 0.025578{10\choose 8}\epsilon^{8}(1-\epsilon)^{2} + 0.028491{10\choose 9}\epsilon^{9}(1-\epsilon)^{1}\\ & + 0.031657\epsilon^{10} 
\end{split}
\end{equation*}

\begin{equation*}
\begin{split}
  & h_{11}^{(F, OBB)}(\epsilon) = 0.255512(1-\epsilon)^{11} + 0.023228{11\choose 1}\epsilon^{1}(1-\epsilon)^{10} + 0.027306{11\choose 2}\epsilon^{2}(1-\epsilon)^{9}\\ & + 0.028891{11\choose 3}\epsilon^{3}(1-\epsilon)^{8} + 0.030293{11\choose 4}\epsilon^{4}(1-\epsilon)^{7} + 0.031500{11\choose 5}\epsilon^{5}(1-\epsilon)^{6}\\ & + 0.032535{11\choose 6}\epsilon^{6}(1-\epsilon)^{5} + 0.033413{11\choose 7}\epsilon^{7}(1-\epsilon)^{4} + 0.034138{11\choose 8}\epsilon^{8}(1-\epsilon)^{3}\\ & + 0.034699{11\choose 9}\epsilon^{9}(1-\epsilon)^{2} + 0.035049{11\choose 10}\epsilon^{10}(1-\epsilon)^{1} + 0.035049\epsilon^{11} \\ 
  & \bar{e}_{11}^{(F, OBB)}(\epsilon) = 0.002112{11\choose 1}\epsilon^{1}(1-\epsilon)^{10} + 0.008118{11\choose 2}\epsilon^{2}(1-\epsilon)^{9} + 0.011588{11\choose 3}\epsilon^{3}(1-\epsilon)^{8}\\ & + 0.014835{11\choose 4}\epsilon^{4}(1-\epsilon)^{7} + 0.017861{11\choose 5}\epsilon^{5}(1-\epsilon)^{6} + 0.020725{11\choose 6}\epsilon^{6}(1-\epsilon)^{5}\\ & + 0.023484{11\choose 7}\epsilon^{7}(1-\epsilon)^{4} + 0.026204{11\choose 8}\epsilon^{8}(1-\epsilon)^{3} + 0.028964{11\choose 9}\epsilon^{9}(1-\epsilon)^{2}\\ & + 0.031863{11\choose 10}\epsilon^{10}(1-\epsilon)^{1} + 0.035049\epsilon^{11} 
\end{split}
\end{equation*}

\begin{equation*}
\begin{split}
  & h_{12}^{(F, OBB)}(\epsilon) = 0.255068(1-\epsilon)^{12} + 0.021256{12\choose 1}\epsilon^{1}(1-\epsilon)^{11} + 0.016566{12\choose 2}\epsilon^{2}(1-\epsilon)^{10}\\ & + 0.015836{12\choose 3}\epsilon^{3}(1-\epsilon)^{9} + 0.015671{12\choose 4}\epsilon^{4}(1-\epsilon)^{8} + 0.015559{12\choose 5}\epsilon^{5}(1-\epsilon)^{7}\\ & + 0.015427{12\choose 6}\epsilon^{6}(1-\epsilon)^{6} + 0.015266{12\choose 7}\epsilon^{7}(1-\epsilon)^{5} + 0.015097{12\choose 8}\epsilon^{8}(1-\epsilon)^{4}\\ & + 0.014939{12\choose 9}\epsilon^{9}(1-\epsilon)^{3} + 0.014821{12\choose 10}\epsilon^{10}(1-\epsilon)^{2} + 0.014762{12\choose 11}\epsilon^{11}(1-\epsilon)^{1}\\ & + 0.014762\epsilon^{12} \\ 
  & \bar{e}_{12}^{(F, OBB)}(\epsilon) = 0.001771{12\choose 1}\epsilon^{1}(1-\epsilon)^{11} + 0.004568{12\choose 2}\epsilon^{2}(1-\epsilon)^{10} + 0.006112{12\choose 3}\epsilon^{3}(1-\epsilon)^{9}\\ & + 0.007409{12\choose 4}\epsilon^{4}(1-\epsilon)^{8} + 0.008507{12\choose 5}\epsilon^{5}(1-\epsilon)^{7} + 0.009451{12\choose 6}\epsilon^{6}(1-\epsilon)^{6}\\ & + 0.010279{12\choose 7}\epsilon^{7}(1-\epsilon)^{5} + 0.011040{12\choose 8}\epsilon^{8}(1-\epsilon)^{4} + 0.011785{12\choose 9}\epsilon^{9}(1-\epsilon)^{3}\\ & + 0.012584{12\choose 10}\epsilon^{10}(1-\epsilon)^{2} + 0.013532{12\choose 11}\epsilon^{11}(1-\epsilon)^{1} + 0.014762\epsilon^{12} 
\end{split}
\end{equation*}

\begin{equation*}
\begin{split}
  & h_{13}^{(F, OBB)}(\epsilon) = 0.254691(1-\epsilon)^{13} + 0.019592{13\choose 1}\epsilon^{1}(1-\epsilon)^{12} + 0.022512{13\choose 2}\epsilon^{2}(1-\epsilon)^{11}\\ & + 0.023702{13\choose 3}\epsilon^{3}(1-\epsilon)^{10} + 0.024770{13\choose 4}\epsilon^{4}(1-\epsilon)^{9} + 0.025713{13\choose 5}\epsilon^{5}(1-\epsilon)^{8}\\ & + 0.026544{13\choose 6}\epsilon^{6}(1-\epsilon)^{7} + 0.027272{13\choose 7}\epsilon^{7}(1-\epsilon)^{6} + 0.027904{13\choose 8}\epsilon^{8}(1-\epsilon)^{5}\\ & + 0.028446{13\choose 9}\epsilon^{9}(1-\epsilon)^{4} + 0.028893{13\choose 10}\epsilon^{10}(1-\epsilon)^{3} + 0.029234{13\choose 11}\epsilon^{11}(1-\epsilon)^{2}\\ & + 0.029438{13\choose 12}\epsilon^{12}(1-\epsilon)^{1} + 0.029438\epsilon^{13} \\ 
  & \bar{e}_{13}^{(F, OBB)}(\epsilon) = 0.001507{13\choose 1}\epsilon^{1}(1-\epsilon)^{12} + 0.005824{13\choose 2}\epsilon^{2}(1-\epsilon)^{11} + 0.008382{13\choose 3}\epsilon^{3}(1-\epsilon)^{10}\\ & + 0.010789{13\choose 4}\epsilon^{4}(1-\epsilon)^{9} + 0.013046{13\choose 5}\epsilon^{5}(1-\epsilon)^{8} + 0.015182{13\choose 6}\epsilon^{6}(1-\epsilon)^{7}\\ & + 0.017223{13\choose 7}\epsilon^{7}(1-\epsilon)^{6} + 0.019200{13\choose 8}\epsilon^{8}(1-\epsilon)^{5} + 0.021144{13\choose 9}\epsilon^{9}(1-\epsilon)^{4}\\ & + 0.023090{13\choose 10}\epsilon^{10}(1-\epsilon)^{3} + 0.025082{13\choose 11}\epsilon^{11}(1-\epsilon)^{2} + 0.027174{13\choose 12}\epsilon^{12}(1-\epsilon)^{1}\\ & + 0.029438\epsilon^{13} 
\end{split}
\end{equation*}

\begin{equation*}
\begin{split}
  & h_{14}^{(F, OBB)}(\epsilon) = 0.254365(1-\epsilon)^{14} + 0.018169{14\choose 1}\epsilon^{1}(1-\epsilon)^{13} + 0.019987{14\choose 2}\epsilon^{2}(1-\epsilon)^{12}\\ & + 0.020964{14\choose 3}\epsilon^{3}(1-\epsilon)^{11} + 0.021899{14\choose 4}\epsilon^{4}(1-\epsilon)^{10} + 0.022751{14\choose 5}\epsilon^{5}(1-\epsilon)^{9}\\ & + 0.023515{14\choose 6}\epsilon^{6}(1-\epsilon)^{8} + 0.024196{14\choose 7}\epsilon^{7}(1-\epsilon)^{7} + 0.024798{14\choose 8}\epsilon^{8}(1-\epsilon)^{6}\\ & + 0.025322{14\choose 9}\epsilon^{9}(1-\epsilon)^{5} + 0.025770{14\choose 10}\epsilon^{10}(1-\epsilon)^{4} + 0.026137{14\choose 11}\epsilon^{11}(1-\epsilon)^{3}\\ & + 0.026414{14\choose 12}\epsilon^{12}(1-\epsilon)^{2} + 0.026576{14\choose 13}\epsilon^{13}(1-\epsilon)^{1} + 0.026576\epsilon^{14} \\ 
  & \bar{e}_{14}^{(F, OBB)}(\epsilon) = 0.001298{14\choose 1}\epsilon^{1}(1-\epsilon)^{13} + 0.004810{14\choose 2}\epsilon^{2}(1-\epsilon)^{12} + 0.006986{14\choose 3}\epsilon^{3}(1-\epsilon)^{11}\\ & + 0.009027{14\choose 4}\epsilon^{4}(1-\epsilon)^{10} + 0.010959{14\choose 5}\epsilon^{5}(1-\epsilon)^{9} + 0.012798{14\choose 6}\epsilon^{6}(1-\epsilon)^{8}\\ & + 0.014559{14\choose 7}\epsilon^{7}(1-\epsilon)^{7} + 0.016260{14\choose 8}\epsilon^{8}(1-\epsilon)^{6} + 0.017923{14\choose 9}\epsilon^{9}(1-\epsilon)^{5}\\ & + 0.019567{14\choose 10}\epsilon^{10}(1-\epsilon)^{4} + 0.021218{14\choose 11}\epsilon^{11}(1-\epsilon)^{3} + 0.022909{14\choose 12}\epsilon^{12}(1-\epsilon)^{2}\\ & + 0.024678{14\choose 13}\epsilon^{13}(1-\epsilon)^{1} + 0.026576\epsilon^{14} 
\end{split}
\end{equation*}

\begin{equation*}
\begin{split}
  & h_{15}^{(F, OBB)}(\epsilon) = 0.254081(1-\epsilon)^{15} + 0.016939{15\choose 1}\epsilon^{1}(1-\epsilon)^{14} + 0.017300{15\choose 2}\epsilon^{2}(1-\epsilon)^{13}\\ & + 0.017856{15\choose 3}\epsilon^{3}(1-\epsilon)^{12} + 0.018498{15\choose 4}\epsilon^{4}(1-\epsilon)^{11} + 0.019093{15\choose 5}\epsilon^{5}(1-\epsilon)^{10}\\ & + 0.019624{15\choose 6}\epsilon^{6}(1-\epsilon)^{9} + 0.020091{15\choose 7}\epsilon^{7}(1-\epsilon)^{8} + 0.020500{15\choose 8}\epsilon^{8}(1-\epsilon)^{7}\\ & + 0.020856{15\choose 9}\epsilon^{9}(1-\epsilon)^{6} + 0.021164{15\choose 10}\epsilon^{10}(1-\epsilon)^{5} + 0.021426{15\choose 11}\epsilon^{11}(1-\epsilon)^{4}\\ & + 0.021643{15\choose 12}\epsilon^{12}(1-\epsilon)^{3} + 0.021808{15\choose 13}\epsilon^{13}(1-\epsilon)^{2} + 0.021907{15\choose 14}\epsilon^{14}(1-\epsilon)^{1}\\ & + 0.021907\epsilon^{15} \\ 
  & \bar{e}_{15}^{(F, OBB)}(\epsilon) = 0.001129{15\choose 1}\epsilon^{1}(1-\epsilon)^{14} + 0.003946{15\choose 2}\epsilon^{2}(1-\epsilon)^{13} + 0.005654{15\choose 3}\epsilon^{3}(1-\epsilon)^{12}\\ & + 0.007271{15\choose 4}\epsilon^{4}(1-\epsilon)^{11} + 0.008791{15\choose 5}\epsilon^{5}(1-\epsilon)^{10} + 0.010224{15\choose 6}\epsilon^{6}(1-\epsilon)^{9}\\ & + 0.011582{15\choose 7}\epsilon^{7}(1-\epsilon)^{8} + 0.012880{15\choose 8}\epsilon^{8}(1-\epsilon)^{7} + 0.014133{15\choose 9}\epsilon^{9}(1-\epsilon)^{6}\\ & + 0.015361{15\choose 10}\epsilon^{10}(1-\epsilon)^{5} + 0.016581{15\choose 11}\epsilon^{11}(1-\epsilon)^{4} + 0.017818{15\choose 12}\epsilon^{12}(1-\epsilon)^{3}\\ & + 0.019095{15\choose 13}\epsilon^{13}(1-\epsilon)^{2} + 0.020446{15\choose 14}\epsilon^{14}(1-\epsilon)^{1} + 0.021907\epsilon^{15} 
\end{split}
\end{equation*}

\begin{equation*}
\begin{split}
  & h_{16}^{(F, OBB)}(\epsilon) = 0.253832(1-\epsilon)^{16} + 0.015865{16\choose 1}\epsilon^{1}(1-\epsilon)^{15} + 0.017781{16\choose 2}\epsilon^{2}(1-\epsilon)^{14}\\ & + 0.018617{16\choose 3}\epsilon^{3}(1-\epsilon)^{13} + 0.019369{16\choose 4}\epsilon^{4}(1-\epsilon)^{12} + 0.020048{16\choose 5}\epsilon^{5}(1-\epsilon)^{11}\\ & + 0.020663{16\choose 6}\epsilon^{6}(1-\epsilon)^{10} + 0.021217{16\choose 7}\epsilon^{7}(1-\epsilon)^{9} + 0.021716{16\choose 8}\epsilon^{8}(1-\epsilon)^{8}\\ & + 0.022162{16\choose 9}\epsilon^{9}(1-\epsilon)^{7} + 0.022559{16\choose 10}\epsilon^{10}(1-\epsilon)^{6} + 0.022907{16\choose 11}\epsilon^{11}(1-\epsilon)^{5}\\ & + 0.023205{16\choose 12}\epsilon^{12}(1-\epsilon)^{4} + 0.023450{16\choose 13}\epsilon^{13}(1-\epsilon)^{3} + 0.023633{16\choose 14}\epsilon^{14}(1-\epsilon)^{2}\\ & + 0.023738{16\choose 15}\epsilon^{15}(1-\epsilon)^{1} + 0.023738\epsilon^{16} \\ 
  & \bar{e}_{16}^{(F, OBB)}(\epsilon) = 0.000992{16\choose 1}\epsilon^{1}(1-\epsilon)^{15} + 0.003835{16\choose 2}\epsilon^{2}(1-\epsilon)^{14} + 0.005592{16\choose 3}\epsilon^{3}(1-\epsilon)^{13}\\ & + 0.007241{16\choose 4}\epsilon^{4}(1-\epsilon)^{12} + 0.008802{16\choose 5}\epsilon^{5}(1-\epsilon)^{11} + 0.010285{16\choose 6}\epsilon^{6}(1-\epsilon)^{10}\\ & + 0.011704{16\choose 7}\epsilon^{7}(1-\epsilon)^{9} + 0.013068{16\choose 8}\epsilon^{8}(1-\epsilon)^{8} + 0.014391{16\choose 9}\epsilon^{9}(1-\epsilon)^{7}\\ & + 0.015684{16\choose 10}\epsilon^{10}(1-\epsilon)^{6} + 0.016962{16\choose 11}\epsilon^{11}(1-\epsilon)^{5} + 0.018239{16\choose 12}\epsilon^{12}(1-\epsilon)^{4}\\ & + 0.019533{16\choose 13}\epsilon^{13}(1-\epsilon)^{3} + 0.020863{16\choose 14}\epsilon^{14}(1-\epsilon)^{2} + 0.022255{16\choose 15}\epsilon^{15}(1-\epsilon)^{1}\\ & + 0.023738\epsilon^{16} 
\end{split}
\end{equation*}

\begin{equation*}
\begin{split}
  & h_{3}^{(F, SBB)}(\epsilon) = \frac{1}{3}(1-\epsilon)^{3} + \frac{1}{9}{3\choose 1}\epsilon^{1}(1-\epsilon)^{2} + \frac{1}{9}{3\choose 2}\epsilon^{2}(1-\epsilon)^{1} + \frac{1}{3}\epsilon^{3} \\ 
  & \bar{e}_{3}^{(F, SBB)}(\epsilon) = \frac{1}{27}{3\choose 1}\epsilon^{1}(1-\epsilon)^{2} + \frac{2}{27}{3\choose 2}\epsilon^{2}(1-\epsilon)^{1} + \frac{1}{3}\epsilon^{3} 
\end{split}
\end{equation*}

\begin{equation*}
\begin{split}
  & h_{4}^{(F, SBB)}(\epsilon) = \frac{1}{4}(1-\epsilon)^{4} + \frac{1}{16}{4\choose 1}\epsilon^{1}(1-\epsilon)^{3} + \frac{1}{12}{4\choose 2}\epsilon^{2}(1-\epsilon)^{2} + \frac{1}{16}{4\choose 3}\epsilon^{3}(1-\epsilon)^{1} + \frac{1}{4}\epsilon^{4} \\ 
  & \bar{e}_{4}^{(F, SBB)}(\epsilon) = \frac{1}{64}{4\choose 1}\epsilon^{1}(1-\epsilon)^{3} + \frac{1}{24}{4\choose 2}\epsilon^{2}(1-\epsilon)^{2} + \frac{3}{64}{4\choose 3}\epsilon^{3}(1-\epsilon)^{1} + \frac{1}{4}\epsilon^{4} 
\end{split}
\end{equation*}

\begin{equation*}
\begin{split}
  & h_{5}^{(F, SBB)}(\epsilon) = 0.264000(1-\epsilon)^{5} + 0.052800{5\choose 1}\epsilon^{1}(1-\epsilon)^{4} + 0.059200{5\choose 2}\epsilon^{2}(1-\epsilon)^{3}\\ & + 0.059200{5\choose 3}\epsilon^{3}(1-\epsilon)^{2} + 0.052800{5\choose 4}\epsilon^{4}(1-\epsilon)^{1} + 0.264000\epsilon^{5} \\ 
  & \bar{e}_{5}^{(F, SBB)}(\epsilon) = 0.010560{5\choose 1}\epsilon^{1}(1-\epsilon)^{4} + 0.024960{5\choose 2}\epsilon^{2}(1-\epsilon)^{3} + 0.034240{5\choose 3}\epsilon^{3}(1-\epsilon)^{2}\\ & + 0.042240{5\choose 4}\epsilon^{4}(1-\epsilon)^{1} + 0.264000\epsilon^{5} 
\end{split}
\end{equation*}

\begin{equation*}
\begin{split}
  & h_{6}^{(F, SBB)}(\epsilon) = 0.259259(1-\epsilon)^{6} + 0.043210{6\choose 1}\epsilon^{1}(1-\epsilon)^{5} + 0.041152{6\choose 2}\epsilon^{2}(1-\epsilon)^{4}\\ & + 0.038272{6\choose 3}\epsilon^{3}(1-\epsilon)^{3} + 0.041152{6\choose 4}\epsilon^{4}(1-\epsilon)^{2} + 0.043210{6\choose 5}\epsilon^{5}(1-\epsilon)^{1}\\ & + 0.259259\epsilon^{6} \\ 
  & \bar{e}_{6}^{(F, SBB)}(\epsilon) = 0.007202{6\choose 1}\epsilon^{1}(1-\epsilon)^{5} + 0.017284{6\choose 2}\epsilon^{2}(1-\epsilon)^{4} + 0.019136{6\choose 3}\epsilon^{3}(1-\epsilon)^{3}\\ & + 0.023868{6\choose 4}\epsilon^{4}(1-\epsilon)^{2} + 0.036008{6\choose 5}\epsilon^{5}(1-\epsilon)^{1} + 0.259259\epsilon^{6} 
\end{split}
\end{equation*}

\begin{equation*}
\begin{split}
  & h_{7}^{(F, SBB)}(\epsilon) = 0.258523(1-\epsilon)^{7} + 0.036932{7\choose 1}\epsilon^{1}(1-\epsilon)^{6} + 0.042440{7\choose 2}\epsilon^{2}(1-\epsilon)^{5}\\ & + 0.043909{7\choose 3}\epsilon^{3}(1-\epsilon)^{4} + 0.043909{7\choose 4}\epsilon^{4}(1-\epsilon)^{3} + 0.042440{7\choose 5}\epsilon^{5}(1-\epsilon)^{2}\\ & + 0.036932{7\choose 6}\epsilon^{6}(1-\epsilon)^{1} + 0.258523\epsilon^{7} \\ 
  & \bar{e}_{7}^{(F, SBB)}(\epsilon) = 0.005276{7\choose 1}\epsilon^{1}(1-\epsilon)^{6} + 0.014952{7\choose 2}\epsilon^{2}(1-\epsilon)^{5} + 0.020046{7\choose 3}\epsilon^{3}(1-\epsilon)^{4}\\ & + 0.023862{7\choose 4}\epsilon^{4}(1-\epsilon)^{3} + 0.027487{7\choose 5}\epsilon^{5}(1-\epsilon)^{2} + 0.031656{7\choose 6}\epsilon^{6}(1-\epsilon)^{1}\\ & + 0.258523\epsilon^{7} 
\end{split}
\end{equation*}

\begin{equation*}
\begin{split}
  & h_{8}^{(F, SBB)}(\epsilon) = 0.257446(1-\epsilon)^{8} + 0.032181{8\choose 1}\epsilon^{1}(1-\epsilon)^{7} + 0.042454{8\choose 2}\epsilon^{2}(1-\epsilon)^{6}\\ & + 0.038256{8\choose 3}\epsilon^{3}(1-\epsilon)^{5} + 0.044332{8\choose 4}\epsilon^{4}(1-\epsilon)^{4} + 0.038256{8\choose 5}\epsilon^{5}(1-\epsilon)^{3}\\ & + 0.042454{8\choose 6}\epsilon^{6}(1-\epsilon)^{2} + 0.032181{8\choose 7}\epsilon^{7}(1-\epsilon)^{1} + 0.257446\epsilon^{8} \\ 
  & \bar{e}_{8}^{(F, SBB)}(\epsilon) = 0.004023{8\choose 1}\epsilon^{1}(1-\epsilon)^{7} + 0.013932{8\choose 2}\epsilon^{2}(1-\epsilon)^{6} + 0.016147{8\choose 3}\epsilon^{3}(1-\epsilon)^{5}\\ & + 0.022166{8\choose 4}\epsilon^{4}(1-\epsilon)^{4} + 0.022109{8\choose 5}\epsilon^{5}(1-\epsilon)^{3} + 0.028522{8\choose 6}\epsilon^{6}(1-\epsilon)^{2}\\ & + 0.028158{8\choose 7}\epsilon^{7}(1-\epsilon)^{1} + 0.257446\epsilon^{8} 
\end{split}
\end{equation*}

\begin{equation*}
\begin{split}
  & h_{9}^{(F, SBB)}(\epsilon) = 0.256678(1-\epsilon)^{9} + 0.028520{9\choose 1}\epsilon^{1}(1-\epsilon)^{8} + 0.032563{9\choose 2}\epsilon^{2}(1-\epsilon)^{7}\\ & + 0.036335{9\choose 3}\epsilon^{3}(1-\epsilon)^{6} + 0.034293{9\choose 4}\epsilon^{4}(1-\epsilon)^{5} + 0.034293{9\choose 5}\epsilon^{5}(1-\epsilon)^{4}\\ & + 0.036335{9\choose 6}\epsilon^{6}(1-\epsilon)^{3} + 0.032563{9\choose 7}\epsilon^{7}(1-\epsilon)^{2} + 0.028520{9\choose 8}\epsilon^{8}(1-\epsilon)^{1}\\ & + 0.256678\epsilon^{9} \\ 
  & \bar{e}_{9}^{(F, SBB)}(\epsilon) = 0.003169{9\choose 1}\epsilon^{1}(1-\epsilon)^{8} + 0.009810{9\choose 2}\epsilon^{2}(1-\epsilon)^{7} + 0.014357{9\choose 3}\epsilon^{3}(1-\epsilon)^{6}\\ & + 0.015994{9\choose 4}\epsilon^{4}(1-\epsilon)^{5} + 0.018299{9\choose 5}\epsilon^{5}(1-\epsilon)^{4} + 0.021978{9\choose 6}\epsilon^{6}(1-\epsilon)^{3}\\ & + 0.022753{9\choose 7}\epsilon^{7}(1-\epsilon)^{2} + 0.025351{9\choose 8}\epsilon^{8}(1-\epsilon)^{1} + 0.256678\epsilon^{9} 
\end{split}
\end{equation*}

\begin{equation*}
\begin{split}
  & h_{10}^{(F, SBB)}(\epsilon) = 0.256038(1-\epsilon)^{10} + 0.025604{10\choose 1}\epsilon^{1}(1-\epsilon)^{9} + 0.027144{10\choose 2}\epsilon^{2}(1-\epsilon)^{8}\\ & + 0.024597{10\choose 3}\epsilon^{3}(1-\epsilon)^{7} + 0.025946{10\choose 4}\epsilon^{4}(1-\epsilon)^{6} + 0.025583{10\choose 5}\epsilon^{5}(1-\epsilon)^{5}\\ & + 0.025946{10\choose 6}\epsilon^{6}(1-\epsilon)^{4} + 0.024597{10\choose 7}\epsilon^{7}(1-\epsilon)^{3} + 0.027144{10\choose 8}\epsilon^{8}(1-\epsilon)^{2}\\ & + 0.025604{10\choose 9}\epsilon^{9}(1-\epsilon)^{1} + 0.256038\epsilon^{10} \\ 
  & \bar{e}_{10}^{(F, SBB)}(\epsilon) = 0.002560{10\choose 1}\epsilon^{1}(1-\epsilon)^{9} + 0.007627{10\choose 2}\epsilon^{2}(1-\epsilon)^{8} + 0.009105{10\choose 3}\epsilon^{3}(1-\epsilon)^{7}\\ & + 0.011405{10\choose 4}\epsilon^{4}(1-\epsilon)^{6} + 0.012791{10\choose 5}\epsilon^{5}(1-\epsilon)^{5} + 0.014541{10\choose 6}\epsilon^{6}(1-\epsilon)^{4}\\ & + 0.015493{10\choose 7}\epsilon^{7}(1-\epsilon)^{3} + 0.019517{10\choose 8}\epsilon^{8}(1-\epsilon)^{2} + 0.023043{10\choose 9}\epsilon^{9}(1-\epsilon)^{1}\\ & + 0.256038\epsilon^{10} 
\end{split}
\end{equation*}

\begin{equation*}
\begin{split}
  & h_{11}^{(F, SBB)}(\epsilon) = 0.255512(1-\epsilon)^{11} + 0.023228{11\choose 1}\epsilon^{1}(1-\epsilon)^{10} + 0.026221{11\choose 2}\epsilon^{2}(1-\epsilon)^{9}\\ & + 0.027262{11\choose 3}\epsilon^{3}(1-\epsilon)^{8} + 0.027760{11\choose 4}\epsilon^{4}(1-\epsilon)^{7} + 0.027974{11\choose 5}\epsilon^{5}(1-\epsilon)^{6}\\ & + 0.027974{11\choose 6}\epsilon^{6}(1-\epsilon)^{5} + 0.027760{11\choose 7}\epsilon^{7}(1-\epsilon)^{4} + 0.027262{11\choose 8}\epsilon^{8}(1-\epsilon)^{3}\\ & + 0.026221{11\choose 9}\epsilon^{9}(1-\epsilon)^{2} + 0.023228{11\choose 10}\epsilon^{10}(1-\epsilon)^{1} + 0.255512\epsilon^{11} \\ 
  & \bar{e}_{11}^{(F, SBB)}(\epsilon) = 0.002112{11\choose 1}\epsilon^{1}(1-\epsilon)^{10} + 0.006888{11\choose 2}\epsilon^{2}(1-\epsilon)^{9} + 0.009509{11\choose 3}\epsilon^{3}(1-\epsilon)^{8}\\ & + 0.011509{11\choose 4}\epsilon^{4}(1-\epsilon)^{7} + 0.013212{11\choose 5}\epsilon^{5}(1-\epsilon)^{6} + 0.014762{11\choose 6}\epsilon^{6}(1-\epsilon)^{5}\\ & + 0.016252{11\choose 7}\epsilon^{7}(1-\epsilon)^{4} + 0.017753{11\choose 8}\epsilon^{8}(1-\epsilon)^{3} + 0.019333{11\choose 9}\epsilon^{9}(1-\epsilon)^{2}\\ & + 0.021117{11\choose 10}\epsilon^{10}(1-\epsilon)^{1} + 0.255512\epsilon^{11} 
\end{split}
\end{equation*}

\begin{equation*}
\begin{split}
  & h_{12}^{(F, SBB)}(\epsilon) = 0.255068(1-\epsilon)^{12} + 0.021256{12\choose 1}\epsilon^{1}(1-\epsilon)^{11} + 0.018466{12\choose 2}\epsilon^{2}(1-\epsilon)^{10}\\ & + 0.016542{12\choose 3}\epsilon^{3}(1-\epsilon)^{9} + 0.016350{12\choose 4}\epsilon^{4}(1-\epsilon)^{8} + 0.015541{12\choose 5}\epsilon^{5}(1-\epsilon)^{7}\\ & + 0.016128{12\choose 6}\epsilon^{6}(1-\epsilon)^{6} + 0.015541{12\choose 7}\epsilon^{7}(1-\epsilon)^{5} + 0.016350{12\choose 8}\epsilon^{8}(1-\epsilon)^{4}\\ & + 0.016542{12\choose 9}\epsilon^{9}(1-\epsilon)^{3} + 0.018466{12\choose 10}\epsilon^{10}(1-\epsilon)^{2} + 0.021256{12\choose 11}\epsilon^{11}(1-\epsilon)^{1}\\ & + 0.255068\epsilon^{12} \\ 
  & \bar{e}_{12}^{(F, SBB)}(\epsilon) = 0.001771{12\choose 1}\epsilon^{1}(1-\epsilon)^{11} + 0.004505{12\choose 2}\epsilon^{2}(1-\epsilon)^{10} + 0.005473{12\choose 3}\epsilon^{3}(1-\epsilon)^{9}\\ & + 0.006457{12\choose 4}\epsilon^{4}(1-\epsilon)^{8} + 0.006983{12\choose 5}\epsilon^{5}(1-\epsilon)^{7} + 0.008064{12\choose 6}\epsilon^{6}(1-\epsilon)^{6}\\ & + 0.008558{12\choose 7}\epsilon^{7}(1-\epsilon)^{5} + 0.009893{12\choose 8}\epsilon^{8}(1-\epsilon)^{4} + 0.011068{12\choose 9}\epsilon^{9}(1-\epsilon)^{3}\\ & + 0.013961{12\choose 10}\epsilon^{10}(1-\epsilon)^{2} + 0.019484{12\choose 11}\epsilon^{11}(1-\epsilon)^{1} + 0.255068\epsilon^{12} 
\end{split}
\end{equation*}

\begin{equation*}
\begin{split}
  & h_{13}^{(F, SBB)}(\epsilon) = 0.254691(1-\epsilon)^{13} + 0.019592{13\choose 1}\epsilon^{1}(1-\epsilon)^{12} + 0.021871{13\choose 2}\epsilon^{2}(1-\epsilon)^{11}\\ & + 0.022711{13\choose 3}\epsilon^{3}(1-\epsilon)^{10} + 0.023175{13\choose 4}\epsilon^{4}(1-\epsilon)^{9} + 0.023439{13\choose 5}\epsilon^{5}(1-\epsilon)^{8}\\ & + 0.023561{13\choose 6}\epsilon^{6}(1-\epsilon)^{7} + 0.023561{13\choose 7}\epsilon^{7}(1-\epsilon)^{6} + 0.023439{13\choose 8}\epsilon^{8}(1-\epsilon)^{5}\\ & + 0.023175{13\choose 9}\epsilon^{9}(1-\epsilon)^{4} + 0.022711{13\choose 10}\epsilon^{10}(1-\epsilon)^{3} + 0.021871{13\choose 11}\epsilon^{11}(1-\epsilon)^{2}\\ & + 0.019592{13\choose 12}\epsilon^{12}(1-\epsilon)^{1} + 0.254691\epsilon^{13} \\ 
  & \bar{e}_{13}^{(F, SBB)}(\epsilon) = 0.001507{13\choose 1}\epsilon^{1}(1-\epsilon)^{12} + 0.005089{13\choose 2}\epsilon^{2}(1-\epsilon)^{11} + 0.007104{13\choose 3}\epsilon^{3}(1-\epsilon)^{10}\\ & + 0.008677{13\choose 4}\epsilon^{4}(1-\epsilon)^{9} + 0.010019{13\choose 5}\epsilon^{5}(1-\epsilon)^{8} + 0.011221{13\choose 6}\epsilon^{6}(1-\epsilon)^{7}\\ & + 0.012340{13\choose 7}\epsilon^{7}(1-\epsilon)^{6} + 0.013420{13\choose 8}\epsilon^{8}(1-\epsilon)^{5} + 0.014498{13\choose 9}\epsilon^{9}(1-\epsilon)^{4}\\ & + 0.015608{13\choose 10}\epsilon^{10}(1-\epsilon)^{3} + 0.016782{13\choose 11}\epsilon^{11}(1-\epsilon)^{2} + 0.018085{13\choose 12}\epsilon^{12}(1-\epsilon)^{1}\\ & + 0.254691\epsilon^{13} 
\end{split}
\end{equation*}

\begin{equation*}
\begin{split}
  & h_{14}^{(F, SBB)}(\epsilon) = 0.254365(1-\epsilon)^{14} + 0.018169{14\choose 1}\epsilon^{1}(1-\epsilon)^{13} + 0.021090{14\choose 2}\epsilon^{2}(1-\epsilon)^{12}\\ & + 0.020198{14\choose 3}\epsilon^{3}(1-\epsilon)^{11} + 0.021062{14\choose 4}\epsilon^{4}(1-\epsilon)^{10} + 0.020892{14\choose 5}\epsilon^{5}(1-\epsilon)^{9}\\ & + 0.021293{14\choose 6}\epsilon^{6}(1-\epsilon)^{8} + 0.021166{14\choose 7}\epsilon^{7}(1-\epsilon)^{7} + 0.021293{14\choose 8}\epsilon^{8}(1-\epsilon)^{6}\\ & + 0.020892{14\choose 9}\epsilon^{9}(1-\epsilon)^{5} + 0.021062{14\choose 10}\epsilon^{10}(1-\epsilon)^{4} + 0.020198{14\choose 11}\epsilon^{11}(1-\epsilon)^{3}\\ & + 0.021090{14\choose 12}\epsilon^{12}(1-\epsilon)^{2} + 0.018169{14\choose 13}\epsilon^{13}(1-\epsilon)^{1} + 0.254365\epsilon^{14} \\ 
  & \bar{e}_{14}^{(F, SBB)}(\epsilon) = 0.001298{14\choose 1}\epsilon^{1}(1-\epsilon)^{13} + 0.004640{14\choose 2}\epsilon^{2}(1-\epsilon)^{12} + 0.005990{14\choose 3}\epsilon^{3}(1-\epsilon)^{11}\\ & + 0.007519{14\choose 4}\epsilon^{4}(1-\epsilon)^{10} + 0.008532{14\choose 5}\epsilon^{5}(1-\epsilon)^{9} + 0.009695{14\choose 6}\epsilon^{6}(1-\epsilon)^{8}\\ & + 0.010583{14\choose 7}\epsilon^{7}(1-\epsilon)^{7} + 0.011598{14\choose 8}\epsilon^{8}(1-\epsilon)^{6} + 0.012359{14\choose 9}\epsilon^{9}(1-\epsilon)^{5}\\ & + 0.013543{14\choose 10}\epsilon^{10}(1-\epsilon)^{4} + 0.014208{14\choose 11}\epsilon^{11}(1-\epsilon)^{3} + 0.016449{14\choose 12}\epsilon^{12}(1-\epsilon)^{2}\\ & + 0.016871{14\choose 13}\epsilon^{13}(1-\epsilon)^{1} + 0.254365\epsilon^{14} 
\end{split}
\end{equation*}

\begin{equation*}
\begin{split}
  & h_{15}^{(F, SBB)}(\epsilon) = 0.254081(1-\epsilon)^{15} + 0.016939{15\choose 1}\epsilon^{1}(1-\epsilon)^{14} + 0.017055{15\choose 2}\epsilon^{2}(1-\epsilon)^{13}\\ & + 0.017717{15\choose 3}\epsilon^{3}(1-\epsilon)^{12} + 0.017592{15\choose 4}\epsilon^{4}(1-\epsilon)^{11} + 0.017847{15\choose 5}\epsilon^{5}(1-\epsilon)^{10}\\ & + 0.017963{15\choose 6}\epsilon^{6}(1-\epsilon)^{9} + 0.017946{15\choose 7}\epsilon^{7}(1-\epsilon)^{8} + 0.017946{15\choose 8}\epsilon^{8}(1-\epsilon)^{7}\\ & + 0.017963{15\choose 9}\epsilon^{9}(1-\epsilon)^{6} + 0.017847{15\choose 10}\epsilon^{10}(1-\epsilon)^{5} + 0.017592{15\choose 11}\epsilon^{11}(1-\epsilon)^{4}\\ & + 0.017717{15\choose 12}\epsilon^{12}(1-\epsilon)^{3} + 0.017055{15\choose 13}\epsilon^{13}(1-\epsilon)^{2} + 0.016939{15\choose 14}\epsilon^{14}(1-\epsilon)^{1}\\ & + 0.254081\epsilon^{15} \\ 
  & \bar{e}_{15}^{(F, SBB)}(\epsilon) = 0.001129{15\choose 1}\epsilon^{1}(1-\epsilon)^{14} + 0.003541{15\choose 2}\epsilon^{2}(1-\epsilon)^{13} + 0.005027{15\choose 3}\epsilon^{3}(1-\epsilon)^{12}\\ & + 0.006017{15\choose 4}\epsilon^{4}(1-\epsilon)^{11} + 0.006997{15\choose 5}\epsilon^{5}(1-\epsilon)^{10} + 0.007850{15\choose 6}\epsilon^{6}(1-\epsilon)^{9}\\ & + 0.008601{15\choose 7}\epsilon^{7}(1-\epsilon)^{8} + 0.009345{15\choose 8}\epsilon^{8}(1-\epsilon)^{7} + 0.010113{15\choose 9}\epsilon^{9}(1-\epsilon)^{6}\\ & + 0.010850{15\choose 10}\epsilon^{10}(1-\epsilon)^{5} + 0.011574{15\choose 11}\epsilon^{11}(1-\epsilon)^{4} + 0.012691{15\choose 12}\epsilon^{12}(1-\epsilon)^{3}\\ & + 0.013513{15\choose 13}\epsilon^{13}(1-\epsilon)^{2} + 0.015810{15\choose 14}\epsilon^{14}(1-\epsilon)^{1} + 0.254081\epsilon^{15} 
\end{split}
\end{equation*}

\begin{equation*}
\begin{split}
  & h_{16}^{(F, SBB)}(\epsilon) = 0.253832(1-\epsilon)^{16} + 0.015865{16\choose 1}\epsilon^{1}(1-\epsilon)^{15} + 0.018639{16\choose 2}\epsilon^{2}(1-\epsilon)^{14}\\ & + 0.018075{16\choose 3}\epsilon^{3}(1-\epsilon)^{13} + 0.018841{16\choose 4}\epsilon^{4}(1-\epsilon)^{12} + 0.018735{16\choose 5}\epsilon^{5}(1-\epsilon)^{11}\\ & + 0.019043{16\choose 6}\epsilon^{6}(1-\epsilon)^{10} + 0.019007{16\choose 7}\epsilon^{7}(1-\epsilon)^{9} + 0.019174{16\choose 8}\epsilon^{8}(1-\epsilon)^{8}\\ & + 0.019007{16\choose 9}\epsilon^{9}(1-\epsilon)^{7} + 0.019043{16\choose 10}\epsilon^{10}(1-\epsilon)^{6} + 0.018735{16\choose 11}\epsilon^{11}(1-\epsilon)^{5}\\ & + 0.018841{16\choose 12}\epsilon^{12}(1-\epsilon)^{4} + 0.018075{16\choose 13}\epsilon^{13}(1-\epsilon)^{3} + 0.018639{16\choose 14}\epsilon^{14}(1-\epsilon)^{2}\\ & + 0.015865{16\choose 15}\epsilon^{15}(1-\epsilon)^{1} + 0.253832\epsilon^{16} \\ 
  & \bar{e}_{16}^{(F, SBB)}(\epsilon) = 0.000992{16\choose 1}\epsilon^{1}(1-\epsilon)^{15} + 0.003711{16\choose 2}\epsilon^{2}(1-\epsilon)^{14} + 0.004890{16\choose 3}\epsilon^{3}(1-\epsilon)^{13}\\ & + 0.006173{16\choose 4}\epsilon^{4}(1-\epsilon)^{12} + 0.007048{16\choose 5}\epsilon^{5}(1-\epsilon)^{11} + 0.007997{16\choose 6}\epsilon^{6}(1-\epsilon)^{10}\\ & + 0.008756{16\choose 7}\epsilon^{7}(1-\epsilon)^{9} + 0.009587{16\choose 8}\epsilon^{8}(1-\epsilon)^{8} + 0.010252{16\choose 9}\epsilon^{9}(1-\epsilon)^{7}\\ & + 0.011046{16\choose 10}\epsilon^{10}(1-\epsilon)^{6} + 0.011687{16\choose 11}\epsilon^{11}(1-\epsilon)^{5} + 0.012668{16\choose 12}\epsilon^{12}(1-\epsilon)^{4}\\ & + 0.013185{16\choose 13}\epsilon^{13}(1-\epsilon)^{3} + 0.014928{16\choose 14}\epsilon^{14}(1-\epsilon)^{2} + 0.014873{16\choose 15}\epsilon^{15}(1-\epsilon)^{1}\\ & + 0.253832\epsilon^{16} 
\end{split}
\end{equation*}

\begin{equation*}
\begin{split}
  & h_{4}^{(H, OBB)}(\epsilon) = \frac{1}{4}(1-\epsilon)^{4} + \frac{1}{16}{4\choose 1}\epsilon^{1}(1-\epsilon)^{3} + \frac{1}{16}{4\choose 2}\epsilon^{2}(1-\epsilon)^{2} + \frac{3}{32}{4\choose 3}\epsilon^{3}(1-\epsilon)^{1} + \frac{3}{32}\epsilon^{4} \\ 
  & \bar{e}_{4}^{(H, OBB)}(\epsilon) = \frac{1}{64}{4\choose 1}\epsilon^{1}(1-\epsilon)^{3} + \frac{1}{32}{4\choose 2}\epsilon^{2}(1-\epsilon)^{2} + \frac{9}{128}{4\choose 3}\epsilon^{3}(1-\epsilon)^{1} + \frac{3}{32}\epsilon^{4} 
\end{split}
\end{equation*}

\begin{equation*}
\begin{split}
  & h_{8}^{(H, OBB)}(\epsilon) = 0.257446(1-\epsilon)^{8} + 0.032181{8\choose 1}\epsilon^{1}(1-\epsilon)^{7} + 0.037857{8\choose 2}\epsilon^{2}(1-\epsilon)^{6}\\ & + 0.042412{8\choose 3}\epsilon^{3}(1-\epsilon)^{5} + 0.044759{8\choose 4}\epsilon^{4}(1-\epsilon)^{4} + 0.046648{8\choose 5}\epsilon^{5}(1-\epsilon)^{3}\\ & + 0.048082{8\choose 6}\epsilon^{6}(1-\epsilon)^{2} + 0.049087{8\choose 7}\epsilon^{7}(1-\epsilon)^{1} + 0.049087\epsilon^{8} \\ 
  & \bar{e}_{8}^{(H, OBB)}(\epsilon) = 0.004023{8\choose 1}\epsilon^{1}(1-\epsilon)^{7} + 0.012783{8\choose 2}\epsilon^{2}(1-\epsilon)^{6} + 0.021266{8\choose 3}\epsilon^{3}(1-\epsilon)^{5}\\ & + 0.027012{8\choose 4}\epsilon^{4}(1-\epsilon)^{4} + 0.032374{8\choose 5}\epsilon^{5}(1-\epsilon)^{3} + 0.037564{8\choose 6}\epsilon^{6}(1-\epsilon)^{2}\\ & + 0.042951{8\choose 7}\epsilon^{7}(1-\epsilon)^{1} + 0.049087\epsilon^{8} 
\end{split}
\end{equation*}

\begin{equation*}
\begin{split}
  & h_{16}^{(H, OBB)}(\epsilon) = 0.253832(1-\epsilon)^{16} + 0.015865{16\choose 1}\epsilon^{1}(1-\epsilon)^{15} + 0.017582{16\choose 2}\epsilon^{2}(1-\epsilon)^{14}\\ & + 0.018614{16\choose 3}\epsilon^{3}(1-\epsilon)^{13} + 0.019368{16\choose 4}\epsilon^{4}(1-\epsilon)^{12} + 0.020048{16\choose 5}\epsilon^{5}(1-\epsilon)^{11}\\ & + 0.020663{16\choose 6}\epsilon^{6}(1-\epsilon)^{10} + 0.021217{16\choose 7}\epsilon^{7}(1-\epsilon)^{9} + 0.021716{16\choose 8}\epsilon^{8}(1-\epsilon)^{8}\\ & + 0.022162{16\choose 9}\epsilon^{9}(1-\epsilon)^{7} + 0.022559{16\choose 10}\epsilon^{10}(1-\epsilon)^{6} + 0.022907{16\choose 11}\epsilon^{11}(1-\epsilon)^{5}\\ & + 0.023205{16\choose 12}\epsilon^{12}(1-\epsilon)^{4} + 0.023450{16\choose 13}\epsilon^{13}(1-\epsilon)^{3} + 0.023633{16\choose 14}\epsilon^{14}(1-\epsilon)^{2}\\ & + 0.023738{16\choose 15}\epsilon^{15}(1-\epsilon)^{1} + 0.023738\epsilon^{16} \\ 
  & \bar{e}_{16}^{(H, OBB)}(\epsilon) = 0.000992{16\choose 1}\epsilon^{1}(1-\epsilon)^{15} + 0.003579{16\choose 2}\epsilon^{2}(1-\epsilon)^{14} + 0.005590{16\choose 3}\epsilon^{3}(1-\epsilon)^{13}\\ & + 0.007241{16\choose 4}\epsilon^{4}(1-\epsilon)^{12} + 0.008802{16\choose 5}\epsilon^{5}(1-\epsilon)^{11} + 0.010285{16\choose 6}\epsilon^{6}(1-\epsilon)^{10}\\ & + 0.011704{16\choose 7}\epsilon^{7}(1-\epsilon)^{9} + 0.013068{16\choose 8}\epsilon^{8}(1-\epsilon)^{8} + 0.014391{16\choose 9}\epsilon^{9}(1-\epsilon)^{7}\\ & + 0.015684{16\choose 10}\epsilon^{10}(1-\epsilon)^{6} + 0.016962{16\choose 11}\epsilon^{11}(1-\epsilon)^{5} + 0.018239{16\choose 12}\epsilon^{12}(1-\epsilon)^{4}\\ & + 0.019533{16\choose 13}\epsilon^{13}(1-\epsilon)^{3} + 0.020863{16\choose 14}\epsilon^{14}(1-\epsilon)^{2} + 0.022255{16\choose 15}\epsilon^{15}(1-\epsilon)^{1}\\ & + 0.023738\epsilon^{16} 
\end{split}
\end{equation*}

\begin{equation*}
\begin{split}
  & h_{4}^{(H, SBB)}(\epsilon) = \frac{1}{4}(1-\epsilon)^{4} + \frac{1}{16}{4\choose 1}\epsilon^{1}(1-\epsilon)^{3} + \frac{1}{8}{4\choose 2}\epsilon^{2}(1-\epsilon)^{2} + \frac{1}{16}{4\choose 3}\epsilon^{3}(1-\epsilon)^{1} + \frac{1}{4}\epsilon^{4} \\ 
  & \bar{e}_{4}^{(H, SBB)}(\epsilon) = \frac{1}{64}{4\choose 1}\epsilon^{1}(1-\epsilon)^{3} + \frac{1}{16}{4\choose 2}\epsilon^{2}(1-\epsilon)^{2} + \frac{3}{64}{4\choose 3}\epsilon^{3}(1-\epsilon)^{1} + \frac{1}{4}\epsilon^{4} 
\end{split}
\end{equation*}

\begin{equation*}
\begin{split}
  & h_{8}^{(H, SBB)}(\epsilon) = 0.257446(1-\epsilon)^{8} + 0.032181{8\choose 1}\epsilon^{1}(1-\epsilon)^{7} + 0.075714{8\choose 2}\epsilon^{2}(1-\epsilon)^{6}\\ & + 0.037857{8\choose 3}\epsilon^{3}(1-\epsilon)^{5} + 0.062317{8\choose 4}\epsilon^{4}(1-\epsilon)^{4} + 0.037857{8\choose 5}\epsilon^{5}(1-\epsilon)^{3}\\ & + 0.075714{8\choose 6}\epsilon^{6}(1-\epsilon)^{2} + 0.032181{8\choose 7}\epsilon^{7}(1-\epsilon)^{1} + 0.257446\epsilon^{8} \\ 
  & \bar{e}_{8}^{(H, SBB)}(\epsilon) = 0.004023{8\choose 1}\epsilon^{1}(1-\epsilon)^{7} + 0.025566{8\choose 2}\epsilon^{2}(1-\epsilon)^{6} + 0.015856{8\choose 3}\epsilon^{3}(1-\epsilon)^{5}\\ & + 0.031158{8\choose 4}\epsilon^{4}(1-\epsilon)^{4} + 0.022001{8\choose 5}\epsilon^{5}(1-\epsilon)^{3} + 0.050148{8\choose 6}\epsilon^{6}(1-\epsilon)^{2}\\ & + 0.028158{8\choose 7}\epsilon^{7}(1-\epsilon)^{1} + 0.257446\epsilon^{8} 
\end{split}
\end{equation*}

\begin{equation*}
\begin{split}
  & h_{16}^{(H, SBB)}(\epsilon) = 0.253832(1-\epsilon)^{16} + 0.015865{16\choose 1}\epsilon^{1}(1-\epsilon)^{15} + 0.035163{16\choose 2}\epsilon^{2}(1-\epsilon)^{14}\\ & + 0.017966{16\choose 3}\epsilon^{3}(1-\epsilon)^{13} + 0.022741{16\choose 4}\epsilon^{4}(1-\epsilon)^{12} + 0.018697{16\choose 5}\epsilon^{5}(1-\epsilon)^{11}\\ & + 0.020894{16\choose 6}\epsilon^{6}(1-\epsilon)^{10} + 0.018984{16\choose 7}\epsilon^{7}(1-\epsilon)^{9} + 0.020593{16\choose 8}\epsilon^{8}(1-\epsilon)^{8}\\ & + 0.018984{16\choose 9}\epsilon^{9}(1-\epsilon)^{7} + 0.020894{16\choose 10}\epsilon^{10}(1-\epsilon)^{6} + 0.018697{16\choose 11}\epsilon^{11}(1-\epsilon)^{5}\\ & + 0.022741{16\choose 12}\epsilon^{12}(1-\epsilon)^{4} + 0.017966{16\choose 13}\epsilon^{13}(1-\epsilon)^{3} + 0.035163{16\choose 14}\epsilon^{14}(1-\epsilon)^{2}\\ & + 0.015865{16\choose 15}\epsilon^{15}(1-\epsilon)^{1} + 0.253832\epsilon^{16} \\ 
  & \bar{e}_{16}^{(H, SBB)}(\epsilon) = 0.000992{16\choose 1}\epsilon^{1}(1-\epsilon)^{15} + 0.007158{16\choose 2}\epsilon^{2}(1-\epsilon)^{14} + 0.004791{16\choose 3}\epsilon^{3}(1-\epsilon)^{13}\\ & + 0.007453{16\choose 4}\epsilon^{4}(1-\epsilon)^{12} + 0.007019{16\choose 5}\epsilon^{5}(1-\epsilon)^{11} + 0.008774{16\choose 6}\epsilon^{6}(1-\epsilon)^{10}\\ & + 0.008742{16\choose 7}\epsilon^{7}(1-\epsilon)^{9} + 0.010296{16\choose 8}\epsilon^{8}(1-\epsilon)^{8} + 0.010242{16\choose 9}\epsilon^{9}(1-\epsilon)^{7}\\ & + 0.012120{16\choose 10}\epsilon^{10}(1-\epsilon)^{6} + 0.011678{16\choose 11}\epsilon^{11}(1-\epsilon)^{5} + 0.015289{16\choose 12}\epsilon^{12}(1-\epsilon)^{4}\\ & + 0.013174{16\choose 13}\epsilon^{13}(1-\epsilon)^{3} + 0.028005{16\choose 14}\epsilon^{14}(1-\epsilon)^{2} + 0.014873{16\choose 15}\epsilon^{15}(1-\epsilon)^{1}\\ & + 0.253832\epsilon^{16} 
\end{split}
\end{equation*}

\end{document}